\documentclass[sigplan,screen]{acmart}
\usepackage[utf8]{inputenc}
\usepackage[T1]{fontenc}

\usepackage{hyperref}
  \usepackage{lipsum}

  \newcounter{mycounter}

\usepackage{amsmath,mathrsfs}

\usepackage[bbgreekl]{mathbbol}
\usepackage{xspace}
\usepackage{wrapfig}
\usepackage{dsfont}
\usepackage{chngcntr}
\usepackage{float}
\usepackage[all]{xy}
\SelectTips {cm}{}
\newdir{ >}{{}*!/-5pt/@{>}}
\newdir{ (}{{}*!/-5pt/@{(}}

\DeclareFontFamily{U}{bbold}{}\DeclareFontShape{U}{bbold}{m}{n}{%
  <4.25>bbold5<5>bbold5<6>bbold6<7>bbold7<8>bbold8<8.5>bbold9<9.5>bbold10<10>bbold10<12>bbold12}{}

\usepackage[notcite,notref,final]{showkeys}

\hypersetup{
    colorlinks=true,
    linkcolor=black,
    citecolor=black,
    filecolor=black,
    urlcolor=black,
    pdftitle={Varieties of Data Languages},
    pdfauthor={Stefan Milius, Henning urbat},
    pdfkeywords={Nominal sets, Stone duality, Algebraic language
      theory, Data languages},
    pdfduplex={DuplexFlipLongEdge},
}

\DeclareMathSymbol{\mathinvertedexclamationmark}{\mathord}{operators}{'074}
\DeclareMathSymbol{\mathexclamationmark}{\mathord}{operators}{'041}
\makeatletter
\newcommand{\raisedmathinvertedexclamationmark}{%
  \mathord{\mathpalette\raised@mathinvertedexclamationmark\relax}%
}
\newcommand{\raised@mathinvertedexclamationmark}[2]{%
  \raisebox{\depth}{$\m@th#1\mathinvertedexclamationmark$}%
}
\makeatother

\newcommand{\terminal}{\mathexclamationmark}
\newcommand{\initial}{\raisedmathinvertedexclamationmark}

\usepackage{seqsplit}
\usepackage{xstring}
\usepackage{xcolor}

\makeatletter
\renewcommand*\showkeyslabelformat[1]{%
\@ifundefined{hideNextShowKeysLabel}{%
\noexpandarg%
\StrSubstitute{#1}{ }{\textvisiblespace}[\TEMP]%
\parbox[t]{\marginparwidth}{\raggedright\normalfont\small\ttfamily\(\{\){\color{red!50!black}\expandafter\seqsplit\expandafter{\TEMP}}\(\}\)}%
}{}
}
\makeatother
\AtEndPreamble{
\theoremstyle{definition}
\newtheorem{remark}[theorem]{Remark}
\newtheorem{assumptions}[theorem]{Assumptions}
\newtheorem{notation}[theorem]{Notation}
}

\usepackage[inline]{enumitem}
\setlist[enumerate,1]{label=(\arabic*),font=\normalfont,align=left,leftmargin=0pt,labelindent=0pt,listparindent=\parindent,labelwidth=0pt,itemindent=!,topsep=0pt,parsep=0pt,itemsep=0pt,start=1}
\setlist[enumerate,2]{label=(\alph*),font=\normalfont,labelindent=*,topsep=0pt,leftmargin=*,start=1}
\setlist[itemize]{labelindent=*,leftmargin=*,topsep=5pt,itemsep=3pt}
\setlist[description]{labelindent=*,leftmargin=*,itemindent=-1 em}

\def\D{\DCat}

\def\epsilon{\varepsilon}

\renewcommand{\rho}{\varrho}
\def\ol{\overline}
\def\o{\cdot}

\newcommand{\lin}[1]{\mathsf{lin}(#1)}

\renewcommand{\vec}[1]{{\overset{{}_{\rightarrow}}{#1}} }
\newcommand{\vecru}[1]{{ \underset{\overset{\leftarrow}{}}{#1} }}
\newcommand{\vecr}[1]{{\overset{{}_{\leftarrow}}{#1} }}

\newcommand{\inl}{\mathsf{inl}}
\newcommand{\outl}{\mathsf{outl}}
\newcommand{\outr}{\mathsf{outr}}

\newcommand{\inr}{\mathsf{inr}}

\renewcommand{\t}{\otimes}
\newcommand{\takeout}[1]{\empty}

\newcommand{\alphaeq}{\equiv_\alpha}

\newcommand{\supp}{\mathop{\mathsf{supp}}}

\newcommand{\At}{\mathbb{A}}
\newcommand{\Perm}{\mathrm{Perm}}
\newcommand{\Nom}{\mathbf{Nom}}

\renewcommand{\phi}{\varphi}

\newcommand{\xra}{\xrightarrow}

\newcommand{\Lra}{\Leftrightarrow}
\newcommand{\seq}{\subseteq}

\newcommand{\Min}[1]{\mathsf{Min}(#1)}

\newcommand{\Alg}{\mathbf{Alg}\,}

\newcommand{\Set}{\mathbf{Set}}
\newcommand{\E}{\mathcal{E}}
\newcommand{\Pos}{\mathbf{Pos}}

\newcommand{\ACat}{\mathscr{A}}
\newcommand{\BCat}{\mathscr{B}}

\newcommand{\DCat}{\mathscr{D}}

\newcommand{\MT}{\mathbf{T}}

\newcommand{\id}{\mathit{id}}
\newcommand{\Id}{\mathsf{Id}}

\newcommand{\Coalg}[1]{\mathbf{Coalg}\,#1}

\newcommand{\Vect}[1]{#1\text{-}\mathbf{Vec}}

\newcommand{\JSL}{{\mathbf{JSL}}}

\newcommand{\colim}{\mathop{\mathsf{colim}}}

\newcommand{\under}[1]{{|#1|}}

\newcommand{\cl}{\mathsf{cl}}
\newcommand{\cs}{\mathsf{cs}}

\newcommand{\epito}{\twoheadrightarrow}

\newcommand{\monoto}{\rightarrowtail}

\newcommand{\M}{\mathcal{M}}

\newcommand{\Pow}{\mathcal{P}}

\newcommand{\op}{\mathsf{op}}

\newcommand{\Syn}[1]{\mathsf{Syn}(#1)}

\newcommand{\K}{\mathbb{K}}

\newcommand{\Nat}{\mathds{N}}

\newcommand{\To}{\Rightarrow}

\mathchardef\ordinarycolon\mathcode`\:
\mathcode`\:=\string"8000
\begingroup \catcode`\:=\active
  \gdef:{\mathrel{\mathop\ordinarycolon}}
\endgroup

\mathchardef\hyph="2D
\newcommand{\dash}{\mathord{-}}
\usepackage[nomargin,marginclue,footnote,final]{fixme}
\FXRegisterAuthor{sm}{asm}{SM}
\FXRegisterAuthor{ls}{als}{LS}
\FXRegisterAuthor{hu}{ahu}{HU}

\usepackage{cleveref}
\usepackage{microtype}
\def\o{\cdot}

\acmConference[]{}{}{}
\acmYear{}
\acmISBN{} 
\acmDOI{} 
\startPage{1}
\setcopyright{none}

\usepackage{booktabs}   
\usepackage{subcaption} 

\begin{document}

\title[Automata Learning: An Algebraic Approach]{Automata Learning: An Algebraic Approach}         


\author{Henning Urbat}
\authornote{The authors acknowledge support by Deutsche Forschungsgemeinschaft (DFG) under project SCHR~1118/8-2.}          
\affiliation{
  \institution{Friedrich-Alexander-Universität Erlangen-Nürnberg}            
  \streetaddress{Martensstr. 3}
  \city{Erlangen}
  \postcode{91058}
  \country{Germany}                    
}
\email{henning.urbat@fau.de}          

\author{Lutz Schröder}
\authornotemark[1]         
\affiliation{
  \institution{Friedrich-Alexander-Universität Erlangen-Nürnberg}            
  \streetaddress{Martensstr. 3}
  \city{Erlangen}
  \postcode{91058}
  \country{Germany}                    
}
\email{lutz.schroeder@fau.de}         

\begin{abstract}
We propose a generic categorical framework for learning unknown
  formal languages of various types (e.g.~finite or infinite words, weighted and nominal languages). Our approach is parametric
  in a monad $\MT$ that represents the given type of languages and
  their recognizing algebraic structures. Using the concept of an
  automata presentation of $\MT$-algebras, we demonstrate that the
  task of learning a $\MT$-recognizable language can be reduced to
  learning an abstract form of algebraic automaton whose transitions are modeled by a functor. For the important case of adjoint automata, we devise a learning algorithm
  generalizing Angluin's $\mathsf{L}^*$. The algorithm is
  phrased in terms of categorically described extension steps; we
  provide for a termination and complexity analysis based on a
  dedicated notion of finiteness. Our framework applies to structures
  like $\omega$-regular languages that were not
  within the scope of existing categorical accounts of automata
  learning. In addition, it yields new learning algorithms for
  several types of languages for which no such algorithms were
  previously known at all, including sorted languages, nominal
  languages with name binding, and cost functions.
\end{abstract}



\maketitle

\section{Introduction}
Active automata learning is the task of inferring a finite
representation of an unknown formal language by asking questions to a
teacher. Such learning situations naturally arise, e.g., in software
verification, where the ``teacher'' is some reactive system and one
aims to construct a formal model of it by running suitable tests
\cite{vandr2017}. Starting with Angluin's~\cite{angluin87} pioneering
work on learning regular languages, active learning algorithms have
been developed for countless types of systems and languages, including
$\omega$-regular languages \cite{fcctw08,af16}, tree languages
\cite{dh03}, weighted languages \cite{bm15,hkrss19}, and nominal languages \cite{mssks17}. Most of these
extensions are tailor-made modifications of Angluin's $\mathsf{L}^*$ algorithm and thus
bear close structural analogies. This has motivated recent work
towards a uniform category theoretic understanding of automata learning,
based on modelling state-based systems as \emph{coalgebras} \cite{hss17_2,bkr19}. In the present paper, we propose a
novel \emph{algebraic} approach to automata learning.

Our contributions are two-fold. First, we study the problem of
learning an abstract form of automata originally introduced by Arbib and Manes \cite{am75} in the context of minimization: given an endofunctor $F$ on a category $\D$ and objects $I,O\in \D$, an
\emph{$F$-automaton} consists of an object $Q$ of states and morphisms $\delta_Q$, $i_Q$ and $f_Q$ as shown below, representing transitions, initial states and final states (or outputs).
\begin{equation*}
\vcenter{\xymatrix@R-1em{
& FQ \ar[d]^{\delta_Q} & \\
I \ar[r]_{i_Q} & Q \ar[r]_{f_Q} & O 
}
}
\end{equation*}
Taking $FQ=\Sigma\times Q$ on
$\Set$ with $I=1$ and $O=\{0,1\}$ yields classical deterministic
automata, but also several other notions of automata  (e.g.\ weighted automata, residual
nondeterministic automata, and nominal automata) arise as
instances. As our first main result, we devise a generalized
  $\mathsf{L}^*$ algorithm for \emph{adjoint $F$-automata}, i.e.~automata whose type functor $F$ admits a right adjoint $G$, based on alternating moves along
the initial chain for the functor $I+F$ and the final cochain for the
functor $O\times G$. Our generic algorithm subsumes known $\mathsf{L}^*$-type
algorithms for all the above classes of automata, and its analysis yields
uniform proofs of their correctness and termination. In addition, it also instantiates to a number of new learning algorithms, e.g.~ for sorted automata and for several versions of nominal automata with name binding.

We subsequently show that learning algorithms for $F$-automata (including our generalized $\mathsf{L}^*$ algorithm) apply far
beyond the realm of automata: they can be used to learn
languages representable by \emph{monads} \cite{boj15,uacm17}. Given a
monad $\MT$ on the category $\D$, we model a \emph{language} as a
morphism $L\colon TI \to O$ in $\D$. At this level of generality, one
obtains a concept of \emph{$\MT$-recognizable language} (i.e.\ a language
recognized by a finite $\MT$-algebra) that uniformly captures numerous
automata-theoretic classes of languages. For instance, regular
and $\omega$-regular languages (the languages accepted by
classical finite automata and Büchi automata, respectively) correspond precisely to
$\MT$-recognizable languages for the monads representing semigroups and
Wilke algebras,
\[
\MT I = I^+  \text{ on } \Set\quad\text{and}\quad \MT(I,J) = (I^+, I^{\mathsf{up}} + I^*\times J) 
  \text{ on } \Set^2.
\]
Here $I^{\mathsf{up}}$ denotes the set of ultimately periodic
infinite words over the alphabet $I$. For $\omega$-regular languages, Farzan et
al. \cite{fcctw08} proposed an algorithm that learns a language
$L\seq I^\omega$ of infinite words by learning the set of lassos in
$L$, i.e.~the regular language of \emph{finite} words given by
\[\mathsf{lasso}(L) = \{\, u\$v : u\in I^*, v\in I^+, uv^\omega \in L
\,\}\seq (I+ \{\$\})^*.\]  We show that this idea extends to general
$\MT$-recognizable languages, using the concept of an \emph{automata
  presentation}. Such a presentation allows
 for the \emph{linearization} of $\MT$-recognizable languages, i.e.~a reduction to ``regular'' languages accepted by finite $F$-automata for
suitable~$F$. 

In combination, our results yield a generic strategy for
learning an unknown $\MT$-recognizable language $L\colon TI\to O$: 
\begin{enumerate}
\item
find an automata presentation for the free $\MT$-algebra $T I$;
\item learn the minimal automaton for the linearization of $L$. 
\end{enumerate}
This approach turns out to be applicable
to a wide range of languages. In particular, it covers several settings for which no learning algorithms are known, e.g.~cost functions \cite{colcombet09}.

\smallskip \noindent \textsf{\textbf{Related work.}} A categorical
interpretation of several key concepts in Angluin's $\mathsf{L}^*$ algorithm
for classical automata was first given by Jacobs and Silva
\cite{js14}, and later extended to $F$-automata in a category, i.e.\
to similar generality as in the present paper, by van Heerdt,
Sammartino, and Silva \cite{hss17}. Their main contribution is an abstract categorical
framework (CALF) for correctness proofs of learning algorithms, while a
concrete generic algorithm is not given. Van Heerdt et
al. \cite{hss17_2} also study learning automata with side effects
modelled via monads; this use of monads is unrelated
to the monad-based abstraction of algebraic recognition in the present
paper. Barlocco, Kupke, and Rot~\cite{bkr19} develop a learning
algorithm for \emph{set} coalgebras (with all underlying concepts phrased categorically), parametric in a coalgebraic logic. Its scope is quite different from our generalized $\mathsf{L}^*$ algorithm: via genericity over the
branching type it covers, e.g., labeled
transition systems, but unlike our algorithm it does not apply to, e.g., nominal
automata. The connections between the two approaches are further discussed in \Cref{rem:coalglogic}.


Automata learning can be seen as an interactive version of automata
minimization, which has been extensively studied from a (co-)algebraic
perspective~\cite{am75,goguen75,at90,bkp12,hkrss19,cp17}. In particular, our
chain-based iterative learning algorithm resembles the coalgebraic approach to partition
refinement~\cite{abhks12}.

\section{Preliminaries}\label{sec:preliminaries}
We proceed to recall concepts from category theory and the theory of nominal sets that we will use throughout the paper. Readers should be familiar with basic notions such as functors, (co-)limits and adjunctions; see, e.g., Mac Lane~\cite{maclane}.

\medskip\noindent\textsf{\textbf{Functor (co-)algebras.}} Let $H\colon \D\to \D$ be an endofunctor on a category $\D$. An \emph{$H$-algebra} is a pair $(A,\alpha)$ consisting of an object $A\in \D$ and a morphism $\alpha\colon HA\to A$. A \emph{homomorphism} $h\colon (A,\alpha)\to (B,\beta)$ between $H$-algebras is a morphism $h\colon A\to B$ such that $h\o \alpha= \beta\o Fh$. An $H$-algebra $(A,\alpha)$ is  \emph{initial} if for every $H$-algebra $(B,\beta)$ there is a unique homomorphism $(A,\alpha)\to (B,\beta)$; we generally denote the initial algebra of~$H$ (unique up to isomorphism if it exists) as~$\mu H$. If $\D$ is cocomplete and $H$ preserves filtered colimits, $\mu H$ can be constructed as the colimit of the \emph{initial $\omega$-chain for $H$} \cite{adamek74}:
\[\mu H \;=\; \colim(\,0 \xra{\initial} H0 \xra{H\initial} H^2 0 \xra{H^2\initial} H^30 \to \cdots \,),\]
where $\initial$ is the unique morphism from the initial object $0$ of $\D$ into $H0$, and $H^n$ means $H$ applied $n$ times. Letting $j_n\colon H^n0 \to \mu H$ ($n\in \Nat$) denote the colimit cocone, we obtain the $H$-algebra structure on $\mu H$ as the unique morphism $\alpha\colon H(\mu H)\to \mu H$ satisfying 
\[\alpha\o Hj_n = j_{n+1} \quad \text{for all $n\in \Nat$}.\]
Dually, one has notions of a \emph{coalgebra} for the endofunctor $H$,
a \emph{coalgebra homomorphism}, and a \emph{final
  coalgebra}. Coalgebras provide an abstract notion of state-based
transition system: We think of the base object~$A$ of an $H$-coalgebra
as an object of \emph{states}, and of its structure map
$\alpha:A\to HA$ as assigning to each state a structured collection of
successors. Coalgebra homomorphisms are behaviour-preserving maps, and
final coalgebras have abstracted behaviours as states.

\medskip\noindent\textsf{\textbf{Monad algebras.}} A \emph{monad} $\MT=(T,\mu,\eta)$ on a category $\D$ is given by an endofunctor $T\colon \D\to \D$ and two natural transformations
$\eta\colon \Id_\D \to T$ and $\mu\colon TT\to T$ (the \emph{unit} and \emph{multiplication}) such that the following diagrams commute:
\[
\xymatrix@R-1em{
TTT \ar[r]^{T\mu} \ar[d]_{\mu T} & TT \ar[d]^\mu \\
TT \ar[r]_{\mu} & T
}
\qquad
\xymatrix@R-1em{
T \ar[r]^{T\eta} \ar@{=}[dr] & TT \ar[d]^\mu  & T \ar[l]_{\eta T} \ar@{=}[dl] \\
& T &
}
\]
A $\MT$-algebra is an algebra $(A,\alpha)$ for the endofunctor $T$ for which the following diagrams commute:
\[
\xymatrix@R-1em{
TTA \ar[r]^{\mu_A} \ar[d]_{T\alpha} & TA \ar[d]^\alpha \\
TA \ar[r]_{\alpha} & A
}
\qquad 
\xymatrix@R-1em{
A \ar@{=}[dr] \ar[r]^{\eta_A} & TA \ar[d]^\alpha \\
& A
}
\]
A \emph{homomorphism} of $\MT$-algebras is just a homomorphism of the underlying $T$-algebras. 
For each $X\in \D$, the $\MT$-algebra $\MT X = (TX,\mu_X)$ is called the \emph{free $\MT$-algebra} on $X$.

Monads form a categorical abstraction of algebraic theories \cite{manes76}. In fact, every algebraic theory (given by a finitary signature $\Gamma$ and a set $E$ of equations between $\Gamma$-terms) induces of monad $\MT$ on $\Set$ where $TX$ is the underlying set of the free $(\Gamma,E)$-algebra on $X$ (i.e.~the set of all $\Gamma$-terms over $X$ modulo equations in $E$), and the maps $\eta_X\colon X\to TX$ and $\mu_X\colon TTX\to TX$ are given by inclusion of variables and flattening of terms, respectively. Then the categories of $\MT$-algebras and $(\Gamma,E)$-algebras are isomorphic. Conversely, every monad $\MT$ on $\Set$ with $T$ preserving filtered colimits arises from some algebraic theory $(\Gamma,E)$ in this way.

Similarly, every ordered algebraic theory \cite{bloom76}, given by a signature $\Gamma$ and a set $E$ of inequations $s\leq t$ between $\Gamma$-terms,  yields a monad $\MT$ on the category $\Pos$ of posets whose algebras are ordered $\Gamma$-algebras (i.e.~$\Gamma$-algebras on a poset with monotone operations) satisfying the inequations in $E$.

\medskip\noindent \textsf{\textbf{Free monads.}} Let $H\colon \D\to \D$ be an endofunctor on a category $\D$ with coproducts, and suppose that, for each $X\in \D$, the initial algebra $\mu(X+H)$ for the functor $X+H$ exists. Then $H$ induces a monad $\MT_H$, the \emph{free monad over~$H$}~\cite{barr70}. It is given on objects by
$T_H X = \mu(X+H)$; its action on morphisms and the unit and multiplication are defined via initiality of the algebras $\mu(X+H)$. Then the categories of $\MT_H$-algebras and $H$-algebras are isomorphic: If $B+H(T_H B) \xra{[i_B,\alpha_B]} T_HB$ denotes the $B+H$-algebra structure of $T_HB=\mu(B+H)$, the isomorphism is given on objects by
\[(T_HB\xra{\beta} B)\quad\mapsto\quad  (HB\xra{Hi_B} H(T_H B) \xra{\alpha_B} T_HB \xra{\beta} B)  \]
and on morphisms by $h\mapsto h$.

\medskip\noindent\textsf{\textbf{Factorization systems.}} A \emph{factorization system} $(\E,\M)$ in a category $\D$ is given by two classes $\E$ and $\M$ of morphisms such that (i) $\E$ and $\M$ are closed under composition and contain all isomorphisms, (ii) every morphism $f$ has a factorization $f=m\o e$ with $e\in \E$ and $m\in \M$, and (iii) the \emph{diagonal fill-in} property holds: given a commutative square $m\o f = g\o e$ with $e\in \E$ and $m\in \M$, there exists a unique morphism $d$ with $f= d\o e$ and $g=m\o d$.
The morphisms $m$ and $e$ in (i) are unique up to isomorphism and are called the \emph{image} and \emph{coimage} of $f$. Categories of (co-)algebras typically inherit factorizations from their underlying category:
\begin{enumerate} 
\item If $H\colon \D\to \D$ is an endofunctor with $H(\E)\seq \E$, the factorization system $(\E,\M)$ for $\D$ lifts to the category of $H$-algebras, that is, every $H$-algebra homomorphism uniquely factorizes into a homomorphism in $\E$ followed by a homomorphism in $\M$. Dually, if $H(\M)\seq \M$, then the category of $H$-coalgebras has a factorization system lifting $(\E,\M)$.
\item If $\MT$ is a monad on $\D$ with $T(\E)\seq \E$, the factorization system $(\E,\M)$ for $\D$ lifts to the category of $\MT$-algebras.
\end{enumerate}
A factorization system $(\E,\M)$ is \emph{proper} if every morphism in $\E$ is epic and every morphism in $\M$ is monic. Whenever a proper factorization system $(\E,\M)$ is fixed, \emph{quotients} and \emph{subobjects} in $\D$ are represented by morphisms in $\E$ and $\M$, respectively. In particular, in the situation of (1) and (2) above, we represent \emph{quotient \mbox{(co-)}algebras} and \emph{sub\mbox{(co-)}algebras} by homomorphisms in $\E$ and $\M$, respectively.

\medskip\noindent\textsf{\textbf{Closed categories.}} A \emph{symmetric monoidal category} is a category $\D$ equipped with a functor $\otimes\colon \D\times \D\to \D$ (\emph{tensor product}), an object $I_\D\in \D$ (\emph{tensor unit}), and isomorphisms
\[ (X\t Y) \t Z \cong X\t (Y\t Z),\, X\t Y\cong Y\t X,\, I_\D\t X \cong X \cong X\t I_\D, \]
natural in $X,Y,Z\in \D$, satisfying coherence laws \cite[Chapter VII]{maclane}. $\D$ is \emph{closed} if the endofunctor $X\t (\dash)\colon \D\to \D$ has a right adjoint (denoted by $[X,\dash]$) for every $X\in \D$, i.e.~there is a natural isomorphism
$\D(X\t Y,Z) \cong \D(Y,[X,Z])$.

\medskip\noindent\textsf{\textbf{Nominal sets.}} Fix a countably infinite set $\At$ of \emph{names}, and let $\Perm(\At)$ be the group of all permutations $\pi\colon\At\to\At$ with $\pi(a)=a$ for all but finitely many $a$. A \emph{nominal set} \cite{pitts2013} is a set $X$ with a group action $\o\colon \Perm(\At)\times X\to X$ subject to the following property: for each $x\in X$ there is a finite set $S\seq \At$ (a \emph{support} of $x$) such that every $\pi\in\Perm(\At)$ that leaves all elements of $S$ fixed satisfies $\pi\o x = x$. This implies that $x$ has a least support $\supp(x)\seq \At$. The idea is that $x$ is a syntactic object with bound and free variables (e.g.~a $\lambda$-term modulo $\alpha$-equivalence), and that $\supp(x)$ is its set of free variables. A nominal set $X$ is \emph{orbit-finite} if the number of orbits (i.e. equivalence classes of the relation $x\equiv y$ iff $x=\pi\o y$ for some $\pi$) is finite. A map $f\colon X\to Y$ between nominal sets is \emph{equivariant} if $f(\pi\o x)=\pi\o f(x)$ for $x\in X$ and $\pi\in \Perm(\At)$. 
\section{Automata in a Category}\label{sec:automata}
We next develop the abstract categorical notion of automaton that
underlies our generic learning algorithm.
\begin{notation}
For the rest of this paper, let us fix
\begin{enumerate}
\item a category $\D$ with a proper factorization system $(\E,\M)$,
\item an endofunctor $F\colon \D\to \D$, and
\item two objects $I,O\in \D$.
\end{enumerate}
\end{notation}

\begin{definition}[Automaton (cf.~\cite{am75,at90})]\label{def:automaton}
An \emph{($F$-)automaton} is given by an object $Q\in \D$ of states and three morphisms 
\[\delta_Q\colon FQ\to Q,\quad i_Q\colon I\to Q, \quad f_Q\colon Q\to O,\]
representing transitions, initial states, and final states (or outputs), respectively. A \emph{homomorphism} between automata $(Q,\delta_Q,i_Q,f_Q)$ and $(Q',\delta_{Q'},i_{Q'} ,f_{Q'})$ is a morphism $h\colon Q\to Q'$ in $\D$ such that the following diagrams commute:
\[
\xymatrix@R-1em{
FQ \ar[r]^{\delta_Q} \ar[d]_{Fh} & Q \ar[d]^h \\
FQ' \ar[r]_{\delta_{Q'}} & Q' 
}
\qquad
\xymatrix@R-1em{
I \ar[r]^{i_Q} \ar[dr]_{i_{Q'}} & Q \ar[d]^h \ar[r]^{f_Q} & O \\
& Q' \ar[ur]_{f_{Q'}} &
}
\]
\end{definition}

\begin{example}[$\Sigma$-automata]\label{ex:categories}
Suppose that $(\D,\t,I_\D)$ is a symmetric monoidal closed category. Choosing the data
\[F=\Sigma\t (\dash),\quad I=I_\D, \quad\text{and}\quad O\in \D \text{ (arbitrary)}\] for a fixed input alphabet $\Sigma\in \D$ yields Goguen's notion of a \emph{$\Sigma$-automaton} \cite{goguen75}. In our applications, we shall work with the categories $\Set$ (sets and functions), $\Pos$ (posets and monotone maps), $\JSL$ (join-semilattices with $\bot$ and semilattice homomorphisms preserving $\bot$), $\Vect{\K}$ (vector spaces over field $\K$ and linear maps) and $\Nom$ (nominal sets and equivariant maps). The factorization systems and monoidal structures are given in the table below. In the fourth row, $\t$ is the usual tensor product of vector spaces representing bilinear maps. Similarly, in the third row, $\t$ is the tensor product of semilattices representing bimorphisms \cite{bn76}, i.e. semilattice morphisms $h\colon A\t B\to C$ correspond to maps $h'\colon A\times B\to C$ preserving $\vee$ and $\bot$ in each component. 
\begin{table}[ht]
\begin{tabular}[ht]{| l l l l l|}
\hline
  $\D$ & $(\E,\M)$ & $\t$ & $I_\D$ & $O$ \\
\hline
  $\Set$ & (surjective, injective) & $\times$ & $1$ & $\{0,1\}$ \\
  $\Pos$ & (surjective, embedding) & $\times$ & $1$ & $\{0<1\}$\\
  $\JSL$ & (surjective, injective) & $\t$ & $\{0<1\}$ & $\{0<1\}$ \\
  $\Vect{\K} $ & (surjective, injective)  & $\t$ & $\K$ & $\K$ \\
  $\Nom$ & (surjective, injective) & $\times$ & $1$ & $\{0,1\}$ \\
\hline
\end{tabular}
 \caption{Symmetric monoidal closed categories}
 \label{table:categories}
\end{table}

\noindent We choose the input alphabet $\Sigma\in \D$ to be a finite set, a discrete finite poset, a free semilattice on a finite set, a finite-dimensional vector space, and the nominal set $\At$ of atoms, respectively, and the output object $O\in \D$ as shown in the last column. Then $\Sigma$-automata are precisely classical
deterministic automata \cite{rs59}, ordered automata \cite{pin15},
semilattice automata \cite{kp08}, linear weighted automata \cite{dkv},
and nominal automata \cite{bkl14}. See
\Cref{ex:automata} and \ref{ex:nom-aut} for further details.
\end{example}

\begin{example}[Tree automata]\label{ex:treeaut}
Let $\Gamma$ be a signature and $F_\Gamma Q = \coprod_{n\in \Nat}\coprod_{\gamma\in \Gamma_n} Q^n$ on $\Set$ the induced polynomial functor, with $\Gamma_n$ the set of $n$-ary operations in $\Gamma$. Choosing $I=\emptyset$ and $O=2$, an $F_\Gamma$-automaton is a {(bottom-up) tree automaton} over $\Gamma$~\cite{tata2007}, shortly a \emph{$\Gamma$-automaton}. For the analogous functor $F_\Gamma$ on $\Pos$ and $O=\{0<1\}$, we obtain \emph{ordered $\Gamma$-automata}.
\end{example}
In the following, we focus on \emph{adjoint automata}, i.e.~automata whose transition type $F$ is a left adjoint:

\begin{assumptions}\label{asm}
For the rest of this section and in \Cref{sec:learningautomata}, our data is required to satisfy the following conditions:
\begin{enumerate}
\item\label{A1} $\D$ is complete and cocomplete; in particular, $\D$ has an initial object $0$ and a terminal object $1$.
\item\label{A3} The unique morphism $\initial\colon 0\to I$ lies in $\M$, and the unique morphism $!\colon O\to 1$ lies in $\E$.
\item\label{A5} The functor $F\colon \D\to \D$ has a right adjoint $G\colon \D\to \D$.
\item\label{A4} The functor $F$ preserves quotients ($F(\E)\seq \E$).

\end{enumerate}
\end{assumptions}
\begin{example}\label{ex:asm}
Every symmetric monoidal closed category $\D$ with $F=\Sigma\t \dash$ satisfies Assumption  \ref{A5}: closedness asserts precisely that $F$ has the right adjoint $G=[\Sigma,\dash]$. The categories $\D$ of \Cref{table:categories} also satisfy the remaining assumptions.
\end{example}

\begin{remark}\label{rem:algcoalg}
The key feature of our adjoint setting is that automata can be dually viewed as \emph{algebras} and \emph{coalgebras} for suitable endofunctors. In more detail:
\begin{enumerate}
\item An automaton $Q$ corresponds precisely to an algebra  
\[(\,F_I Q\xra{\alpha_Q} Q\,) \;=\; (\,I+FQ \xra{[i_Q,\delta_Q]} Q\,)\] for the endofunctor
$F_I = I + F$
equipped with an output morphism $f_Q\colon Q\to O$. Since $F_I$ preserves filtered colimits (using that the left adjoint $F$ preserves all colimits and the functor $I+(\dash)$ preserves filtered colimits), the
initial algebra $\mu F_I$ for $F_I$ emerges as
the colimit of the initial $\omega$-chain:
\[\mu F_I \;=\; \colim(\,0 \xra{\initial} F_I0 \xra{F_I\initial} F_I^2 0 \xra{F_I^2\initial} F_I^30 \to \cdots \,).\]
The colimit injections and the $F_I$-algebra structure on $\mu F_I$ are denoted by  
\[j_n\colon F_I^n0 \to \mu F_I \quad(n\in \Nat)\qquad\text{and}\qquad \alpha\colon F_I(\mu F_I)\to \mu F_I.\] 
 For any automaton $Q$ (viewed as an $F_I$-algebra), we write 
\[e_Q\colon \mu F_I \to Q\] for the unique $F_I$-algebra homomorphism from  $\mu F_I$ into $Q$.
\item Dually, replacing $\delta_Q\colon FQ\to Q$ by its adjoint transpose $\delta^@_Q\colon Q\to GQ$, an automaton can be presented as a coalgebra
\[(\,Q\xra{\gamma_Q} G_O Q\,) \;=\; (\,Q\xra{\langle f_Q,\delta^@_Q\rangle} O\times GQ\,)\] for the endofunctor $G_O = O\times G$
equipped with an initial state $i_Q\colon I\to Q$. Since $G_O$ preserves cofiltered limits, the
final coalgebra $\nu G_O$ arises as the limit of the final $\omega^\op$-cochain:
\[ \nu G_O = \lim (\,1 \xleftarrow{\terminal} G_O1 \xleftarrow{G_O\terminal} G_O^21 \xleftarrow{G_O^2\terminal} G_O^3 1\leftarrow  \cdots\,).\]
The limit projections and the $G_O$-coalgebra structure on $\nu G_O$  are denoted by
\[j_k'\colon \nu G_O \to G_O^k1\quad(k\in\Nat)\qquad\text{and}\qquad \nu G_O\xra{\gamma} G_O(\nu G_O).\] 
 For any automaton $Q$ (viewed as a $G_O$-coalgebra), we write 
\[m_Q\colon Q \to \nu G_O\] for the unique $G_O$-coalgebra homomorphism into $\nu G_O$.
\end{enumerate}
\end{remark}

\begin{definition}[Language]\label{def:language}
\begin{enumerate}
\item A \emph{language} is a morphism \[L\colon \mu F_I \to O.\]
\item The language \emph{accepted} by an automaton $Q$ is defined by
\[L_Q \;=\; (\,  \mu F_I \xra{e_Q} Q \xra{f_Q} O \,).\]
\end{enumerate}
\end{definition}

\begin{example}[$\Sigma$-automata, continued]\label{ex:automata}
\begin{enumerate}
\item In the setting of \Cref{ex:categories}, the initial algebra $\mu F_I$ and
  the initial chain for the functor $F_I = I_\D+\Sigma\t \dash$ can be
  described as follows \cite{goguen75}. Let
  $\Sigma^n = \Sigma\t \Sigma\t \cdots \t \Sigma$ denote the $n$th
  tensor power of $\Sigma$ (where $\Sigma^0=I_\D$), and put
  \[\Sigma^{<n} = \coprod_{m<n} \Sigma^m\;\;(n\in \Nat) \quad\text{and}\quad
  \Sigma^* = \coprod_{n\in \Nat} \Sigma^n.\] Then $\mu F_I$ is carried by the object~$\Sigma^*$ of words, and the initial chain is given by the coproduct injections
  \[\Sigma^{<0} \monoto \Sigma^{<1} \monoto \Sigma^{<2} \monoto \Sigma^{<3} \monoto
  \cdots.\]
\item For the functor $G_O = O \times [\Sigma,\dash]$ the final coalgebra $\nu G_O$ is carried by the object
  $[\Sigma^*,O]$ of languages and we have the final cochain
  \[[\Sigma^{<0},O] \leftarrow [\Sigma^{<1},O] \leftarrow
  [\Sigma^{<2},O] \leftarrow [\Sigma^{<3},O] \leftarrow \cdots \] with connecting morphisms given by restriction. To see this, consider the contravariant functor $P=[\dash,O]\colon \D\to\D^\op$. It is not difficult to verify that $P$ is a left adjoint (with right adjoint $P^{\op}$) and that there is a natural isomorphism 
\begin{equation*}PF_I \cong G_O^{\op} P.\end{equation*}
If  $\Alg{F_I}$ and $\Coalg{G_O}$ denote the categories of $F_I$-algebras and $G_O$-coalgebras, it follows \cite[Theorem 2.4]{hj98} that $P$ lifts to a left adjoint
$\ol P \colon \Alg{F_I}\to (\Coalg{G_O})^{\op}$
given by \[(\,F_IQ\xra{\alpha_Q} Q\,) \quad\mapsto\quad (\,PQ \xra{P\alpha_Q} PF_IQ \cong G_OPQ\,). \]  
Since left adjoints preserve initial objects, $\ol P$ maps the initial algebra $\mu F_I$ to the final coalgebra $\nu G_O$, i.e. one has $\nu G_O=P(\mu F_I)$ with the coalgebra structure
\[ \gamma = (\, \nu G_O= P(\mu F_I)  \xra{P\alpha} PF_I(\mu F_I) \cong G_OP(\mu F_I) = G_O(\nu G_O) \,). \]
Moreover, applying $P$ to the initial chain for $F_I$ yields the final cochain for $G_O$:
\[(\,1\xleftarrow{\terminal} G_O1 \xleftarrow{G_O\terminal} G_O^21 \cdots \,) = (\, P0\xleftarrow{P\initial} PF_I0 \xleftarrow{PF_I\initial} PF_I^2 0 \cdots \,).\]
Since $\mu F_I=\Sigma^*$ and $P=[\dash,O]$, we obtain the above description of $\nu G_O$ and of the final cochain for $G_O$.
 
%

\item For the categories of \Cref{table:categories}, the categorical notion of (accepted) language given
  in \Cref{def:language} thus specializes to the familiar ones. For
  illustration, let us spell out the case $\D=\Set$. A
  $\Sigma$-automaton in $\Set$ is precisely a classical deterministic automaton: it
  is given by a set $Q$ of states, a transition map
  $\delta_Q\colon \Sigma\times Q \to Q$, a map $i_Q\colon 1\to Q$
  (representing an initial state $q_0=i_Q(\ast)$), and a map
  $f_Q\colon Q\to 2$ (representing a set 
  $f_Q^{-1}[1]$ of final states). From (1) and (2) we obtain the well-known
  description of the initial algebra for $F_I=1+\Sigma\times \dash$ as
  the set $\Sigma^*$ of finite words over $\Sigma$ (with algebra
  structure $\alpha\colon 1+\Sigma\times \Sigma^*\to \Sigma^*$ given
  by $\ast\mapsto \epsilon$ and $(a,w)\mapsto wa$) and of the
  final coalgebra for $G_O=2\times [\Sigma,\dash]$ as the set
  $[\Sigma^*,2]\cong \Pow\Sigma^*$ of all languages $L\seq\Sigma^*$
  \cite{rutten}. The unique $F_I$-algebra homomorphism $e_Q\colon \Sigma^*\to Q$ maps a
  word $w\in \Sigma^*$ to the state of $Q$ reached on input $w$. Thus, the
  language $L_Q=f_Q\o e_Q$ accepted by $Q$ is the usual concept: $w$
  lies in $L_Q$ if and only if $Q$ reaches a final state on input $w$.
\end{enumerate}
\end{example}

\begin{example}[Nominal automata]\label{ex:nom-aut}
  Our notion of automaton (Definition~\ref{def:automaton}) has several natural instantiations to the
  category $\Nom$ of nominal sets and equivariant maps.
  \begin{enumerate}
  \item The simplest instance was already mentioned in \Cref{ex:categories}: 
    a $\Sigma$-automaton in $\Nom$
    corresponds precisely to a \emph{nominal deterministic
      automaton}~\cite{bkl14}. For simplicity, we choose the alphabet $\Sigma=\At$. A nominal automaton is given by a nominal set $Q$
     of states, an equivariant transition map
     $\delta_Q\colon \At\times Q \to Q$, an equivariant map
     $i_Q\colon 1\to Q$ (representing an equivariant initial
     state $q_0\in Q$), and an equivariant map $f_Q\colon Q\to 2$
     (representing an equivariant subset $F\seq Q$ of final
     states). The initial algebra $\At^*$ is the nominal set
     of words over $\At$ with group action
     $\pi\o (a_1\ldots a_n) = (\pi\o a_1)\ldots (\pi \o a_n)$
      for
      $a_1\ldots a_n\in \At^*$ and $\pi\in \Perm(\At)$.
     Thus, a language $L\colon \At^*\to 2$ corresponds to an
     equivariant set of words over $\At$. 
     
     Nominal automata with orbit-finite state space are known to be expressively equivalent to Kaminski and Francez' \cite{fk94} \emph{deterministic finite memory automata}.
		
  \item Now $\Nom$ carries a further symmetric monoidal closed structure, the
    \emph{separated product}~$\ast$ given on objects by
    \[X \ast Y=\{\,(x,y)\in X\times Y\;:\;
    x\,\#\,y\,\},\] where $x\,\#\,y$ means that $\supp(x)\cap \supp(y)=\emptyset$. The right
    adjoint of $F=\At \ast (-)$ is the \emph{abstraction functor}
    $G=[\At](-)$ \cite{pitts2013}
     which maps a nominal set~$X$ to the quotient of
     $\At\times X$ modulo
     the equivalence relation~$\sim$ defined by
     $(a,x)\sim(b,y)$ iff
     $(ac)\cdot x=(bc)\cdot y$ for some (equivalently, all)
     $c\in \At$ with $c\,\#\,a,b,x,y$.  We
     write $\langle a\rangle x$ for the equivalence class
     of $(a,x)$,
     which we think of as the result of binding the
     name~$a$ in~$x$. $F$-automata are precisely the \emph{separated automata} recently introduced by Moerman and Rot \cite{mr19}. 
\item By combining the adjunctions of (1) and (2), we obtain the adjoint pair of functors $F\dashv G$ with \[F=\At\times(-)+\At\ast(-),\quad G=[\At,\dash]\times[\At](-).\] The ensuing
    notion of automaton coincides with one used in Kozen et
    al.'s~\cite{KozenEA15} coalgebraic representation of nominal
    Kleene algebra~\cite{GabbayCiancia11}. Such automata have two
    types of transitions, \emph{free} transitions ($[\At,\dash]$) and
    \emph{bound} transitions ($[\At](-)$). They accept \emph{bar
      languages}~\cite{SchroderEA17}: putting
    $\bar\At=\At\cup\{\langle a\mid a\in\At\}$ (changing the original
    notation from $|a$ to $\langle a$ for compatibility with
    \emph{dynamic sequences} as discussed next), a \emph{bar string}
    is just a word over $\bar\At$. We consider $\langle a$ as
    binding~$a$ to the
    right. 
    This gives rise to the expected notions of free names and
    $\alpha$-equivalence $\alphaeq$. 
A
   bar string is \emph{clean} if its bound names are mutually
    distinct and distinct from all its free names. Simplifying
    slightly, we define a \emph{bar language} to be an equivariant set
    of bar strings modulo $\alpha$-equivalence, i.e.\ an equivariant
    subset of $\bar\At^*/\alphaeq$.  The initial algebra $\mu F_1$ is
    the nominal set of clean bar strings. A language in our sense is
    thus an equivariant set of clean bar strings; such languages are
    in bijective correspondence with bar
    languages~\cite{SchroderEA17}.
		
  \item\label{ex:dynseq} We note next that $[\At](-)$ is itself a left adjoint, our first
    example of a left adjoint that is not of the form $\Sigma\t \dash$
    for a closed structure $\t$. The right adjoint~$R$ is given on
    objects by
    $RX=\{\,f\in [\At,X]\;:\; a\,\#\,
    f(a) \text{ for all $a\in \At$}\,\}$~\cite{pitts2013}. We extend the above notion of automaton
    with this feature, i.e.\ we now work with the adjoint pair $F\dashv G$ given by
    \[ F=\At\times(-)+\At\ast(-)+[\At](-),\quad
    G=[\At,\dash]\times [\At](-)\times R.\] The initial algebra $\mu F_1$
    now consists of words built from three types of letters; we denote
    the new type of letters induced by the new summand $[\At](-)$
    in~$F$ by $a\rangle$ (for $a\in\At$). Recalling that words grow to
    the right, we see that $a\rangle$ binds to the left. 
    We read $a\rangle$ as \emph{deallocating} the name or
    resource~$a$. Languages in this model consist of \emph{dynamic
      sequences}~\cite{GabbayEA15}. We
    associate such languages with a species of nominal automata having
    three types of transitions: free and bound transitions as above,
    and \emph{deallocating transitions}
    $q\xrightarrow{a\rangle} q'$ with $a\,\#\, q'$. To the best of our knowledge, this notion of nominal automaton has not appeared in the literature before.
	\end{enumerate}
\end{example}

\begin{example}[Sorted $\Sigma$-automata]\label{ex:automata_sorted}
  In our applications in \Cref{sec:learningalgebras}, we shall encounter a generalized
  version of $\Sigma$-automata where (1) the input object $I$ is
  arbitrary, not necessarily equal to the tensor unit $I_\D$, and (2)
  the automaton has a sorted object of states and consumes sorted
  words. This reflects the fact that the algebraic structures arising in
  algebraic language theory are often sorted. For brevity, we only
  treat the case of sorted automata in $\Set$. Fix a set $S$ of sorts and
  a family of sets $\Sigma=(\Sigma_{s,t})_{s,t\in S}$; we think of the
  elements of $\Sigma_{s,t}$ as letters with domain sort $s$ and codomain
  sort $t$. We instantiate our setting to the adjoint pair
  $F\dashv G\colon \Set^S\to \Set^S$ defined as follows for $Q\in \Set^S$
  and $s,t\in S$:
  \begin{equation*}
    \textstyle (FQ)_t = \coprod_{s\in S} \Sigma_{s,t}\times Q_s,\qquad
    (GQ)_s = \prod_{t\in S} [\Sigma_{s,t},Q_t].
  \end{equation*}
  Choosing $I\in \Set^S$ arbitrary and the output object $O=2$, the $S$-sorted set with two
  elements in each component, an $F$-automaton is a \emph{sorted
    $\Sigma$-automaton}. It is given by an $S$-sorted set of states
  $Q$, transitions
  $\delta_{Q,s,t}\colon \Sigma_{s,t}\times Q_t\to Q_t$ ($s,t\in S$),
  initial states $i\colon I\to Q$ and an output map
  $f_{Q}\colon Q\to 2$ (representing an $S$-sorted set of final
  states). The initial algebra $\mu F_I$ is the $S$-sorted set of all
  well-sorted words over $\Sigma$ with an
  additional first letter from $I$. More precisely, $(\mu F_I)_t$ consists of all words
$xa_1\ldots a_n$ with $x\in \coprod_{s\in S} I_s$ and
$a_1,\ldots, a_n\in \coprod_{r,s}\Sigma_{r,s}$ such that the sorts of
consecutive letters match, i.e. there exist sorts
$s=s_0,s_1,\ldots,s_n=t\in S$ such that $x\in I_s$ and
$a_i\in \Sigma_{s_{i-1},s_i}$ for $i=1,\ldots, n$. In particular, in the single-sorted case we have $\mu F_I=I\times \Sigma^*$. For any well-sorted input
word $w=xa_1\ldots a_n$ one obtains the run
\[\xra{x} q_0 \xra{a_1} q_1 \to \cdots \xra{a_{n}} q_n\] in $Q$ where
$q_0 = i_{Q,s}(x)$ and $q_{i}=\delta_{Q,s_{i-1},s_{i}}(a_i,q_{i-1})$ for
$i=1,\ldots,n$, and $w$ is accepted if and only if $q_n$ is a final
state.
\end{example}
\noindent We conclude with a discussion of minimal automata.
\begin{definition}[Minimal automaton]
  An automaton $Q$ is called (1) \emph{reachable} if the unique $F_I$-algebra homomorphism $e_Q\colon \mu F_I\to Q$ lies in $\E$, and (2) \emph{minimal} if it is reachable and for every reachable automaton $Q'$ with $L_Q=L_{Q'}$, there exists a unique automata homomorphism from $Q'$ to $Q$.
\end{definition}

\begin{theorem}\label{thm:minaut} For every language $L$ there exists a minimal automaton $\Min{L}$ accepting $L$, unique up to isomorphism.
\end{theorem}

\begin{proof}[Proof sketch] 
We describe the construction of the minimal automaton. By equipping $\mu F_I$ with the final states $L\colon \mu F_I\to O$, we can view $\mu F_I$ as a $G_O$-coalgebra. Consider the $(\E,\M)$-factorization of the unique coalgebra homomorphism $m_{\mu F_I}$:
\[
m_{\mu{F_I}} \;=\; (\xymatrix@C+1em{\, \mu F_I \ar@{->>}[r]^<<<<<<{e_{\Min{L}}} & \Min{L} \ar@{>->}[r]^<<<<<<<{m_{\Min{L}}} & \nu G_O }). \]
The object $\Min{L}$ can be uniquely equipped with an automaton structure for which $e_{\Min{L}}$ is an $F_I$-algebra homomorphism and $m_{\Min{L}}$ is a $G_O$-coalgebra homomorphism. This automaton is the minimal acceptor for $L$.
\end{proof}
\noindent The minimization theorem and its proof are closely related to the classical work of Arbib and Manes \cite{am75} on the minimal realization of \emph{dynamorphisms}, i.e.~$F$-algebra homomorphisms from $\mu F_I$ into $\nu G_O$. Under different assumptions on the type functor $F$ and the base category $\D$ (e.g.~co-wellpoweredness), minimization results were also established by Ad\'amek and Trnkov\'a~\cite{at90} and, recently, by van Heerdt et al.~\cite{hkrss19}.

\section{A Categorical $\mathsf{L}^*$ Algorithm}\label{sec:learningautomata}
To motivate our learning algorithm for adjoint automata, we recall
Angluin's classical $\mathsf{L}^*$ algorithm~\cite{angluin87} for learning an unknown
$\Sigma$-automaton $Q$ in $\Set$. The algorithm assumes that the
learner has access to an oracle (the \emph{teacher}) that can be asked
two types of questions: 
\begin{enumerate}
\item \emph{Membership queries:} given a word
$w\in \Sigma^*$, is $w\in L_Q$?
\item \emph{Equivalence queries:} given an automaton $H$, is $L_H=L_Q$?
\end{enumerate} If the answer in (2) is ``no'', the teacher discloses a
\emph{counterexample}, i.e. a word
$w\in L_Q\setminus L_H \cup L_H\setminus L_Q$, to the learner.

The idea of $\mathsf{L}^*$ is to compute a sequence of approximations of the unknown automaton $Q$ by considering finite (co-)restrictions of the morphism $m_Q\o e_Q$, as indicated by the diagram below. Note that the kernel of $m_Q\o e_Q$ is precisely the well-known \emph{Nerode congruence} of $L_Q$.
\begin{equation}
\label{eq:lstar}
\vcenter{
\xymatrix@C-1.9em@R-1em{
\Sigma^{<0} \ar@{>->}[r] &  \cdots & \Sigma^ {<N} \ar@{>->}[r] & \Sigma^{<N+1} \ar@{>->}[r] & \cdots & \Sigma^* \ar@{-->}[dd]^{e_Q} \\
&&  S \ar@/_3em/[dd]_{h_{S,T}} \ar@{->>}[d]^{{e_{S,T}}} \ar@{>->}[u] &&& \\
&& {H_{S,T}} \ar@{>->}[d]^{{m_{S,T}}} &&& {Q} \ar@{-->}[dd]^{m_Q} \\ 
&& [T,2] &&& \\
[\Sigma^{<0},2] & \ar@{->>}[l] \cdot\cdot   & [{\Sigma^ {<K}},2] \ar@{->>}[u] & \ar@{->>}[l] ~~~[\Sigma^{<K+1},2] & \ar@{->>}[l]\cdot\cdot   & [\Sigma^*,2]
}
}
\end{equation}
In more detail, the algorithm maintains a pair $(S,T)$ of finite sets
$S,T\seq \Sigma^*$ (``states'' and ``tests''). For any such pair,
the restriction of $m_Q\o e_Q$ to the domain $S$ and codomain $[T,2]$,
\[h_{S,T}\colon S \to [T,2], \quad h_{S,T}(s)(t)=L_Q(st)
\quad \text{for $s\in S,\, t\in T,$}\]
is called the \emph{observation table} for $(S,T)$. It is usually represented as an $\under{S}\times\under{T}$-matrix with binary entries. The learner can compute $h_{S,T}$ via
  membership queries.  The pair $(S,T)$ is  \emph{closed} if
for each $s\in S$ and $a\in \Sigma$ there exists $s'\in S$ with
\[h_{S\cup S\Sigma,T}(sa)=h_{S,T}(s').\] 
It is \emph{consistent} if, for
all $s,s'\in S$,
\[ h_{S,T}(s)=h_{S,T}(s') \quad\text{implies}\quad h_{S,T\cup \Sigma T}(s)=h_{S,T \cup \Sigma T}(s').\]
Initially, one puts $S=T=\{ \epsilon \}$. If at some stage the pair $(S,T)$ is not closed or not consistent,  either $S$ or $T$ can be extended by invoking one of the following two procedures:

\medskip\noindent\fbox{\parbox{0.967\columnwidth}{\textbf{\underline{Extend $\mathbf{S}$}}\\
\textbf{Input:} A pair $(S,T)$ that is not closed.
\begin{enumerate}[leftmargin=1.8em]
\item[(0)] Choose $s\in S$ and $a\in \Sigma$ such that\[h_{S\cup S\Sigma,T}(sa)\neq h_{S,T}(s') \quad \text{for all $s'\in S$.}\]
\item[(1)] Put $S:= S\cup \{sa\}$.
\end{enumerate}
}}

\medskip\noindent\fbox{\parbox{0.8\columnwidth}{\textbf{\underline{Extend $\mathbf{T}$}}\\
\textbf{Input:} A pair $(S,T)$ that is not consistent.
\begin{enumerate}[leftmargin=1.8em]
\item[(0)] Choose $s,s'\in S$, $t\in T$ and $a\in\Sigma$ such that 
\[h_{S,T}(s)=h_{S,T}(s')\;\; \text{and}\;\; h_{S,T\cup \Sigma T}(s)(at) \neq h_{S,T\cup \Sigma T}(s')(at).\]
\item[(1)] Put $T:= T\cup \{at\}$.
\end{enumerate}
}}

\medskip\noindent The two procedures are applied repeatedly until the
pair $(S,T)$ is closed and consistent. Then, one constructs an automaton $H_{S,T}$, the \emph{hypothesis}
for $(S,T)$. Its set of states is
the image $h_{S,T}[S]$, the transitions $\delta_{S,T}\colon \Sigma\times H_{S,T}\to H_{S,T}$ are given by
$\delta_{S,T}(a,h_{S,T}(s)) = h_{S\cup S\Sigma,T}(sa)$ for $s\in S$
and $a\in \Sigma$, the initial state is $h_{S,T}(\epsilon)$, and a
state $h_{S,T}(s)$ is final if $s\in L_Q$ (i.e.\ $h_{S,T}(s)(\epsilon)=1$). Note that the well-definedness of $\delta_{S,T}$ is equivalent to
$(S,T)$ being closed and consistent.

The learner now tests whether $L_{H_{S,T}}=L_Q$ by asking an equivalence query. If the answer is ``yes'', the algorithm terminates successfully; otherwise,  the teacher's counterexample and all its prefixes are added to $S$. In summary:\\

\noindent\fbox{\parbox{0.98\columnwidth}{\textbf{\underline{$\mathbf{\mathsf{L}^*}$ Algorithm}}\\
\textbf{Goal:} Learn an automaton equivalent to an unknown automaton $Q$. 
\begin{enumerate}[leftmargin=0.5cm] 
\setcounter{enumi}{-1}
  \item Initialize $S=T=\{\epsilon\}$.
  \item\label{firststep} While $(S,T)$ is not closed or not consistent:
\begin{enumerate}\item If $(S,T)$ is not closed: Extend $S$.
\item If $(S,T)$ is not consistent: Extend $T$.
\end{enumerate}
  \item Construct the hypothesis $H_{S,T}$.
\begin{enumerate}
\item If $L_{H_{S,T}}=L_Q$: Return $H_{S,T}$.
\item If $L_{H_{S,T}}\neq L_Q$: Put $S:=S\cup C$, where $C$ is the set of prefixes of the teacher's counterexample.
\end{enumerate}
\item Go to \ref{firststep}.
\end{enumerate}}}\\~\\
The algorithm runs in polynomial time w.r.t.\ the size of the minimal
automaton $\Min{L_Q}$ and the length of the longest counterexample
provided by the teacher. The learned automaton (i.e.~the correct hypothesis returned in Step (2a)) is isomorphic
to $\Min{L_Q}$. Correctness and termination rest on the invariant
that~$S$ is prefix-closed and $T$ is suffix-closed. Note that if
$T\seq \Sigma^{<K}$, then $T$ yields a quotient
$[\Sigma^{<K},2]\epito [T,2]$ given by restriction. In the following,
$T$ is represented via this quotient.

\medskip \noindent We shall now develop all ingredients of $\mathsf{L}^*$ for adjoint $F$-automata. This requires additional assumptions, which hold for all the functors discussed in \Cref{ex:categories}, \ref{ex:nom-aut} and \ref{ex:automata_sorted}:
\begin{assumptions}\label{asm2}
On top of our \Cref{asm}, we require for the rest of this section that $F_I=I+F$ preserves subobjects ($F_I(\M)\seq \M$) and pullbacks of $\M$-morphisms, and that $G_O=O\times G$ preserves quotients ($G_O(\E)\seq \E$).
\end{assumptions}
Our categorical learning algorithm generalizes \eqref{eq:lstar} to the diagram shown below, where the upper and lower part are given by the initial chain for $F_I$ and the final cochain for $G_O$:
\begin{equation}
\label{eq:lstarcat}
\vcenter{
\xymatrix@C-0.9em@R-1em{
F_I^0 0 \ar@{>->}[r]_\initial &  \cdots & F_I^N0 \ar@/^3ex/[rrr]^{j_N} \ar@{>->}[r]_<<<{F_I^N \initial} & F_I^{N+1}0 \ar@{>->}[r]_<<<<{F_I^{N+1}\initial} & \cdots & \mu F_I \ar@{-->}[dd]^{e_Q} \\
&&  S \ar@/_3em/[dd]_{h_{s,t}} \ar@{->>}[d]^{{e_{s,t}}} \ar@{>->}[u]_s &&& \\
&& {H_{s,t}} \ar@{>->}[d]^{{m_{s,t}}} &&& {Q} \ar@{-->}[dd]^{m_Q} \\ 
&& T &&& \\
G_O^0 1 & \ar@{->>}[l]_<<<{!} \cdots & G_O^K1 \ar@{->>}[u]_t & \ar@{->>}[l]_{G_O^K!} G_O^{K+1}1 & \ar@{->>}[l]_<<{G_O^{K+1}!} \cdots &  \nu G_O \ar@/^3ex/[lll]^{j_K'}
}
}
\end{equation}
The algorithm maintains a pair $(s,t)$
of an $F_I$-subcoalgebra and a $G_O$-quotient algebra
\begin{equation}
\label{eq:st} s\colon (S,\sigma)\monoto (F_I^N0, F_I^N\initial), \quad t\colon (G_O^K1, G_O^K!)\epito (T,\tau),
 \end{equation}
with $N,K>0$. For $\Sigma$-automata in $\Set$, this means precisely that $S$ is a prefix-closed subset of $\Sigma^{<N}$, and that $T$ represents a suffix-closed subset of $\Sigma^{<K}$. 

Initially, one takes $N=K=1$, $s=\id_I$ and $t=\id_O$,
which corresponds to Step (0) of the original $\mathsf{L}^*$ algorithm.
\begin{remark}\label{rem:subcoalgext}
  By \Cref{asm}\ref{A3} and \ref{asm2}, every subcoalgebra
  $s\colon (S,\sigma)\monoto (F_I^N0,F_I^N\initial)$ induces the two
  subcoalgebras 
\[
\xymatrix@R-2em{
(S,\sigma)~ \ar@{>->}[r]^<<<<<{F_I^N\initial\o s} & ~(F_I^{N+1}0,F_I^{N+1}\initial)~ &  ~(F_IS,F_I\sigma). \ar@{>->}[l]_<<<<{F_Is}
}
\]
In the case of
  $\Sigma$-automata in $\Set$, the construction of these two subcoalgebras corresponds to
  viewing a prefix-closed subset $S\seq \Sigma^{<N}$ as a subset of
  $\Sigma^{<N+1}$, and to extending $S$ to the prefix-closed subset
  $S\Sigma\cup \{\epsilon\} = S\cup S\Sigma \seq \Sigma^{<N+1}$. A dual remark applies to quotient algebras
  of~$(G_O^K1,G_O^K!)$.
\end{remark}
\begin{definition}[Observation table] Let $(s,t)$ be a pair as in \eqref{eq:st}, and let $Q$ be an automaton.
The \emph{observation table} for $(s,t)$ w.r.t.\ $Q$ is the morphism
\[ h^Q_{s,t} \;=\; (\, S \xra{s} F_I^N0 \xra{j_N} \mu F_I \xra{e_Q} Q \xra{m_Q} \nu G_O \xra{j_K'} G_O^K1 \xra{t} T  \,).\]
Its $(\E,\M)$-factorization is denoted by
\[h_{s,t}^Q \;=\; (\xymatrix{S \ar@{->>}[r]^<<<<{e_{s,t}^Q} & H_{s,t}^Q \ar@{>->}[r]^<<<<<{m_{s,t}^Q} & T }).\]
In the following, we fix $Q$ (the unknown automaton to be learned) and omit the superscripts $(\dash)^Q$.
\end{definition}

\begin{remark}\label{rem:hst}
In our categorical setting, membership queries are replaced by the assumption that the learner can compute the observation table $h_{s,t}^Q$ for each pair $(s,t)$. Importantly, this morphism depends only on the language of $Q$: one can show that for every automaton $Q'$ with $L_Q=L_{Q'}$ one has $m_Q\o e_Q = m_{Q'}\o e_{Q'}$, whence $h_{s,t}^Q=h_{s,t}^{Q'}$.
\end{remark}

\begin{definition}[Closed/Consistent pair]
For any pair $(s,t)$ as in \eqref{eq:st}, let $\cl_{s,t}$ and $\cs_{s,t}$ be the unique diagonal fill-ins making all parts of the diagram below commute:
\[
\xymatrix@C+1em@R-1em{
	& H_{s,G_Ot} \ar@{-->>}[d]^{\cs_{s,t}} \ar@{>->}[r]^{m_{s,G_Ot}} & G_O T \ar[d]^\tau \\
S \ar@{->>}[ur]^<<<<<<<{e_{s,G_Ot}} \ar@{->>}[r]^{e_{s,t}} \ar[d]_\sigma & H_{s,t} \ar@{>->}[r]^{m_{s,t}} \ar@{>-->}[d]_{\cl_{s,t}} & T \\
F_IS \ar@{->>}[r]_{e_{F_Is,t}}  & H_{F_Is,t} \ar@{>->}[ur]_<<<<<<<<{~~m_{F_Is,t}} &	
}
\]
The pair $(s,t)$ is \emph{closed} if $\cl_{s,t}$ is an isomorphism, and \emph{consistent} if $\cs_{s,t}$ is an isomorphism.
\end{definition}
\noindent If $(s,t)$ is not closed or not consistent, at least one of the two dual procedures below applies. ``Extend $s$'' replaces $S\monoto F_I^N0$ by a new subcoalgebra $S'\monoto F_I^{N+1}0$, i.e. it moves to the right in the initial chain for $F_I$. Analogously, ``Extend $t$'' replaces $G_O^K1\epito T$ by a new quotient algebra $G_O^{K+1}1\epito T'$, and thus moves to the right in the final cochain for $G_O$.

\medskip\noindent\fbox{\parbox{0.98\columnwidth}{\textbf{\underline{Extend $\mathbf{s}$}}\\
\textbf{Input:} A pair $(s,t)$ as in \eqref{eq:st} that is not closed.
\begin{enumerate}[leftmargin=1.8em]
\item[(0)] Choose an object $S'$ and $\M$-morphisms $s_0\colon S\monoto S'$ and $s_1\colon S'\monoto F_I S$ such that \[\sigma=s_1\o s_0 \quad \text{and}\quad e_{F_I s,t}\o s_1\in \E.\]
\item[(1)] Replace $s\colon (S,\sigma)\monoto (F_I^N0,F_I^N\initial)$ by the subcoalgebra 
\[F_Is \o s_1\colon (S',F_Is_0\o s_1) \monoto (F_I^{N+1}0, F_I^{N+1}\initial).\]
\end{enumerate}
}}
\begin{remark}\label{rem:extends}
\begin{enumerate}
\item One trivial choice in Step (0) is \[S'=F_IS\,\quad s_0=\sigma,\quad s_1 =\id.\] To get an efficient implementation of the algorithm, one aims to choose the subobject $s_1\colon S'\monoto F_IS$ as small as possible. 
\item The update of $s$ in Step (1) is well-defined, i.e.\ $F_I s\o s_1$ is a subcoalgebra. Indeed, the commutative diagram below shows that $F_Is\o s_1$ is a coalgebra homomorphism:
\[  
\xymatrix@R-1em{
	F_I^{N+1}0 \ar[rr]^{F_I^{N+1}\initial} & & F_I(F_I^{N+1}0) \\
	F_I S \ar[u]^{F_I s} \ar[rr]^{F_I\sigma} & & F_IF_I S \ar[u]_{F_IF_Is} \\
	S' \ar@{>->}[u]^{s_1}  \ar[r]_{s_1} & \ar[ur]^{F_I\sigma} F_I S \ar[r]_{F_I s_0} & F_I S' \ar@{>->}[u]_{F_Is_1} 
}
\]
Moreover, since $s, s_1\in \M$ and $F_I$ preserves $\M$ (see \Cref{asm2}), we have $F_Is\o s_1\in \M$.
\item In the case of $\Sigma$-automata in $\Set$, the condition $\sigma=s_1\o s_0$ states that $S\seq S'\seq S\cup S\Sigma = S\Sigma \cup \{\epsilon\}$. The condition $e_{F_Is,t}\o s_1\in \E$ asserts that given $s\in S$ and $a\in \Sigma$ such that $h_{S\cup S\Sigma,T}(sa)\neq h_{S,T}(r)$ for all $r\in S$, there exists $s'\in S'$ with $h_{S\cup S\Sigma,T}(sa)=h_{S\cup S\Sigma, T}(s')$. Thus, ``Extend $s$'' subsumes several executions of ``Extend $S$'' in the original $\mathsf{L}^*$ algorithm.
\end{enumerate}
\end{remark}

\noindent\fbox{\parbox{0.98\columnwidth}{\textbf{\underline{Extend $\mathbf{t}$}}\\
\textbf{Input:} A pair $(s,t)$ as in \eqref{eq:st} that is not consistent.
\begin{enumerate}[leftmargin=1.8em]
\item[(0)] Choose an object $T'$ and $\E$-morphisms
  $t_0\colon G_OT\epito T'$ and $t_1\colon T'\epito T$ such that
  \[\tau=t_1\o t_0\quad\text{and}\quad t_0\o m_{s,G_O t}\in \M.\]
\item[(1)] Replace $t\colon (G_O^K1,G_O^K\terminal)\epito (T,\tau)$ by the quotient algebra
\[ t_0\o G_Ot\colon (G_O^{K+1}1,G_O^{K+1}\terminal)\epito (T', t_0 \o G_Ot_1 ).\]
\end{enumerate}
}}
\begin{remark}\label{rem:extendt}
\begin{enumerate}\item Dually to \Cref{rem:extends}, a trivial choice in Step (0) is given by $T'=G_O T$, $t_0=\id$, $t_1=\tau$,
 and Step (1) is well-defined, i.e. $t_0\o G_O t$ is a quotient algebra. \item In the case of $\Sigma$-automata in $\Set$, we view the quotients $T$ and $T'$ as subsets of $\Sigma^{<K}$ and $\Sigma^{<K+1}$, respectively, using the above identification between subsets and quotients.  The condition $\tau=t_1\o t_0$ then states that $T \seq T'\seq T\cup \Sigma T$. The condition $t_0\o m_{s,G_O t}\in \M$ states that every inconsistency admits a witness in $T'$: given $s,s'\in S$ with $h_{S,T}(s)=h_{S,T}(s')$ but $h_{S,T\cup \Sigma T}(s) \neq h_{S,T\cup \Sigma T}(s')$, there exists $t'\in T'$ with $h_{S,T'}(s)(t')\neq h_{S,T'}(s')(t')$. Thus, ``Extend $t$'' subsumes several executions of ``Extend $T$'' in the original $\mathsf{L}^*$ algorithm.
\end{enumerate}
\end{remark}
\noindent If $(s,t)$ is both closed and consistent, then we can define
an automaton structure on $H_{s,t}$:

\begin{definition}[Hypothesis]\label{def:hypothesis}
Let the pair $(s,t)$ be closed and consistent. The \emph{hypothesis} for $(s,t)$ is the automaton \[(H_{s,t},\delta_{s,t},i_{s,t},f_{s,t})\]
with states $H_{s,t}$ and structure defined below. Here, $\inl$/$\inr$ are coproduct injections, $\outl$/$\outr$ are product projections, and $(\dash)^\#$ denotes adjoint transpose along the adjunction $F\dashv G$:
\begin{enumerate}[topsep=0.2em]
\item The transitions $\delta_{s,t}\colon FH_{s,t}\to H_{s,t}$ are given by the diagonal fill-in of the  commutative square
\[
\xymatrix@R-1em{
FS \ar[d]_{l_{s,t}} \ar@{->>}[r]^{Fe_{s,t}} & FH_{s,t} \ar@{-->}[dl]_{\delta_{s,t}} \ar[d]^{r_{s,t}^\#}  \\
H_{s,t}\ar@{>->}[r]_{m_{s,t}} & T
}
\]
with the two vertical morphisms defined by
\begin{align*}
l_{s,t} &= (FS \xra{\inr} I+FS=F_I S \xra{e_{F_Is,t}} H_{F_I s,t} \xra{\cl_{s,t}^{-1}} H_{s,t} ),\\	r_{s,t} &= (H_{s,t} \xra{\cs_{s,t}^{-1}} H_{s,G_O t} \xra{m_{s,G_O t}} G_O T=O\times GT  \xra{\outr} GT).
\end{align*}
\item The initial states are \[i_{s,t} \;=\; (\, I \xra{\inl} I+F S = F_I S \xra{e_{F_I s,t}}  H_{F_I s,t} \xra{\cl_{s,t}^{-1}} H_{s,t} \,).\]
\item The final states  are \[f_{s,t} \;=\; (\, H_{s,t} \xra{\cs_{s,t}^{-1}} H_{s,G_O t} \xra{m_{s,G_O t}} G_O T = O\times GT \xra{\outl} O  \,).\]
\end{enumerate}
\end{definition}

\begin{remark}\label{rem:diagfillin}
The square defining $\delta_{s,t}$ commutes: both legs can be shown to be equal to $FS\xra{\inr } I+FS = F_IS \xra{h_{F_Is,t}} T$. The idea of constructing the $F$-algebra structure of a hypothesis via diagonal fill-in originates in the abstract framework of CALF \cite{hss17}. An important difference is that in the latter  the existence of the two vertical morphisms of the corresponding square is postulated, while our present setting features a concrete description of $l_{s,t}$ and $r_{s,t}$.
\end{remark}
\noindent Recall that in $\mathsf{L}^*$, if a hypothesis $H_{S,T}$ is not correct (i.e. $L_{H_{S,T}}\neq L_Q$), the learner receives a counterexample $w\in \Sigma^*$ from the teacher and adds the set $C$ of all its prefixes to $S$. Identifying the word $w$ with this set, the concept of a counterexample has the following categorical version:
\begin{definition}[Counterexample]
Let $(s,t)$ be closed and consistent. A \emph{counterexample for $H_{s,t}$} is a subcoalgebra 
\[c\colon (C,\gamma)\monoto (F_I^M 0, F_I^M \initial) \quad \text{for some $M>0$}\]
such that $H_{s,t}$ and $Q$ do not agree on inputs from $C$, that is, 
\[L_{H_{s,t}}\o j_M\o c \;\neq\; L_Q\o j_M\o c.\]
\end{definition}

\begin{remark}
\begin{enumerate}
\item If $L_{H_{s,t}}\neq L_Q$, then a counterexample always exists. Indeed, since the colimit injections $j_M\colon F_I^M 0\to \mu F_I$ are jointly epimorphic, one has $L_{H_{s,t}}\o j_M \neq L_Q\o j_M$ for some $M>0$ and thus $(C,\gamma)=(F_I^M0, F_I^M\initial)$ is a counterexample. To obtain an efficient algorithm, it is often assumed that the teacher delivers a \emph{minimal} counterexample, i.e. $M$ is minimal and no proper subcoalgebra is a counterexample.
\item Given a counterexample $c\colon (C,\gamma)\monoto (F_I^M0,F_I^M\initial)$, one can add $c$ to the subcoalgebra $s\colon (S,\sigma)\monoto (F_I^N0,F_I^N\initial)$ as follows: by \Cref{rem:subcoalgext}, we can assume that $M=N$, and then form the supremum 
$s\vee c\colon (S\vee C, \sigma\vee \gamma) \monoto (F_I^N0,F_I^N\initial)$ of $s$ and $c$ in the lattice of subcoalgebras of $(F_I^N0,F_I^N\initial)$, viz. the image of the homomorphism $[s,c]\colon S+C\to  F_I^N0$.
\end{enumerate}
\end{remark}
\noindent With all these ingredients at hand, we obtain the following abstract learning algorithm for adjoint $F$-automata:\\~\\
\noindent\fbox{\parbox{0.98\columnwidth}{\textbf{\underline{Generalized $\mathbf{\mathsf{L}^*}$ Algorithm}}\\
\textbf{Goal:} Learn an automaton equivalent to an unknown automaton $Q$. 
\begin{enumerate}[leftmargin=0.5cm] 
\setcounter{enumi}{-1}
  \item Initialize $N=K=1$, $s=\id_I$ and $t=\id_O$.
  \item While $(s,t)$ is not closed or not consistent:
\begin{enumerate} 
\item If $(s,t)$ is not closed: Extend $s$.
\item If $(s,t)$ is not consistent: Extend $t$.
\end{enumerate}
  \item Construct the hypothesis $H_{s,t}$.
\begin{enumerate}
  \item If $L_{H_{s,t}}=L_{Q}$: Return $H_{s,t}$.
\item If $L_{H_{s,t}}\neq L_Q$: Replace the subcoalgebra $s$ by $s\vee c$, where $c$ is the teacher's counterexample.
\end{enumerate}
\item Go to (1).
\end{enumerate}}}

\vspace*{0.3cm}\noindent To prove the termination and correctness of
Generalized $\mathsf{L}^*$, we need a finiteness assumption on the unknown
automaton $Q$. We call a $\D$-object $Q$ \emph{Noetherian} if both its
poset of subobjects (ordered by $m\leq m'$ iff $m=m'\o p$ for some
$p$) and that of its quotients (ordered by $e\leq e'$ iff $e=q\o e'$
for some $q$) contain no infinite strictly ascending chains.

\begin{theorem}\label{thm:termination}
If $Q$ is Noetherian, then the generalized $\mathsf{L}^*$ algorithm terminates and returns $\Min{L_Q}$.
\end{theorem}

\begin{remark}\label{rem:complexity}
  Under a slightly stronger finiteness condition on $Q$, we obtain a
  complexity bound. Suppose that $Q$ has finite \emph{height} $n$, that is, $n$ is the maximum length of any strictly ascending chain of subobjects or quotients of $Q$. Then Steps (1a), (1b) and (2b) are executed 
  $O(n)$ times.
\end{remark}

\begin{example}\label{ex:noetherian}
  In $\D=\Set$, $\Pos$, $\JSL$, $\Vect{\K}$, and $\Nom$, the Noetherian
  objects are precisely the finite sets, finite posets, finite semilattices,
  finite-dimensional vector spaces and orbit-finite nominal sets. The
  height of $Q$
  is equal to the number of elements of $Q$ (for $\D=\Set,\Pos$) or
  the dimension (for $\D=\Vect{\K}$). For $\D=\Nom$, the height of an
  orbit-finite set~$Q$ can be shown to be polynomial in the number of orbits of~$Q$
  and $\max\{\,|\supp(q)|\mid q\in Q\,\}$, 
  using upper bounds on the length of subgroup chains in symmetric
  groups~\cite{Babai86}.
\end{example}

\begin{remark}\label{rem:counterex}
  In the generalized $\mathsf{L}^*$ algorithm, counterexamples are added to
  $S$. Dually, one may opt to add them to $T$ instead; for
  $\Sigma$-automata in $\Set$, this corresponds to a modification of
  Angluin's algorithm due to Maler and Pnueli~\cite{mp95} that makes it possible to avoid inconsistent observation tables, i.e.~all tables constructed in the modified algorithm are
  consistent. In this dual approach, the accepted language of an
  automaton~$Q$ is defined coalgebraically as the morphism
\[{L}_Q'\;=\;(\,I\xra{i_Q}Q \xra{m_Q} \nu G_O\,),\]
and a \emph{counterexample} is a quotient algebra $c\colon (G_O^M1,G_O^M\terminal)\epito (C,\gamma)$ for some $M>0$ such that $c\o j_M' \o {L}'_{H_{s,t}} \neq c\o j_M' \o {L}_{Q}'$.
In Step (2b), a counterexample $c$ is added to the quotient algebra $t\colon (G_O^K1,G-O^K\terminal)\epito (T,\tau)$ by forming the supremum of $t$ and $c$. To guarantee termination, our original requirement that $F_I$ preserves pullbacks of $\M$-morphisms (see \Cref{asm2}) needs to be replaced by the dual requirement that $G_O$ preserves pushouts of $\E$-morphisms.
\end{remark}

\begin{remark}\label{rem:coalglogic}
We elaborate on the connection between Generalized $\mathsf{L}^*$ and the learning algorithm for coalgebras due to Barlocco et al.~\cite{bkr19}. The latter is concerned with coalgebras whose semantics is given in terms of a \emph{coalgebraic logic}, i.e.~a natural transformation $\delta\colon L^\op P\to PB$ where $L\colon \ACat\to \ACat$ and $B\colon \BCat\to \BCat$ are endofunctors and $P\colon \BCat \to \ACat^\op$ is a left adjoint (see the left-hand square below).
\[
\xymatrix@R-1em{
\ACat^\op \ar@{=>}[dr]|\delta \ar[d]_{L^\op} & \ar[l]_<<<<<P \BCat \ar[d]^B \\
\ACat^\op & \BCat \ar[l]^<<<<<{P}
}
\quad\quad
\xymatrix@R-1em{
(\D^\op)^\op \ar[d]_{(G_O^\op)^\op} \ar@{=>}[dr]|\id & \ar[l]_<<<<<{\Id} \D \ar[d]^{G_O} \\
(\D^\op)^\op & \D \ar[l]^<<<<<{\Id}
}
\]
Here, $L$ represents the syntax (usually modalities over  a
propositional base logic embodied by~$\ACat$), and~$B$ the behaviour
(defining the branching type of coalgebras on~$\BCat$).  The coalgebraic semantics
of $F$-automata corresponds to the trivial logic shown in the
right-hand square. In this sense, $F$-automata are formally covered
by the framework of \cite{bkr19}.

While Generalized $\mathsf{L}^*$ is based on Angluin's $\mathsf{L}^*$ algorithm, the
coalgebraic learning algorithm in \emph{op.~cit.} generalizes Maler
and Pnueli's approach, and thus needs to keep
observation tables consistent (\Cref{rem:counterex}). To this end,
tables are required to satisfy a property called \emph{sharpness}, which entails that the existence of extensions of non-closed tables is
nontrivial and can only be guaranteed under strong
assumptions on epimorphisms in the base category (e.g., all
epimorphisms must split). Thus, the algorithm is effectively limited to coalgebras in $\Set$ and does not apply, e.g., to $\Sigma$-automata in $\Nom$; see Appendix. In our Generalized~$\mathsf{L}^*$, no such
assumptions are needed since table extensions always exist
(\Cref{rem:extends}). This makes our algorithm
applicable in categories beyond $\Set$, including the ones in \Cref{ex:categories}. 
\end{remark}
\noindent Generalized $\mathsf{L}^*$ provides a unifying perspective on known
learning algorithms for several notions of deterministic automata,
including classical $\Sigma$-automata ($\D=\Set$ \cite{angluin87}),
linear weighted automata ($\D=\Vect{\K}$ \cite{bm15}) and nominal
automata ($\D=\Nom$ \cite{bhlm14,mssks17}). For $\D=\JSL$, finite semilattice automata can be interpreted as \emph{nondeterministic} finite automata by means of an equivalence between the category of finite semilattices and a suitable category of finite closure spaces and relational morphisms \cite{mamu14,ammu14_2}. For any regular language $L$, the minimal
$\Sigma$-automaton $\Min{L}$ in $\JSL$ corresponds under this equivalence to the minimal
  \emph{residual finite state automaton (RFSA)} \cite{dlt01}, a canonical nondeterministic acceptor for $L$ whose states are the
join-irreducible elements of $\Min{L}$. Consequently, the $\mathsf{NL}^*$ algorithm for learning
RFSA due to Bollig et al. \cite{BHKL09} is
also subsumed by our categorical setting. We note that although $\mathsf{NL}^*$ learns a minimal RFSA, the intermediate hypotheses arising in the algorithm are not necessarily RFSA, but general nondeterministic finite automata. Our categorical perspective provides an explanation of this phenomenon: it shows that $\mathsf{NL}^*$ implicitly computes deterministic finite automata over $\JSL$, and not every such automaton corresponds to an RFSA.

 Finally, our algorithm instantiates to
new learning algorithms for nominal languages with name binding,
including languages of dynamic sequences (\Cref{ex:nom-aut}), and for
sorted languages (\Cref{ex:automata_sorted}). A special instance of
sorted automata where all transitions are sort-preserving
(i.e. $\Sigma_{s,t}=\emptyset$ for $s\neq t$) appeared in the work of
Moerman \cite{moerman19} on learning product automata.

In each of the above settings, in order to turn Generalized $\mathsf{L}^*$ into a
concrete algorithm, one only needs to provide a suitable data
structure for representing observation tables $h_{s,t}$ by finite means, and a strategy for choosing the objects
$S'$ and $T'$ in the procedures ``Extend $s$'' and ``Extend $t$''. We emphasize that these design choices can be non-trivial and depend on the specific structure of the underlying category $\D$. The typical approach is to represent the map $h_{s,t}\colon S\to T$ by restricting the objects $S$ and $T$ to 
 finite sets of generators. For instance, finite-dimensional vector spaces can be represented by their bases ($\D=\Vect{\K}$), finite semilattices by their join-irreducible elements ($\D=\JSL$) and orbit-finite sets by subgroups of finite symmetric groups ($\D=\Nom$).

Our above results demonstrate, however, that the core of our learning algorithm is independent from such implementation details; in particular, its correctness and termination, and parts of the complexity analysis, always come for free as instances of the general results in \Cref{thm:termination} and \Cref{rem:complexity}. In this way, the categorical approach provides a clean separation between generic structures and design choices tailored to a specific application. This leads to a simplified derivation of learning algorithms in new settings.

\section{Learning Monad-Recognizable Languages}\label{sec:learningalgebras}

In this section, we investigate languages recognizable by monad algebras and show that the task of learning them can be reduced to learning $F$-automata.

\begin{notation}
  Fix a monad $\MT=(T,\mu,\eta)$ on $\D$ that preserves quotients
  ($T(\E)\seq \E$). We continue to work with the fixed objects
  $I,O\in\D$ of inputs and outputs (with~$I$ now thought of as an
  input alphabet, so not normally the monoidal unit). Finally, we
  fix a full subcategory $\D_f\seq \D$ closed
  under subobjects and quotients, and call the objects of $\D_f$ the \emph{finite} objects of
  $\D$.
\end{notation}
\begin{example}\label{ex:monads}
Choose $\Set_f$, $\Pos_f$, $\JSL_f$, $\Vect{\K}_f$ and $\Nom_f$ to be the class of all Noetherian objects (see \Cref{ex:noetherian}). Our monads of interest model formal languages:
\begin{center}
	\begin{tabular}{| l l |}
\hline
		$\D$ & $\MT$ \\
		\hline
		$\Set$ & $T_+X=X^+$  \\
		$\Set^2$ & $T_\infty(X,Y)=(X^+,X^{\mathsf{up}}+X^*Y)$ \\
		$\Set$ & $T_\Gamma X=$ $\Gamma$-trees over $X$ \\
		$\JSL$ & $T_* X=$ free idempotent semiring on $X$  \\
		$\Vect{\K}$ & $T_*X=$ free $\K$-algebra on $X$  \\
		$\Pos$ & $T_SX=$ free stabilization algebra on $X$ \\
		$\Nom$ & $T_*X=X^*$ \\
\hline
	\end{tabular}
\end{center}
In the second row, 
  $X^{\mathsf{up}}=\{\,vw^\omega\;:\; v\in X^*,\, w\in X^+\,\}$ denotes the set of ultimately periodic words over $X$, and in the third row, $\Gamma$ is a finitary algebraic signature. Finite algebras for the above seven monads correspond to finite semigroups, finite Wilke algebras \cite{wilke91}, finite $\Gamma$-algebras, finite-dimensional $\K$-algebras, finite stabilization algebras \cite{pin16}, and orbit-finite nominal monoids \cite{boj2013}, respectively.
\end{example}
\noindent In the present setting, we shall consider the following generalized concept of a language:
\begin{definition}[Language]\label{def:languaget}
A \emph{language} is a morphism \[L\colon TI\to O \quad \text{in}\quad \D.\] It is called \emph{recognizable} if there exists a
$\MT$-homomorphism $e\colon (TI,\mu_I)\to(A,\alpha)$ into a finite $\MT$-algebra $(A,\alpha)$ and a morphism $p\colon A\to O$ in $\D$ with $L=p\o e$. 
\[
\xymatrix@R-1em@C-1em{
TI \ar[dr]_e \ar[rr]^L & & O \\
 & A \ar[ur]_p & 	
}
\]
In this case, we say that \emph{$e$ recognizes $L$ (via $p$)}.
\end{definition}

\begin{remark}\label{rem:fexamples}
The above definition generalizes the concepts of the previous sections. Indeed, if $F$ is functor for which the free monad $\MT_F$ (see \Cref{sec:preliminaries}) exists, then a language $L\colon T_F I\to O$ in the sense of \Cref{def:languaget} is precisely a language $L\colon \mu F_I\to O$ in the sense of \Cref{def:language}. Moreover, since the categories of $F$-algebras and $\MT_F$-algebras are isomorphic, $L$ is $\MT_F$-recognizable if and only if $L$ is \emph{regular}, i.e. accepted by some finite $F$-automaton.
\end{remark}

\begin{example}\label{ex:reclang}
	Many important automata-theoretic classes of languages can be characterized algebraically as recognizable languages for a monad. For the monads of \Cref{ex:monads} we obtain the following languages:
\begin{center}
\begin{tabular}{| l l l |}
\hline
  $\D$ & $\MT$ & $\MT$-recognizable languages \\
\hline
  $\Set$ & $\MT_+$ & regular languages \cite{pin15}  \\
  $\Set^2$ & $\MT_\infty$ & $\omega$-regular languages \cite{pinperrin04} \\
 $\Set$ & $\MT_\Gamma$  & tree languages over $\Gamma$ \cite{tata2007} \\
  $\JSL$ & $\MT_*$ & regular languages \cite{polak01}  \\
  $\Vect{\K}$ & $\MT_*$ & recognizable weighted languages \cite{reu80} \\
  $\Pos$ & $\MT_S$ & regular cost functions \cite{colcombet09}\\
  $\Nom$ & $\MT_*$ & monoid-recognizable data languages \cite{boj2013} \\
\hline
\end{tabular}
\end{center}
 In the following, we focus on ($\omega$-)regular languages and cost functions; see \cite{uacm17_2,um19} for details on the remaining examples.
\begin{enumerate}
\item For the semigroup monad $\MT_+$ on $\Set$ we obtain the classical concept of algebraic language recognition: a language $L\seq I^+$ is recognizable if there exists a semigroup morphism $e\colon I^+\to S$ into a finite semigroup $S$ and a subset $P\seq S$ with  $L=e^{-1}[P]$. Recognizable languages are exactly the ($\epsilon$-free) regular languages \cite{pin15}. In fact, the expressive equivalence between $\Sigma$-automata in $\Set$ and semigroups generalizes to $\Sigma$-automata in symmetric monoidal closed categories \cite{amu15}.
\item Languages of infinite words can be captured algebraically as
  follows. A \emph{Wilke algebra} \cite{wilke91} is a two-sorted set
  $(S_+, S_\omega)$ with a product
  $\cdot\colon S_+\times S_+\to S_+$, a mixed product
  $\cdot\colon S_+\times S_\omega \to S_\omega$ and a unary operation
  $(\dash)^\omega\colon S_+\to S_\omega$ subject to the laws \[(st)u=s(tu),\; (st)z=s(tz),\;
  s(ts)^\omega = (st)^\omega,\; (s^n)^\omega = s^\omega,\] for all
  $s,t,u\in S_+$, $z\in S_\omega$ and $n>0$. The free Wilke algebra generated by the
  two-sorted set $(X,Y)$ is
  $T_\infty(X,Y)=(X^+,X^{\mathsf{up}}+X^*Y)$ with the two products given by concatenation of words, and $w^\omega=www\ldots$ for $w\in X^+$. In particular, choosing the input object
  $(I,\emptyset)$ for some set $I$ and the output object
  $O=(\{0,1\},\{0,1\})$, we have
  $T_\infty(I,\emptyset)=(I^+,I^{\mathsf{up}})$, and thus a language
  $L\colon T_\infty(I,\emptyset)\to O$ specifies a set of finite or
  ultimately periodic infinite words. Languages recognizable by
  Wilke algebras correspond to \emph{$\omega$-regular
    languages}, i.e.~languages accepted by Büchi automata
  \cite{wilke91,pinperrin04}.
\item \emph{Regular cost functions} were introduced by Colcombet~\cite{colcombet09} as a quantitative extension of regular languages that provides a unifying framework for studying limitedness problems. A \emph{cost function} over the alphabet $I$ is a function $f\colon I^*\to \Nat\cup\{\infty\}$. Two cost functions $f$ and $g$ are identified if, for every subset $A\seq\Nat$, the function $f$ is bounded on $A$ iff $g$ is bounded on $A$. Regular cost functions correspond to languages recognizable by finite \emph{stabilization algebras}. The latter are ordered algebras over the signature $\Gamma=\{1/0,\,\o/2,\,(\dash)^\#/1,\, (\dash)^\omega/1\}$, with $\dash/n$ denoting arities, subject to suitable inequations; see \cite{pin16,uacm17_2}. We let $\MT_S$ denote the monad on $\Pos$ induced by this ordered algebraic theory.
\end{enumerate}
\end{example}
\noindent Our generic approach to learning $\MT$-recognizable languages is based on the idea of presenting the free algebra $\MT I=(TI,\mu_I)$ and its finite quotient algebras as automata:
\begin{definition}[$\MT$-refinable]
A quotient $e\colon TI\epito A$ in $\D$ is \emph{$\MT$-refinable} if there exists a finite quotient algebra $e'\colon \MT I\epito (B,\beta)$ of $\MT I$ and a morphism $f\colon B\epito A$  with $e=f\o e'$. 
\end{definition}

\begin{definition}[Automata presentation]\label{def:autpres}
An \emph{automata presentation} of the free $\MT$-algebra $\MT I$ is given by an endofunctor $F$ on $\D$ and an $F$-algebra structure $\delta\colon FTI\to TI$ such that 
\begin{enumerate}
\item $F(\E)\seq \E$, the initial algebra $\mu F_I$ exists, and every regular language $L\colon \mu F_I\to O$ admits a minimal automaton $\Min{L}$;
\item the $F_I$-algebra $(TI,[\eta_I,\delta])$ is reachable (i.e. $e_{TI}\in \E$);
\item a $\MT$-refinable quotient $e\colon TI\epito A$ in $\D$ carries
  a $\MT$-algebra quotient iff $e$ carries an $F$-algebra quotient;
  that is, there exists $\alpha_A$ making the left-hand square below
  commute iff there exists $\delta_A$ making the right-hand square
  commute.
\end{enumerate}
\[
\vcenter{
\xymatrix@R-1em{
TTI \ar[r]^{\mu_I} \ar@{->>}[d]_{Te} & TI \ar@{->>}[d]^e \\
TA \ar@{-->}[r]_{\exists \alpha_A} & A
}
}
\qquad\Longleftrightarrow\qquad
\vcenter{
\xymatrix@R-1em{
FTI \ar[r]^\delta \ar@{->>}[d]_{Fe} &  TI \ar@{->>}[d]^{e} \\
FA \ar@{-->}[r]_{\exists \delta_{A}} & A
}
}
\]
If in (3) only the implication ``$\To$'' is required, $(F,\delta)$ is called a \emph{weak automata presentation}.
\end{definition}

\begin{remark} 
\begin{enumerate}
\item Examples of functors $F$ for which the first condition is satisfied include all functors satisfying the \Cref{asm}, see \Cref{rem:algcoalg}(1) and \Cref{thm:minaut}, and
  polynomial functors $F=F_\Gamma$ on $\Set$ or $\Pos$ for a signature $\Gamma$. Recall from \Cref{ex:treeaut} that $F_\Gamma$-automata are $\Gamma$-automata. 
\item Presentations of $\MT$-algebras as (sorted) $\Sigma$-automata were previously studied by Urbat, Ad\'amek, Chen, and Milius \cite{uacm17} for the special case where $\D$ is a variety of algebras and $\Sigma\in \D$ is a free algebra, and called \emph{unary presentations}.
\end{enumerate}
\end{remark}

\begin{example}\label{ex:unpres}
For all monads of \Cref{ex:monads}, free algebras admit an automata presentation (in fact, a presentation as (sorted) $\Sigma$-automata \cite{uacm17,uacm17_2,um19}). Here we consider three cases:
\begin{enumerate}
\item \label{ex:unpresmonoids}\emph{Semigroups.} The free semigroup $T_+I = I^+$ has a $\Sigma$-automata presentation $\delta\colon \Sigma\times I^+\to I^+$ given by the alphabet \[\Sigma = \{ \vec{a}\;:\; a\in I \} \cup \{ \vecr{a} \;:\; a\in I \}\] and the transitions
\[\delta(\vec{a},w)=wa\quad\text{and}\quad  \delta(\vecr{a},w)=aw \quad \text{for}\quad w\in I^+,\, a\in I.\]
Recall from \Cref{ex:automata_sorted} that $\mu F_I = I\times \Sigma^*$. The unique homomorphism $e_{I^+}\colon I\times \Sigma^*\to I^+$ interprets a word in $I\times \Sigma^*$ as a list of instructions for forming a word in $I^+$, e.g. \[e_{I^+}(a\vec{a}\vec{b}\vecr{b}\vec{a}) \;=\; baaba.\] For a weak automata presentation of $I^+$, it suffices to take the restriction $\delta'\colon \Sigma'\times I^+\to I^+$ of $\delta$ where $\Sigma'=\{ \vec{a}\;:\; a\in I \}$.
\item\label{ex:unpresomegasem} \emph{Wilke algebras}. The
free  Wilke algebra $T_\infty(I,\emptyset)=(I^+,I^{\mathsf{up}})$ can be
  presented as a two-sorted $\Sigma$-automaton with the sorted
  alphabet
  $\Sigma=(\Sigma_{+,+},\,\Sigma_{+,\omega},\,\Sigma_{\omega,\omega},\emptyset)$
   given by
  \begin{align*}
  	\Sigma_{+,+} &= \{ \vec{a}: a\in I \} \cup \{ \vecr{a} : a\in I \}\\
  \Sigma_{+,\omega}&=\{\omega\}\cup \{ \vec{v}^\omega : v\in I^+ \}\\
  \Sigma_{\omega,\omega}&=\{ \vecru{a}: a\in I \}
  \end{align*}
  and the transitions below, where  $v,w\in I^+$,
  $z\in I^{\mathsf{up}}$, $a\in I$:
  \begin{align*}
  	\delta_{+,+}(\vec{a},w)&=wa, & \delta_{+,+}(\vecr{a},w)&=aw,\\
  \delta_{+,\omega}(\omega,w)&=w^\omega,&
  \delta_{+,\omega}(\vec{v}^\omega,w)&=wv^\omega,\\
  \delta_{\omega,\omega}(\vecru{a},z)&= az. & &
  \end{align*}
Recall from \Cref{ex:automata_sorted} that the initial algebra $\mu F_I$ consists of sorted words over $\Sigma$ with an additional first letter from $I$. The homomorphism $e_{(I^+,I^{\mathsf{up}})}\colon \mu F_I\to (I^+,I^{\mathsf{up}})$ views such a word as an instruction for forming a word in $(I^+,I^{\mathsf{up}})$, e.g.
\[ e_{(I^+,I^{\mathsf{up}})}(a\vec{b}\vec{a}\omega\vecru{a}\vecru{a})\;=\;aa(aba)^\omega. \]    To obtain a weak automata
  presentation, it suffices to restrict $\Sigma_{+,+}$ and
  $\Sigma_{+,\omega}$ to the finite subalphabets
  $\Sigma_{+,+}'= \{ \vec{a}: a\in I \}$ and
  $\Sigma_{+,\omega}'=\{\omega\}$. A $\Sigma'$-automaton is
  similar to a \emph{family of DFAs}, a concept recently employed by
  Angluin and Fisman \cite{af16} for learning $\omega$-regular
  languages.
\item\label{ex:unprestrees} \emph{Stabilization algebras.} Suppose that $\MT$ is a monad on $\Set$ or $\Pos$ induced by a finitary signature $\Gamma$ and (in-)equations $E$; see \Cref{sec:preliminaries}. Then $\MT I$ can be presented as the $\Gamma$-automaton $\delta\colon F_{\Gamma}(TI) \to TI$ given by the $\Gamma$-algebra structure on the free $(\Gamma,E)$-algebra $TI$. The initial algebra $\mu (F_\Gamma)_I$ is the algebra $T_\Gamma I$ of $\Gamma$-terms over $I$, and the unique homomorphism $e_{TI}\colon T_\Gamma I\epito TI$ interprets  $\Gamma$-terms in $TI$. In particular, for the monad $\MT=\MT_S$ on $\Pos$, the free stabilization algebra $\MT_S I$ admits a $\Gamma$-automata presentation for the signature $\Gamma$ of \Cref{ex:reclang}(3).
\end{enumerate}
\end{example}

\noindent From now on, we fix a weak automata presentation
$(F,\delta)$ of the free $\MT$-algebra $\MT I$.

\begin{definition}[Linearization] The \emph{linearization} of a language $L\colon TI\to O$ is given by
\[ \lin{L} \;=\; (\xymatrix{\mu F_I \ar@{->>}[r]^{e_{TI}} & TI  \ar[r]^L & O}). \]
\end{definition}

\begin{example}\label{ex:linearization}
\begin{enumerate}
\item \emph{Semigroups.} Take the $\Sigma$-automata presentation of \Cref{ex:unpres}(1). Given $L\seq I^+$, the language
  $\lin{L}\seq I\times \Sigma^*$ consists of all possible ways of
  generating words in $L$ by starting with a letter $a\in I$ and
  adding letters on the left or on the right. For instance, if $L$
  contains the word $abc$, then $\lin{L}$ contains the words
  $a\vec{b}\vec{c},\, b\vecr{a}\vec{c},\,
  b\vec{c}\vecr{a},\,c\vecr{b}\vecr{a}$.
\item \emph{Wilke algebras.} Take the weak presentation of \Cref{ex:unpres}\ref{ex:unpresomegasem}. Given $L\seq (I^+,I^{\mathsf{up}})$, the language $\lin{L}$ consists of all possible ways of generating words in $L$ by starting with a letter $a\in I$ and repeatedly applying any of the following operations: (i) right concatenation of a finite word with a letter; (ii) left concatenation of an infinite word with a letter; (iii) taking the $\omega$-power of a finite word. For instance, if $L$ contains the word $(ab)^\omega$, then $\lin{L}$ contains
$a\vec{b}\omega$, $b\vec{a}\omega\vecru{a}$, $a\vec{b}\omega\vecru{b}\vecru{a}$, $b\vec{a}\omega\vecru{a}\vecru{b}\vecru{a},\ldots$. Thus, $\lin{L}$ is a two-sorted version of the language $\mathsf{lasso}(L)$ mentioned in the Introduction. 
\item \emph{Stabilization algebras.} Take the presentation of \Cref{ex:unpres}(3). Given a language $L\seq T_S I$, the set $\lin{L}\seq T_\Gamma I$ consists of all $\Gamma$-trees whose interpretation in $T_S I$ lies in $L$.
\end{enumerate}
\end{example}
\noindent As demonstrated by the above examples, the linearization allows us to
identify a language $L\colon TI\to O$ with a language $\lin{L}\colon \mu F_I\to O$ of finite
words or trees. Since the morphism $e_{TI}\colon \mu F_I\epito TI$ is assumed to be epic by \Cref{def:autpres}(2), this
identification is unique; that is, $\lin{L}$ uniquely determines
$L$. In particular, in order to learn $L$, it is sufficient to learn
$\lin{L}$. This approach is supported by the following result:

\begin{theorem}\label{thm:ltolinl}
  If $L\colon TI\to O$ is a $\MT$-recognizable language, then its linearization
  $\lin{L}\colon \mu F_I\to O$ is {regular}, i.e.\ accepted by
  some finite $F$-automaton.
\end{theorem}
\begin{proof}[Proof sketch]
Let $e\colon \MT I\to (A,\alpha)$ be a $\MT$-homomorphism recognizing $L$ via $p\colon A\to O$. By replacing $e$ with its coimage, we may assume that $e\in \E$. The weak automata presentation yields an $F$-algebra structure on $A$ making $e$ an $F$-algebra homomorphism. Then $A$, viewed as an automaton with initial states $e\o \eta_I\colon I\to A$ and final states $p$, accepts $\lin{L}$.
\end{proof}

\noindent In view of this theorem, one can apply any learning algorithm for finite $F$-automata (e.g.~ Generalized $\mathsf{L}^*$ for the case of adjoint automata, or a learning algorithm for tree automata \cite{dh03} if $F$ is a polynomial functor) to learn the
minimal automaton $Q_L$ for $\lin{L}$. This automaton, together with the
epimorphism $e_{TI}$, constitutes a finite representation of the
unknown language $L\colon TI\to O$. If the given automata presentation
for $\MT I$ is non-weak, we can go one step further and infer from
$Q_L$ a minimal algebraic representation of $L$:

\begin{definition}[Syntactic $\MT$-algebra]
  Let $L\colon TI\to O$ be recognizable. A \emph{syntactic
    $\MT$-algebra} for $L$ is a quotient $\MT$-algebra
  $e_L\colon \MT I\epito \Syn{L}$ of $\MT I$ such that (1) $e_L$
  recognizes~$L$, and (2) $e_L$ factorizes through every finite quotient
  $\MT$-algebra $e\colon \MT I \epito (A,\alpha)$ recognizing $L$.
\[
\xymatrix@R-1em{
\MT I \ar@{->>}[r]^e \ar@{->>}[dr]_{e_L} & (A,\alpha) \ar@{-->}[d]\\
& \Syn{L}
}
\]
\end{definition}

\begin{theorem}\label{thm:synalgconst}
Let $(F,\delta)$ be an automata presentation for $\MT I$. Then every $\MT$-recognizable language $L\colon TI\to O$ has a syntactic $\MT$-algebra $\Syn{L}$, and its corresponding $F$-automaton (via the given presentation) is the minimal automaton for $\lin{L}$:
\[ \Syn{L}\;\cong\; \Min{\lin{L}}. \]
\end{theorem}
\noindent This theorem asserts that we can uniquely equip the learned minimal $F$-automaton $Q_L=\Min{\lin{L}}$ with a $\MT$-algebra structure $\alpha_L\colon TQ_L\to Q_L$ for which the unique automata homomorphism $e_L\colon TI\epito Q_L$ is a $\MT$-algebra homomorphism $e_L\colon \MT I\epito (Q_L,\alpha_L)$. Then $e_L$ is the syntactic algebra for $L$.

\begin{remark}
To make the construction of $\Syn{L}$ from the learned automaton $Q_L$ effective, we need to assume that the morphisms $e_{Q_L}$, $e_{TI}$, $Te_{Q_L}$, $Te_{TI}$ and $\mu_I$ can be represented as (sorted families of) computable maps and moreover the maps $e_{TI}$ and $Te_{Q_L}$ admit computable (not necessarily morphic) right inverses $m$ and $n$, respectively. Then the $\MT$-algebra structure $\alpha_L$ of $\Syn{L}$ can be represented as the computable map $e_{Q_L}\o m\o \mu_I\o Te_{TI}\o n$; see the commutative diagram below.
\[
\xymatrix@R-1em{
T(\mu F_I) \ar@{->>}[r]^{Te_{TI}} \ar@{->>}[dr]_{Te_{Q_L}} & TTI \ar@{->>}[d]^{Te_L} \ar[r]^{\mu_I} & TI \ar@{->>}[d]^{e_L} & \mu F_I \ar@{->>}[l]_{e_{TI}} \ar@{->>}[dl]^{e_{Q_L}}\\
& TQ_L \ar[r]_{\alpha_L} & Q_L & 
}
\]
\end{remark}

\begin{example}
This computation strategy works for all monads of \Cref{ex:monads}. We consider our running examples:
\begin{enumerate}
\item \emph{Semigroups.} For the $\Sigma$-automata presentation of \Cref{ex:unpres}\ref{ex:unpresmonoids} and $L\seq I^+$, we compute the semigroup structure $\bullet\colon Q_L\times Q_L\to Q_L$ on $Q_L$ from its automaton structure as follows. Given $q,q'\in Q_L$ choose words $w,w'\in I\times \Sigma^*$  with $e_{Q_L}(w) = q$, $e_{Q_L}(w')=q'$\lsnote{Corrected}, i.e.~witnesses for the reachability of $q$ and $q'$. Next, choose $v\in I\times \Sigma^*$ with $e_{I^+}(v)=e_{I^+}(w)e_{I^+}(w')\in I^+$, and put $q\bullet q' := e_{Q_L}(v)$.
\item \emph{Wilke algebras.} Analogous to the case of semigroups.
\item \emph{Cost functions.} For a monad $\MT$ on $\Set$ or $\Pos$ given by a signature $\Gamma$ and (in-)equations $E$ and the $\Gamma$-automata presentation of $\MT I$ in \Cref{ex:unpres}\ref{ex:unprestrees}, the computation of $\alpha_L$ is trivial: the structure of the $\Gamma$-algebra $\Syn{L}$ is just the automaton structure of $Q_L$. In particular, this applies to the monad $\MT_S$ on $\Pos$ representing cost functions (\Cref{ex:monads}(3)). Thus, we obtain the first learning algorithm for this class of languages.
\end{enumerate} 
\end{example}

\section{Conclusions and Future Work}
We have presented a generic algorithm (Generalized $\mathsf{L}^*$) for learning
$F$-automata that forms a uniform abstraction of 
$\mathsf{L}^*$-type algorithms, their correctness proofs, and parts of their complexity analysis, and instantiates to several new learning algorithms, e.g.~for various notions of nominal automata with name binding. Moreover, we have shown how to extend the scope of Generalized $\mathsf{L}^*$, and other learning algorithms for $F$-automata, to languages recognizable by monad algebras. This gives rise to a generic approach to learning numerous types of languages, including cases for which no learning algorithms are known (e.g.~cost functions).

The next step is to turn our high-level categorical
approach into an implementation-level algorithm, parametric in the monad $\MT$ and its automata
presentation, with corresponding tool support. We expect that the recent work on coalgebraic minimization algorithms and their implementation \cite{dmsw17,dmsw19} can provide guidance. It should be illuminating to experimentally compare the performance of the generic algorithm with tailor-made algorithms for specific types of automata.

Our generalized $\mathsf{L}^*$ algorithm is concerned with adjoint $F$-automata and applies to a wide variety of automata on finite words (including weighted, residual nondeterministic, and nominal automata), but presently not to tree automata. To deal with the latter, the adjointness of the type functor $F$ needs to be relaxed, which entails that a coalgebraic semantics is no longer directly available. A categorical approach to learning tree automata, assuming a purely algebraic point of view, was recently proposed by van Heerdt et al~\cite{heerdt2020}. The subtle interplay between the algebraic and coalgebraic aspects underlying learning algorithms is up for further investigation.


\bibliography{refs,coalgebra,ourpapers}


\clearpage
\appendix

\section{Appendix: Omitted Proofs and Details}
In this appendix, we provide full proofs of all our results and
more detailed treatment of some examples omitted due to space restrictions.


\section*{Discussion of the \Cref{asm} and \ref{asm2}}
We comment on some technical consequences of our \Cref{asm} and \ref{asm2}.

\begin{remark}\label{rem:autfact}
The assumption $F(\E)\seq \E$ implies that the factorization system $(\E,\M)$ of $\D$ lifts to automata: given an automata homomorphism $h\colon Q\to Q'$ and its $(\E,\M)$-factorization $\xymatrix{h=(\, Q\ar@{->>}[r]^e & Q'' \ar@{>->}[r]^m & Q' \,)}$ in $\D$, there exists a unique automata structure $(Q'',\delta_{Q''}, i_{Q''}, f_{Q''})$ on $Q''$ such that both $e$ and $m$ are automata homomorphisms. Indeed, the transitions $\delta_Q''$ are given by diagonal fill-in
\[
\xymatrix{
FQ \ar[r]^{\delta_Q} \ar@{->>}[d]_{Fe} & Q \ar@{->>}[d]^e \\
FQ'' \ar@{-->}[r]^{\delta_{Q''}} \ar[d]_{Fm} & Q'' \ar@{>->}[d]^m \\
FQ' \ar[r]_{\delta_{Q'}} & Q'
}
\]
and the initial and final states by
\begin{eqnarray*}
i_{Q''} &=& (\, I \xra{i_Q} Q \xra{e} Q'' \,),\\
f_{Q''} &=& (\, Q''\xra{m} Q' \xra{f_{Q''}} O\,).
\end{eqnarray*}
\end{remark}

\begin{remark}
The condition $F_I(\M)\seq \M$ makes sure that the factorization system $(\E,\M)$ lifts from $\D$ to $\Coalg{F_I}$, the category of $F_I$-coalgebras: given an $F_I$-coalgebra homomorphism $h\colon (C,\gamma)\to (C',\gamma')$ and its $(\E,\M)$-factorization $\xymatrix{h=(\, C\ar@{->>}[r]^e & C'' \ar@{>->}[r]^m & C' \,)}$ in $\D$, there is a unique $F_I$-coalgebra structure $(C'',\gamma'')$ on $C''$ such that both $e$ and $m$ are coalgebra homomorphisms. The structure $\gamma''$ is defined via diagonal fill-in in analogy to \Cref{rem:autfact}.

Dually, the condition $G_O(\E)\seq \E$ implies that $\Alg{G_O}$, the category of $G_O$-algebras, has a factorization system lifting $(\E,\M)$.
\end{remark}

\section*{Details for \Cref{ex:categories}}
We show that for each of the five categories $\D$ of \Cref{table:categories} and the endofunctors $F$ and $G$ on $\D$ given by
\[ F=\Sigma\t (\dash)\quad\text{and}\quad G=[\Sigma,\dash],  \]
the \Cref{asm}(1)--(4) and \ref{asm2} are satisfied. 

Clearly, all the categories $\D$ with the corresponding choices of $I$ and $O$ satisfy the \Cref{asm}(1)(2). Moreover, (3) holds because $\D$ is closed. For (4), note that in all cases $\E$ coincides with the class of all epimorphisms. Since every left adjoint $F$ preserves epimorphisms, it follows that $F(\E)\seq \E$. It remains to verify the \Cref{asm2}. We consider the cases $\D=\Set$, $\Pos$, $\JSL$, $\Vect{\K}$;  for $\D=\Nom$, see the details for \Cref{ex:nom-aut}.

\paragraph{$F_I$ preserves $\M$ and intersections of $\M$-morphisms.} This is clear for $\D=\Set,\Pos$ since in these categories coproducts commute with intersections, i.e. one has
\[ (A+B)\cap (C+D) \cong (A\cap C) + (B\cap D). \]

 For $\D=\JSL$ recall that we have chosen $\Sigma$ to be the free semilattice $\Pow_f\Sigma_0$ over a finite set $\Sigma_0$ of generators, i.e. the $\cup$-semilattice of finite subsets of $\Sigma_0$. It follows that
\[ F_I X = I+\Sigma\t X = I + (\coprod_{a\in \Sigma_0} I)\t X \cong I+\coprod_{a\in \Sigma_0} I \t X \cong I+ \coprod_{a\in \Sigma_0} X \]
using that $I=\Pow_f 1$, $I\t X\cong X$, and the left adjoint $(\dash)\t X$ preserves coproducts. Now note that the coproduct $X+Y$ of two semilattices coincides with the product $X\times Y$, with injections given by
\begin{eqnarray*} 
\inl\colon X\to X\times Y,&\quad x\mapsto (x,\bot) \\
  \inr\colon Y\to X\times Y,&\quad Y\mapsto (\bot,y)
\end{eqnarray*}
This implies that monomorphisms in $\JSL$ are stable under coproducts, and that intersections commute with coproducts.  It thus follows from the above formula for $F_I X$ that $F_I$ preserves monomorphisms and intersections.

For $\D=\Vect{\K}$, the proof is analogous, using again the product/coproduct coincidence.

\paragraph{$G_O$ preserves epimorphims.} 
We first show that the functor $[\Sigma,\dash]$ preserves epimorphisms (i.e.~surjections). 
Note first that in $\D=\Set,\Pos,\JSL,\Vect{\K}$, the object $[\Sigma,X]$ is carried by the set $\D(\Sigma,X)$ with the $\D$-structure inherited from $X$ (i.e.~ defined pointwise), and that for any morphism $e\colon X\to Y$ the morphism $[\Sigma,e]\colon [\Sigma,X]\to [\Sigma,Y]$ is given by $f\mapsto e\o f$. We need to prove that $[\Sigma,e]$ is surjective provided that $e$ is surjective; that is, for every morphism $g\colon \Sigma\to Y$ there exists a morphism $f\colon \Sigma\to X$ making the following triangle commute:
\[
\xymatrix@R-1em{
\Sigma \ar@{-->}[r]^f \ar[dr]_g & X \ar@{->>}[d]^e \\
& Y
}
\]
This follows from the fact that in each case, $\Sigma$ has been chosen as a projective object of $\D$. For instance, for $\D=\JSL$ we construct $f$ as follows. Recall that $\Sigma$ is the free semilattice on a finite set $\Sigma_0$, and denote by $\eta\colon \Sigma_0\to \Sigma$ the universal map. For each $a\in \Sigma_0$, choose $x_a\in X$ with $e(x_a)=g(\eta(a))$, using that $e$ is surjective. This gives a map
\[
f_0\colon \Sigma_0\to X,\quad a\mapsto x_a.
\]
Let $f\colon \Sigma\to X$ be the unique semilattice homomorphism extending $f_0$, i.e. with $f\o \eta = f_0$. Then $e\o f=g$ since this equation holds when precomposed with the universal map $\eta$, as shown by the diagram below:
\[
\xymatrix@R-1em{
\Sigma_0 \ar[r]^{\eta} \ar@/^2em/[rr]^{f_0} & \Sigma \ar@{-->}[r]^f \ar[dr]_g & X \ar@{->>}[d]^e \\
& & Y
}
\]
This shows that the functor $[\Sigma,\dash]$ preserves epimorphisms.
Since epimorphisms in our categories $\D$ are stable under products, it follows that also the functor $G_O = O\times [\Sigma,\dash]$ preserves epimorphisms.

\section*{Details for \Cref{ex:automata}}
\begin{enumerate}
	\item The functor $P=[\dash,O]\colon \D\to\D^\op$ is a left adjoint (with right adjoint $P^{\op}\colon \D^{\op}\to \D$) because, for each $X,Y\in \D$, 
	\begin{align*} \D(X,P Y) &= \D(X,[Y,O])\\
& \cong \D(X\t Y,O) \\&\cong \D(Y\t X,O) \\& \cong \D(Y,[X,O])\\& = \D(Y,PX). \end{align*}
\item We have a natural isomorphism 
\[PF_I \cong G_O^{\op} P.\]
To see this, observe that all parts of the following diagram commute up to isomorphism.
\[
\xymatrix{
	\D \ar@/_12ex/[dd]_{F_I} \ar[d]_F \ar[r]^{P} & \D^{\op} \ar[d]^{G^{\op}} \ar@/^12ex/[dd]^{G_O^{\op}} \\
	\D \ar[r]_{P} \ar[d]_{I+\dash} & \D^{\op} \ar[d]^{(O\times \dash)^{\op}} \\
	\D \ar[r]_{P} &  \D^{\op}
}
\]
The left and right parts commute by definition. The two squares commute because for each $X\in \D$,
\[ PFX = [\Sigma\t X,O] \cong [\Sigma, [X,O]] = GPX\]
and \[ P(I+X) \cong PI\times PX \cong  [I_\D,O]\times PX \cong O\times PX. \]
The isomorphism $P(I+X)\cong PI\times PX$ uses that $P$ is a left adjoint, i.e. preserves coproducts.
\end{enumerate}

\section*{Details for \Cref{ex:nom-aut}}\label{ex:nomautdetails}
We verify that the functors of \Cref{ex:nom-aut}(1)--(4), see the table below, satisfy our \Cref{asm}(4) and \ref{asm2}. Recall that we have chosen
$I = 1$ and $O=2$,
and that the factorization system of $\Nom$ is the one given by epimorphisms (= surjective equivariant maps) and monomorphisms (= injective equivariant maps). 
\begin{center}
\begin{tabular}[ht]{|l l l|}
\hline
 & $F$ & $G$ \\
\hline
(1) & $\At\times (\dash)$ & $[\At,\dash]$ \\
(2) & $\At\ast (\dash)$ & $[\At](\dash)$ \\
(3) & $\At\times (\dash) + \At\ast (\dash)$ & $[\At,\dash]\times [\At](\dash)$ \\
(4) & $\At\times (\dash) + \At\ast (\dash) + [\At](\dash)$ &  $[\At,\dash]\times [\At](\dash) \times R$ \\
\hline
\end{tabular}
\end{center}

\paragraph{$F$ preserves epimorphisms.} This follows from the fact that $F$ is a left adjoint.

\paragraph{$F_I$ preserves monomorphisms.} The functors $\At\times (\dash)$ and $\At\ast (\dash)$  preserve monomorphisms by definition, recalling that for an equivariant map $e\colon X\to Y$ the map $\At\ast e$ is given by
\[ \At\ast e\colon \At\ast X\to \At\ast Y, \quad (a,x) \mapsto (a,e(x)).  \]
The functor $[\At](\dash)$ preserves monomorphisms because it is a right adjoint. Since coproducts in $\Nom$ are formed at the level of $\Set$, it follows that monomorphisms in $\Nom$ are stable under coproducts. This implies that for all the functors $F$ in (1)--(4), the functor $F_I = I+F$ preserves monomorphisms.

\paragraph{$F_I$ preserves intersections.} Note that intersections of subobjects (i.e.~equivariant subsets) in $\Nom$ are just set-theoretic intersections. Thus, the functors $\At\times (\dash)$ and $\At\ast (\dash)$ clearly preserve intersections by definition. The functor $[\At](\dash)$ preserves them because it is right adjoint and thus preserves all limits. Since intersections commute with coproducts in $\Set$ and thus also in $\Nom$, it follows that for all the functors $F$ in (1)--(4), the functor $F_I=I+F$ preserves intersections.

\paragraph{$G_O$ preserves epimorphisms.}
The functor $[\At](\dash)$ preserves epimorphisms because it is a left adjoint. Moreover, we have

\begin{lemma}\label{lem:hompresepi}
The functors $[\At,\dash]\colon \Nom\to \Nom$ and $R\colon \Nom\to \Nom$ preserve epimorphisms.
\end{lemma}

\begin{proof}
\begin{enumerate}
\item We first show that $[\At,\dash]$ preserves epimorphisms (i.e.~surjections). This can be deduced from the fact that every polynomial functor on $\Nom$ preserves epimorphisms (like in $\Set$) and that $[\At,\dash]$ can be expressed as a quotient functor of a polynomial functor \cite[Lemma 6.9]{msw16}. In the following, we give a direct proof for the convenience of the reader.

 Recall from \cite[Theorem 2.19]{pitts2013} that $[\At,X]$ is the nominal set of finitely supported maps $f\colon \At\to X$; here $f$ is \emph{finitely supported} if there exists a finite subset $S\seq \At$ such that for all permutations $\pi\in \Perm(\At)$ that fix $S$ and all $a\in \At$ one has
 $f(\pi\o a) = \pi\o f(a)$. In particular, equivariant maps are finitely supported maps with support $S=\emptyset$. For any equivariant map $e\colon X\to Y$, the map $[\At,e]$ is given by
\[ [\At,e]\colon [\At,X]\to [\At, Y],\quad f\mapsto e\o f.\]
We need to show that $[\At,e]$ is surjective provided that $e$ is surjective; in other words, for every finitely supported map $g\colon \At\to Y$, there exists a finitely supported map $f\colon \At\to X$ making the following triangle commute:
\[
\xymatrix@R-1em{
\At \ar@{-->}[r]^f \ar[dr]_g & X \ar@{->>}[d]^e \\
& Y
}
\]
Fix an arbitrary atom $a\not \in \At\setminus \supp g$. Moreover, choose $x\in X$ with $e(x)=g(a)$, and choose $x_b\in X$ with $e(x_b)=g(b)$ for every $b\in \supp g\cup\supp x$,  using that $e$ is surjective. Define the map $f\colon \At\to X$ as follows:
\[ 
f(b) = 
\begin{cases}
(b\, a) \o x & \text{for $b\in \At\setminus (\supp g\cup \supp x)$};\\
x_b & \text{for $b\in \supp g \cup \supp x$}. 
\end{cases}
\]
We claim that (i) the map $f$ is finitely supported and (ii) it satisfies $e\o f=g$.

\medskip\noindent \emph{Ad (i).} We show that the finite set of atoms
\[S\;=\;\supp g \,\cup\,\supp x \,\cup\, \bigcup_{b\in \supp g\,\cup\, \supp x} \supp x_b\] supports the map $f$. Thus, let $\pi\in \Perm(\At)$ be a permutation fixing $S$;  we need to prove that $f(\pi\o b)=\pi \o f(b)$ for all $b\in \At$. For $b\in \supp g\cup \supp x$, we have
\[ f(\pi \o b) = f(b) = x_b = \pi \o x_b = \pi\o f(b). \] 
For $b\in \At\setminus (\supp g\cup \supp x)$, we get
\[
f(\pi \o b) = (\pi(b)\, a)\o x = \pi \o (b\, a)\o x = \pi\o f(b).
\]
Here the first and last equation use the definition of $f$. 
The middle equation holds because the two permutations $(\pi(b)\, a)$ and $\pi \o (b\,a)$ are equal on $\supp x$. Indeed, both permutations send $a$ to $\pi(b)$, and all elements of $\supp x\setminus \{a\}$ are fixed by both permutations because $b,\pi(b)\not\in \supp x$ and $\pi$ fixes $\supp x$.

\medskip\noindent \emph{Ad (ii).} We show that $e(f(b))=g(b)$ for all $b\in \At$. For $b\in \supp g\cup \supp x$ we have
\[ e(f(b)) = e(x_b) = g(b) \]
by definition of $f$ and $x_b$. For $b\in \At\setminus(\supp g\cup \supp x)$,
\begin{align*} e(f(b)) &= e((b\, a)\o x) & \text{def. $f$} \\
& = (b\,a)\o e(x) & \text{$e$ equivariant} \\
&= (b\, a)\o g(a) & \text{def. $x$} \\
& = g((b\, a)\o a) & \text{$a,b\not\in \supp g$} \\
& = g(b) &. 
\end{align*}
\item We show that $R$ preserves surjections. Recall that $R$ is the subfunctor of $[\At,\dash]$ given by
\[ RX = \{\, f\in [\At,X] \;:\; \text{$a\,\#\,f(a)$ for every $a\in \At$}\,\}. \]
We need to show that $Re\colon RX\to RY$ is surjective for every surjective equivariant map $e\colon X\epito Y$; that is, for every $g\in RY$, there exists $f\in RX$ with $e\o f=g$.

The definition of $f$ is the same as in part (1) of the proof, except that the elements $x$ and $x_b$ ($b\in \supp g\cup \supp x$) are now additionally required to satisfy $a\# x$ and $b\# x_b$. Such a choice of $x$ and $x_b$ is always possible: if $x$ is any element of $X$ with $e(x)=g(a)$, choose $a'$ with $a'\# g(a),x$ and put $x'=(a'\, a)\o x$. Then $a\#x'$ and \[e(x')=e((a'\,a)\o x) = (a'\, a)\o e(x) = (a'\, a)\o g(a) = g(a),\] where the last equation uses that $a,a'\# g(a)$. Thus, we can replace $x$ by $x'$. Analogously for $x_b$. 

Part (1) now shows that $f$ is finitely supported and satisfies $e\o f=g$. Moreover, we clearly have $b\# f(b)$ for every $b\in \At$ by definition of $f$ and the above choices of $x$ and $x_b$, i.e.~$f\in RX$.\qedhere
\end{enumerate}
\end{proof}
Since epimorphisms in $\Nom$ are stable under products (which follows from the corresponding property in $\Set$), we conclude that for all the functors $G$ in (1)--(4), the functor $G_O=2\times G$ preserves epimorphisms.

\section*{Details for \Cref{ex:automata_sorted}}
We describe sorted $\Sigma$-automata for the case of general base categories $\D$. Suppose that $(\D,\t,I_\D)$ is a symmetric monoidal closed category satisfying our \Cref{asm}\ref{A1}--\ref{A3}, and let $S$ be a set of sorts. Then the category $\D^S$ (equipped with the monoidal structure and the factorization system inherited sortwise from $\D$) is also symmetric monoidal closed and satisfies the \Cref{asm}\ref{A1}--\ref{A3}.

 Fix an arbitrary object $I\in \D^S$ inputs (not necessarily the tensor unit), an arbitrary object $O\in \D^S$ of outputs, and a family of objects $\Sigma=(\Sigma_{s,t})_{s,t\in S}$ in $\D$; we think of $\Sigma_{s,t}$ as a set of letters with input sort $s$ and output sort $t$. Take the functors
\[ F\colon \D^S\to \D^S,\qquad (FQ)_t = \coprod_{s\in S} \Sigma_{s,t}\t Q_s\quad (t\in S), \]

\[ G\colon \D^S\to \D^S,\qquad (GQ)_s = \prod_{t\in S} [\Sigma_{s,t},Q_t] \quad (s\in S). \]
The functor $F$ is a left adjoint of $G$: we have the isomorphisms (natural in $P,Q\in \D^S$)
\begin{align*}
	\D^S(FQ,P) &=\prod_{t\in S} \D((FQ)_t,P_t) \\
	& = \prod_{t\in S} \D(\coprod_{s\in S} \Sigma_{s,t}\t Q_s,P_t) \\
	& \cong \prod_{t\in S} \prod_{s\in S} \D(\Sigma_{s,t}\t Q_s,P_t) \\
	& \cong \prod_{s\in S} \prod_{t\in S} \D(\Sigma_{s,t}\t Q_s,P_t) \\
	& \cong \prod_{s\in S} \prod_{t\in S} \D(Q_s,[\Sigma_{s,t},P_t]) \\
	& \cong \prod_{s\in S} \D(Q_s, \prod_{t\in S} [\Sigma_{s,t},P_t]) \\
	& = \prod_{s\in S} \D(Q_s, (GP)_s). \\
	&= \D^S(Q,GP)
\end{align*}
Instantiating \Cref{def:automaton} to the above data, we obtain the concept of a \emph{sorted $\Sigma$-automaton}. It is given by an $S$-sorted object of states $Q\in \D^S$ together with morphisms $\delta_{Q,s,t}$, $i_{Q,t}$ and $f_{Q,t}$ as in the diagram below for $s,t\in S$:
\[ 
\xymatrix{
& \Sigma_{s,t}\t Q_t \ar[d]^{\delta_{Q,s,t}} & \\
I_t \ar[r]_{i_{Q,t}} & Q_t \ar[r]_{f_{Q,t}} & O_t
}
 \] 
In generalization of the single-sorted case (see \Cref{ex:automata}), the initial algebra for $F_I$ can be described as follows. For $n\in \Nat$ and $s,t\in S$ define the object $\Sigma_{s,t}^n\in \D$ inductively by
\[ \Sigma_{s,t}^0 = I_\D, \quad \Sigma_{s,t}^{n+1} = \coprod_{r\in S} \Sigma_{s,r}\t \Sigma_{r,t}^n.  \]
and put
\[ \Sigma_{s,t}^* = \coprod_{n\in \Nat} \Sigma_{s,t}^n. \]
The initial algebra for the functor $F_I$ is given by
\[  
(\mu F_I)_t = \coprod_{s\in S} I_s\t \Sigma_{s,t}^*\quad (t\in S).
\]


\section*{Proof of \Cref{thm:minaut}}
We first establish some basic observations about automata homomorphisms and languages:
\begin{proposition}\label{prop:hom_pres_language}
	For each automata homomorphism $h\colon Q\to Q'$ one has 
	$L_Q = L_{Q'}$
\end{proposition}

\begin{proof}
	This follows from the commutative diagram below. The upper triangle commutes by initiality of $\mu F_I$, and all remaining parts commute by definition.
	\[
	\xymatrix@C-2em{
		& \ar@/_5em/[dd]_{L_Q} \ar@/^5em/[dd]^{L_{Q'}} \ar[dl]_{e_Q} \ar[dr]^{e_{Q'}} \mu F_I & \\
		Q \ar[dr]_{f_Q} \ar[rr]^h & & Q' \ar[dl]^{f_{Q'}} \\
		& O &  
	}
	\]
\end{proof}

\begin{remark}\label{rem:fghom}
Every $F$-algebra homomorphism $h\colon (Q,\delta)\to (Q',\delta')$ is also a $G$-coalgebra homomorphism $h\colon (Q,\delta^@)\to (Q',(\delta')^@)$, and vice versa. Indeed, the corresponding commutative squares  
are just adjoint transposes of each other.
\[
\xymatrix{
FQ \ar[r]^\delta \ar[d]_{Fh} & Q \ar[d]^h \\
FQ' \ar[r]_{\delta'} & Q'
}
\qquad
\xymatrix{
Q \ar[r]^{\delta^@} \ar[d]_h & GQ \ar[d]^{Gh} \\
Q' \ar[r]_{(\delta')^@} & GQ'
}
\]
\end{remark}

\begin{proposition}\label{prop:lang_vs_behavior}
	For all automata $Q$ and $Q'$, we have
	\[ L_Q=L_{Q'} \quad\text{iff}\quad m_Q\o e_Q = m_{Q'}\o e_{Q'}. \]
\end{proposition}

\begin{proof}
	\begin{enumerate}
		\item For the ``if'' direction, suppose that $m_Q\o e_Q = m_{Q'}\o e_{Q'}$. Then the following diagram (where $\outl\colon G_O=O\times G\to O$ denotes the left product projection) commutes by the definition of $\gamma_Q$ in \Cref{rem:algcoalg} and because $m_Q$ is a $G_O$-coalgebra homomorphism.
		\begin{equation}\label{eq:fq_vs_outl}
		\xymatrix{
			Q \ar[r]^{\gamma_Q} \ar[d]_{m_Q} \ar@/^5ex/[rr]^{f_Q} & G_O Q \ar[d]_{G_O m_Q} \ar[r]^{\outl} & O \\
			\nu G_O \ar[r]_<<<<<{\gamma} & G_O(\nu G_O) \ar[ur]_{\outl}
		}
		\end{equation}
		Thus $f_Q=\outl\o \gamma \o m_Q$ and analogously $f_{Q'}=\outl\o \gamma\o m_{Q'}$. This implies
		\[ L_Q = f_Q\o e_Q = \outl\o \gamma \o m_Q\o e_Q = \outl \o \gamma \o m_{Q'}\o e_{Q'} = \cdots = L_{Q'}. \]
\item For the ``only if'' direction, suppose that $L:=L_Q=L_{Q'}$. By equipping $\mu F_I$ with final states $L\colon \mu F_I\to O$, we can view $\mu F_I$ as a $G_O$-coalgebra, and thus $e_Q\colon \mu F_I\to Q$ as a $G_O$-coalgebra homomorphism (see \Cref{rem:fghom}). It follows that $m_Q\o e_Q\colon \mu F_I\to \nu G_O$ is a $G_O$-coalgebra homomorphism. Analogously, $m_{Q'}\o e_{Q'}$ is a coalgebra homomorphism. Thus, $m_Q\o e_Q = m_{Q'}\o e_{Q'}$ by finality of $\nu G_O$.\qedhere
	\end{enumerate}
\end{proof}

\begin{remark}\label{rem:laccept}
For every language $L\colon \mu F_I\to O$ there exists an automaton $Q$ accepting $L$. Indeed, one can choose $Q=\mu F_I$ with output morphism $L\colon \mu F_I\to O$.	
\end{remark}
We are prepared to prove the minimization theorem:

\begin{proof}[Proof of \Cref{thm:minaut}] Fix an arbitrary automaton $Q$ with $L_Q=L$ (see \Cref{rem:laccept}). Viewing $\mu F_I$ as an automaton with output morphism $L_Q = f_Q\o e_Q\colon \mu F_I\to O$, the unique $F_I$-algebra homomorphism $e_Q$ is an automata homomorphism. Analogously, equipping $\nu G_O$ with the initial states $m_Q\o i_Q\colon I\to \nu G_O$ makes the unique $G_O$-coalgebra homomorphism $m_Q\colon Q\to \nu G_O$ an automata homomorphism. Thus $m_Q\o e_Q$ is an automata homomorphism. Form its $(\E,\M)$-factorization, see \Cref{rem:autfact}:
\[
\xymatrix{
& \mu F_I \ar@{->}[dl]_{e_Q} \ar@{->>}[dr]^{e_{\Min{L}}} & \\
Q \ar[dr]_{m_Q} & & \Min{L} \ar@{>->}[dl]^{m_{\Min{L}}} \\
& \nu G_O  &  
}
\]
We claim that $\Min{L}$ is the minimal automaton for $L$. To this end, note first that $L_{\Min{L}} = L_Q=L$ by the ``if'' direction of \Cref{prop:lang_vs_behavior}. Thus, $\Min{L}$ accepts the language $L$. Moreover, $\Min{L}$ is reachable because $e_{\Min{L}}\in \E$.

To establish the universal property of $\Min{L}$, suppose that $R$ is a reachable automaton accepting $L$; we need to show that there is a unique homomorphism from $R$ into $\Min{L}$. From $L_{\Min{L}}=L_R=L$ it follows that $m_R\o e_R = m_{\Min{L}}\o e_{\Min{L}}$ by the ``only if'' direction of \Cref{prop:lang_vs_behavior}. Thus, diagonal fill-in yields a unique automata homomorphism $h\colon R\to \Min{L}$ making the diagram below commute:
\[
\xymatrix{
& \mu F_I \ar@{->>}[dl]_{e_R} \ar@{->>}[dr]^{e_{\Min{L}}} & \\
R \ar[dr]_{m_R} \ar@{-->}[rr]^h & & \Min{L} \ar@{>->}[dl]^{m_{\Min{L}}} \\
& \nu G_O  &  
}
\]
Given another automata homomorphism $h'\colon R\epito \Min{L}$, we have $h'\o e_R=e_{\Min{L}}$ by initiality of $\mu F_I$. Thus $h'\o e_R = h\o e_R$, which implies $h'=h$ because $e_R$ is epic. This proves the desired universal property of $\Min{L}$. 

The uniqueness of $\Min{L}$ up to isomorphism follows immediately from its universal property.
\end{proof}
The construction of $\Min{L}$ is the above proof also shows:
\begin{corollary}\label{cor:minreachsimp}
An automaton $Q$ is minimal if and only if it is both reachable ($e_Q\in \E)$ and simple ($m_Q\in \M$).
\end{corollary}
\section*{Details for \Cref{rem:hst}}
That $L_Q=L_{Q'}$ implies $h_{s,t}^Q=h_{s,t}^{Q'}$ follows immediately from the ``only if'' direction of \Cref{prop:lang_vs_behavior} and the definition of $h_{s,t}^{(\dash)}$.

\section*{Details for \Cref{def:hypothesis}}
For the diagonal fill-in $\delta_{s,t}$ to exist, we need to verify that for each pair $(s,t)$ as in \eqref{eq:st}, the square below is commutative:
\[
\xymatrix{
	FS \ar[d]_{l_{s,t}} \ar@{->>}[r]^{Fe_{s,t}} & FH_{s,t}  \ar[d]^{r_{s,t}^\#}  \\
	H_{s,t}\ar@{>->}[r]_{m_{s,t}} & T
}
\]
where
\[  l_{s,t} \;=\; (FS \xra{\inr} I+FS=F_I S \xra{e_{F_Is,t}} H_{F_I s,t} \xra{\cl_{s,t}^{-1}} H_{s,t} ) \]
and
\[  r_{s,t} \;=\; (H_{s,t} \xra{\cs_{s,t}^{-1}} H_{s,G_O t} \xra{m_{s,G_O t}} G_O T=O\times GT  \xra{\outr} GT).  \]
\begin{proof}
By definition of $\cl_{s,t}$ and $\cs_{s,t}$, the lower path of the square is equal to
\[ FS \xra{\inr} F_IS \xra{h_{F_Is,t}} T \]
and the upper path is equal to
\[ FS \xra{Fh_{s,G_O t}} FG_O T \xra{\outr^\#} T. \]
We therefore need to verify that the outside of the following diagram commutes:
\[
\xymatrix@C-2em{
FS \ar[rrrrrr]^{Fh_{s,G_Ot}} \ar[dr]^{Fs} \ar[dddd]_{\inr} & & & & & & FG_O T \ar[dddd]^{\outr^\#} \\
& FF_I^N0 \ar[r]^{Fj_N} \ar[dd]_\inr & F(\mu F_I) \ar[r]^{Fe_Q} \ar[d]_\inr & FQ \ar[d]_\inr \ar[r]^{Fm_Q} & F(\nu G_O) \ar[r]^{Fj_{K+1}'}  \ar[d]^{F\gamma} & FG_O^{K+1}1 \ar@{}[dd]|{(\ast)} \ar[ur]^{FG_Ot} & \\
& & F_I(\mu F_I) \ar[d]_\alpha \ar[r]_{F_Ie_Q} & F_I Q \ar[d]_{\alpha_Q}  & FG_O(\nu G_O) \ar[d]^{\outr^\#} & & \\
& F_I^{N+1}0 \ar[ur]^{F_I j_N} \ar[r]_{j_{N+1}} & \mu F_I \ar[r]_{e_Q} & Q \ar[r]_{m_Q}& \nu G_O \ar[r]_{j_K'} & G_O^K1 \ar[dr]_{t} & \\
F_I S \ar[ur]^{F_I s} \ar[rrrrrr]_{h_{F_Is,t}} & & & & & & T
}
\]
All parts except ($\ast$) clearly commute either by definition or by naturality of $\inr\colon F\to F_I$ and $\outr\colon G_O\to G$. For ($\ast$), note that the lower path is the adjoint transpose of
\[ \nu G_O \xra{\gamma} G_O(\nu G_O) \xra{\outr} G(\nu G_O) \xra{G j_K'} GG_O^K1 \xra{Gt} GT \]
 the upper path is the adjoint transpose of 
\[ \nu G_O \xra{j_{K+1}'} G_O^{K+1}1 \xra{G_O t} G_OT \xra{\outr} GT, \]
and the commutative diagram below shows that these two morphisms are equal:
\[
\xymatrix{
\nu G_O \ar[r]^{j_{K+1}'} \ar[d]_{\gamma} & G_O^{K+1}1 \ar[r]^{G_Ot}  \ar[dd]^{\outr}  & G_O T \ar[dd]^{\outr} \\
G_O(\nu G_O) \ar[ur]_{G_Oj_K'} \ar[d]_{\outr} & & \\
G(\nu G_O) \ar[r]_{Gj_K'} & GG_O^K1  \ar[r]_{Gt} & GT 
}
\]  
This concludes the proof.
\end{proof}

\section*{Proof of \Cref{thm:termination}}
The proof of the correctness and termination of the generalized $\mathsf{L}^*$ algorithm requires some preparation. First, recall that for any endofunctor $H$, an $H$-coalgebra $C\xra{\gamma} HC$ is \emph{recursive} \cite{taylor} if for each $H$-algebra $HA\xra{\alpha} A$ there exists a unique coalgebra-to-algebra homomorphism $h$ from $(C,\gamma)$ into $(A,\alpha)$; that is, $h$ makes the square below commute. 
\[ 
\xymatrix{
C \ar@{-->}[r]^h \ar[d]_{\gamma} & A \\
HC \ar[r]_{Hh} & HA \ar[u]_\alpha 
}
\]
Dually, an $H$-algebra $HA\xra{\alpha} A$ is \emph{corecursive} if for each $H$-coalgebra $C\xra{\gamma} HC$ there exists a unique coalgebra-to-algebra homomorphism $h$ from $(C,\gamma)$ into $(A,\alpha)$.
\begin{lemma}[see \cite{cuv06}, Prop. 6]\label{lem:hc_corecursive}
For each recursive coalgebra $C\xra{\gamma}HC$, the coalgebra $HC\xra{H\gamma}HHC$ is also recursive.
\end{lemma}
Barlocco et al. \cite{bkr19} model prefix-closed sets as recursive subcoalgebras of an initial algebra $\mu H$. In our present setting, recursivity comes for free:

\begin{proposition}\label{prop:subalg_recursive}
Every subcoalgebra of $(F_I^N0,F_I^N\initial)$, $N\geq 0$, is recursive.
\end{proposition}
In particular, this result applies to the subcoalgebras $(S,\sigma)$ in the generalized $\mathsf{L}^*$ algorithm.

\begin{proof} Suppose that $s\colon (S,\sigma)\monoto (F_I^N0,F_I^N\initial)$ is a subcoalgebra for some $N\geq 0$. We prove that $(S,\sigma)$ is recursive by induction on $N$. 

\medskip\noindent
For $N=0$, note first that in any category $\D$ the initial object $0$ has no proper subobjects. (Indeed, suppose that $m\colon S\monoto 0$ is a subobject. Then the unique morphism $\initial_S\colon 0\to S$ satisfies $m\o \initial_S=\id_0$ by initiality of $0$, so $m$ is both monic and split epic, i.e.~an isomorphism.) Consequently, we have $(S,\sigma)=(0,\initial)$, and this coalgebra is trivially recursive by initiality of $0$.

\medskip\noindent For the induction step, let $N>0$, and let $(A,\alpha)$ be an arbitrary $F_I$-algebra. We need to prove that there is a unique coalgebra-to-algebra homomorphism $h\colon (S,\sigma)\to (A,\alpha)$. 

\begin{enumerate}\item \emph{Existence.} Since $(F_I^N0, F_I^N\initial)$ is a recursive coalgebra by \Cref{lem:hc_corecursive}, we have a unique coalgebra-to-algebra homomorphism $h'$ from $(F_I^N0, F_I^N\initial)$ to $(A,\alpha)$. Thus $h=h'\o s$ is a coalgebra-to-homomorphism from $(S,\sigma)$ to $(A,\alpha)$.

\item \emph{Uniqueness.} Suppose that $h\colon (S,\sigma)\to (A,\alpha)$ is a coalgebra-to-algebra homomorphism. Form the pullback of $s$ and $F_I^{N-1}\initial$:
\[  
\xymatrix@C+1em{
F_I^{N-1}0 \ar@{>->}[r]^{F_I^{N-1}\initial} & F_I^N0 \\
S' \ar@{>->}[u]^{s'} \ar@{>->}[r]_m & S \ar@{>->}[u]_s 
}
\]
Note that $F_I^{N-1}\initial\in \M$ because $\initial\colon 0\to F_I0=I$ lies in $\M$ by Assumption \ref{asm}\ref{A3} and $F_I$ preserves $\M$ by \Cref{asm2}. Since in any factorization system $(\E,\M)$ the class $\M$ is stable under pullbacks \cite[Prop. 14.15]{ahs}, it follows that $m,s'\in \M$. Since $F_I$ preserves pullbacks of $\M$-morphisms by \Cref{asm2}, the upper right square in the diagram below is a pullback, and the outer part commutes because $s$ is a coalgebra homomorphism. Thus, there is a unique morphism $n$ making the two triangles commute:
\[
\xymatrix{
& F_I^N0 \ar@{>->}[r]^{F_I^N\initial} & F_I^{N+1}0 \\
& F_IS' \ar@{>->}[u]^{F_Is'} \ar@{>->}[r]_{F_Im} & F_IS \ar@{>->}[u]_{F_Is} \\
S \ar@{>-->}[ur]^n \ar@/^2ex/@{>->}[uur]^s \ar@/_2ex/@{>->}[urr]_\sigma & &
}
\]
It follows that $m\colon (S',n\o m)\monoto (S,\sigma)$ and $s'\colon (S',n\o m) \monoto (F_I^{N-1}0,F_I^{N-1}\initial)$ are coalgebra homomorphisms, as shown by the two commutative diagrams below:
\[
\xymatrix{
S \ar[rr]^\sigma && F_IS \\
S' \ar@{>->}[r]_m \ar@{>->}[u]^m & S \ar@{>->}[r]_n \ar@{>->}[ur]^\sigma & F_IS' \ar@{>->}[u]_{F_Im}
}
\qquad
\xymatrix{
F_I^{N-1}0 \ar[rr]^{F_I^{N-1}\initial} && F_I^N0 \\
S' \ar@{>->}[r]_{m} \ar@{>->}[u]^{s'} & S \ar@{>->}[r]_n \ar@{>->}[ur]^s & F_IS' \ar@{>->}[u]_{F_Is'}
}
\]
\end{enumerate}
By induction we know that the coalgebra $(S',n\o m)$ is recursive, that is, we have a unique coalgebra-to-algebra homomorphism $g\colon (S',n\o m)\to (A,\alpha)$.
Since also $h\o m\colon (S',n\o m)\to (A,\alpha)$ is coalgebra-to-algebra homomorphism (being the composite of a coalgebra homomorphism with a coalgebra-to-algebra homomorphism), we get $h\o m = g$. Then the commutative diagram below shows that $h=\alpha\o F_Ig\o n$, i.e. $h$ is uniquely determined by $g$.
\begin{equation}\label{eq:fq}
\xymatrix{
& S \ar@/_2ex/[dl]_n \ar[r]^h \ar[d]_\sigma & A \\
F_IS' \ar@/_5ex/[rr]_{F_Ig} \ar[r]_{F_Im} & FS \ar[r]_{F_Ih}  & FA \ar[u]_\alpha 
}
\end{equation}
\end{proof}
Note that the proof of \Cref{prop:subalg_recursive} uses our assumption that $F_I$ preserves pullbacks im $\M$-morphisms. Since we do not require $G_O$ to preserve pushouts of $\E$-morphisms, the corresponding statement that every $G_O$-quotient algebra of $(G_O^K1, G_O^K\terminal)$ is corecursive does not hold. However, we have the following weaker result:

\begin{proposition}\label{prop:T_corecursive}
At each stage of Generalized $\mathsf{L}^*$, the algebra $(T,\tau)$ is corecursive.
\end{proposition} 

\begin{proof}
Recall that $(T,\tau)$ is a quotient algebra $t\colon (G_O^K1, G_O^K\terminal)\epito (T,\tau)$ for some $K>0$. We need to show that (1) $(T,\tau)$ is corecursive after its initialization in Step 0 of the algorithm, and that (2) every application of ``Extend $t$'' preserves corecursivity.

\medskip \noindent \emph{Proof of (1).} Initially, we have $(T,\tau) = (G_O1,G_O\terminal)$. Since the algebra $(1,\terminal)$ is trivially corecursive by terminality of $1$, the dual of \Cref{lem:hc_corecursive} shows that $(T,\tau)$ is corecursive.

\medskip\noindent \emph{Proof of (2).} Suppose that $(T,\tau)$ is corecursive. Applying ``Extend $t$''  replaces $(T,\tau)$ by the algebra $(T',t_0\o G_Ot_1)$, where $\tau=t_1\o t_0$. Then $t_0\colon (G_OT,G_O\tau)\to (T',t_0\o G_Ot_1)$ and $t_1\colon (T',t_0\o G_Ot_1)\to (T,\tau)$ are $G_O$-algebra homomorphisms, as shown by the diagram below.
\[
\xymatrix{
T & T' \ar[l]_{t_1} & G_OT \ar[l]_{t_0} \\ 
& G_O T \ar[u]_{t_0} & \\
G_OT \ar[uu]^{\tau} & G_O T' \ar[l]^{G_O t_1} \ar[u]_{G_O t_1} & G_OG_O T \ar[uu]_{G_O\tau} \ar[l]^{G_O t_0}
}
\]
To show that $(T',t_0\o G_Ot_1)$ is corecursive, let $(C,\gamma)$ be a $G_O$-coalgebra. We need to prove that there is a unique coalgebra-to-algebra homomorphism $h$ from $(C,\gamma)$ into $(T',t_0\o G_O t_1)$.

\medskip \noindent\emph{Existence.} Since $(T,\tau)$ is corecursive, the algebra $(G_OT,G_O\tau)$ is also corecursive by the dual of \Cref{lem:hc_corecursive}. Thus, there exists a unique coalgebra-to-algebra homomorphism $h'$ from $(C,\gamma)$ into $(G_OT,G_O\tau)$. It follows that $h=t_0\o h'$ is a coalgebra-to-algebra homomorphism from $(C,\gamma)$ into $(T',t_0\o G_O t_1)$, being the composite of the coalgebra-to-algebra homomorphism $h'$ with the algebra homomorphism $t_0$.

\medskip \noindent\emph{Uniqueness.} Let $h$ be a coalgebra-to-algebra homomorphism from $(C,\gamma)$ into $(T',t_0\o G_O t_1)$, and denote by $g$ the unique coalgebra-to-algebra homomorphism from $(C,\gamma)$ into the corecursive algebra $(T,\tau)$. Since also $t_1\o h$ is such a homomorphism (being the composite of a coalgebra-to-algebra homomorphism with an algebra homomorphism), we have $t_1\o h = g$. From the commutative diagram below it then follows that $h=t_0\o G_O g\o \gamma$, which shows that $h$ is uniquely determined by $g$. 
\[
\xymatrix{
T' & C \ar[l]_h \ar[dd]^\gamma \\
G_O T \ar[u]^{t_0} & \\
G_O T' \ar[u]^{G_Ot_1} & G_O C \ar[l]^{G_O h} \ar[ul]_{G_O g} 
}
\]
\end{proof}

\begin{lemma}\label{lem:est_vs_ehst}
Let $(s,t)$ be closed and consistent, and suppose that the algebra $(T,\tau)$ is corecursive. Then the associated hypothesis automaton $H_{s,t}$ (see \Cref{def:hypothesis}) is minimal. Moreover, the two diagrams below commute:
\[ 
\xymatrix{
S \ar@{>->}[r]^s \ar[dr]_{e_{s,t}} & F_I^N0 \ar[r]^{j_N} & \mu F_I \ar[dl]^{e_{H_{s,t}}} \\
& H_{s,t}  &  
}
\quad\quad
\xymatrix{
 & H_{s,t} \ar[dl]_{m_{s,t}} \ar[dr]^{m_{H_{s,t}}} &  \\
T & \ar[l]^{t} G_O^K1 & \nu G_O \ar[l]^{j_K'}
}
 \]
\end{lemma}
In particular, by \Cref{prop:T_corecursive}, this lemma applies to the pairs $(s,t)$ constructed in the generalized $\mathsf{L}^*$ algorithm.

\begin{proof}
\begin{enumerate}\item
We first prove that the left-hand diagram commutes. Consider the $F_I$-algebra structure on $H_{s,t}$ given by
\[ [i_{s,t},\delta_{s,t}]\colon F_IH_{s,t}\to H_{s,t}. \]
Then $e_{s,t}\colon (S,\sigma)\to (H_{s,t}, [i_{s,t},\delta_{s,t}])$ is a coalgebra-to-algebra homomorphism, as shown by the commutative diagram below:
\[
\xymatrix{
S \ar[rr]^{e_{s,t}} \ar[dd]_\sigma & & H_{s,t}  \\
& H_{F_is,t} \ar[ur]^{\cl_{s,t}^{-1}} & \\
F_I S \ar[ur]^{e_{F_Is,t}} \ar[rr]_{F_I e_{s,t}}  & & F_I H_{s,t} \ar[uu]_{[i_{s,t},\delta_{s,t}]} \\
}
\]
Indeed, the upper left part commutes by the definition of $\cl_{s,t}$, and the lower right part commutes by definition of $i_{s,t}$ and $\delta_{s,t}$ (consider the two coproduct components of $F_I S = I+FS$ separately).

Since also $e_{H_{s,t}}\o j_N\o s\colon (S,\sigma)\to (H_{s,t}, [i_{s,t},\delta_{s,t}])$ is a coalgebra-to-algebra homomorphism (being the composite of the $F_I$-coalgebra homomorphism $s$, the coalgebra-to-algebra homomorphism $j_N$ and the $F_I$-algebra homomorphism $e_{H_{s,t}}$) and the coalgebra $(S,\sigma)$ is recursive by \Cref{prop:subalg_recursive}, we conclude that $e_{s,t}= e_{H_{s,t}}\o j_N\o s$.

\item The proof that the right-hand diagram commutes is completely analogous: one views $H_{s,t}$ as a $G_O$-coalgebra 
\[ \langle f_{s,t}, \delta_{s,t}^@\rangle \colon H_{s,t}\to G_O H_{s,t}, \]
where $\delta_{s,t}^@\colon H_{s,t}\to GH_{s,t}$ denotes the adjoint transpose of $\delta_{s,t}\colon FH_{s,t}\to H_{s,t}$, and shows that both $m_{s,t}$ and $t\o j_K'\o m_{H_{s,t}}$ are coalgebra-to-algebra homomorphisms from $(H_{s,t}, \langle f_{s,t}, \delta_{s,t}^@\rangle)$ into the corecursive algebra $(T,\tau)$.

\item Since $e_{s,t}\in \E$ and $m_{s,t}\in \M$, it follows from the two commutative diagrams that $e_{H_{s,t}}\in \E$ and  $m_{H_{s,t}}\in \M$ (see \cite[Prop. 14.11]{ahs}). Thus, the automaton $H_{s,t}$ is minimal by \Cref{cor:minreachsimp}.\qedhere
\end{enumerate}
\end{proof}
An important invariant of the generalized $\mathsf{L}^*$ algorithm is that the subcoalgebra $s$ is pointed and that the quotient algebra $t$ is co-pointed:
\begin{definition}
An $F_I$-coalgebra $(R,\rho)$ is \emph{pointed} if there is a morphism $i_R$ such that the left-hand triangle below commutes. A $G_O$-algebra $(B,\beta)$ is \emph{co-pointed} if there is a morphism $f_R$ such that the right-hand triangle below commutes:
\[ 
\xymatrix{
I \ar[r]^{i_R} \ar[dr]_\inl & R \ar[d]^\rho \\
& F_I R
}
\qquad\qquad
 \xymatrix{
 	O  & B \ar[l]_{f_B}\\
 	& G_O B \ar[ul]^\outl \ar[u]_\beta
 }
 \]
\end{definition}
Note that if $(R,\rho)$ is a subcoalgebra of $(F_I^M0,F_I^M\initial)$, then $i_R$ is necessarily unique because $F_I^M\initial$ is monic by \Cref{asm}\ref{A3} and \Cref{asm2}. Dually for co-pointed quotient algebras of $(G_O^M0,G_O^M!)$.

\begin{lemma}\label{lem:pointedsubcoalg}
At each stage of the generalized $\mathsf{L}^*$ algorithm, the coalgebra $(S,\sigma)$ is pointed and the algebra $(T,\tau)$ is co-pointed.
\end{lemma}

\begin{proof}
We proceed by induction on the number of steps of the algorithm required to construct the pair $(s,t)$. Initially, after Step (0), $(S,\sigma)$ is equal to  $(I,F_I\initial)$, and thus pointed via $i_S=\id_I$. 
\[
\xymatrix{
I \ar[r]^\id \ar[dr]_{\inl} & I \ar[d]^{F_I\initial=\inl}  \\
& F_I I
}
\]
Dually, $(T,\tau)$ is co-pointed via $f_T=\id_O$.

Now suppose that at some stage of the algorithm, $(S,\sigma)$ is pointed and $(T,\tau)$ is co-pointed. We need to show that $(S,\sigma)$ remains pointed after executing ``Extend $s$'' or adding a counterexample to $s$, and that $(T,\tau)$ remains co-pointed after executing ``Extend $t$''.
\begin{enumerate}
\item \emph{Extend $s$}. When calling ``Extend $s$'', the coalgebra $(S,\sigma)$ is replaced by the coalgebra $(S',F_Is_0\o s_1)$. This coalgebra is pointed via $i_{S'}=s_0\o i_S$, as witnessed by the commutative diagram below:
\[
\xymatrix{
I \ar[r]^{i_S} \ar@/^2em/[rr]^{i_{S'}} \ar@/_1em/[drr]_\inl \ar@/_2em/[ddrr]_\inl & S \ar[r]^{s_0} \ar[dr]^{\sigma} & S' \ar[d]^{s_1} \\
& & F_IS \ar[d]^{F_Is_0} \\
& & F_IS'
}
\]
\item \emph{Extend $t$.} Symmetric to (1).
\item \emph{Adding a counterexample.} Let $(C,\gamma)$ be the counterexample added to $(S,\sigma)$, and denote by $i\colon (S,\sigma)\monoto (S\vee C, \sigma\vee \gamma)$ the embedding. Then the coalgebra $(S\vee C,\sigma\vee \gamma)$ is pointed via $i_{S\vee C}=i\o i_S$, as shown by the commutative diagram below:
\[
\xymatrix{
I \ar@/_4em/[drr]_\inl \ar[r]^{i_S} \ar@/^2em/[rr]^{i_{S\vee C}} \ar[dr]_\inl & S \ar[d]^\sigma \ar[r]^i & S\vee C \ar[d]^{\sigma\vee \gamma} \\
& F_IS \ar[r]_{F_I i} & F_I(S\vee C)
}
\]
\end{enumerate}

\end{proof}

%
%

\begin{lemma}\label{lem:pointed_initial_final}
Let $A$ be an automaton. For any pointed subcoalgebra $r\colon (R,\rho)\monoto (F_I^M0, F_I^M\initial)$, we have
\[  
i_A \;=\; (\, I\xra{i_R} R \xra{r} F_I^M 0 \xra{j_M} \mu F_I \xra{e_A} A \,)
\]
Dually, for any co-pointed quotient algebra $b\colon (G_O^M1,G_O^M!)\epito (B,\beta)$, we have 
\[ f_A \;=\; (\, A \xra{m_A} \nu G_O \xra{j_M'} G_O^M1 \xra{b} B \xra{f_B} O \,). \] 
\end{lemma}

\begin{proof}
The first statement follows from the commutative diagram below, all of whose parts either commute trivially or by definition.
\[
\xymatrix@C-0.5em{
I \ar[rrrrr]^{i_A} \ar[ddd]_{i_R} \ar[ddr]^\inl \ar[ddrrr]^\inl \ar@/^2ex/[ddrrrr]^\inl \ar[drrrr]^\inl  & & & & & A \\
& & & & F_I A \ar[ur]_{\alpha_A} & \\
& F_I R \ar[rr]^{F_Ir} & & F_I^{M+1}0 \ar[r]^{F_I j_M} \ar[drr]_{j_{M+1}} & F_I(\mu F_I) \ar[u]_{F_I e_A} \ar[dr]^\alpha & \\
R \ar[rr]_{r} \ar[ur]_\rho & &  F_I^M0 \ar[rrr]_{j_M} \ar[ur]_{F_I^M\initial} & & & \mu F_I \ar[uuu]_{e_A}
}
\]
The proof of the second statement is dual.
\end{proof}

\begin{proposition}\label{prop:hypothesis_correct_for_st}
Let $(s,t)$ be a closed and consistent pair as in \eqref{eq:st}, and suppose that $t$ is co-pointed. Then the hypothesis $H=H_{s,t}$ and the unknown automaton $Q$ have the same observation tables for $(s,t)$:
\[ h_{s,t}^H = h_{s,t}^{Q}. \]
In particular, $H$ and $Q$ agree on inputs from $S$, that is,
\[ L_H\o j_N\o s = L_Q\o j_N\o s. \]
\end{proposition}

\begin{proof}
\begin{enumerate}
\item For the first equality, consider the following diagram:
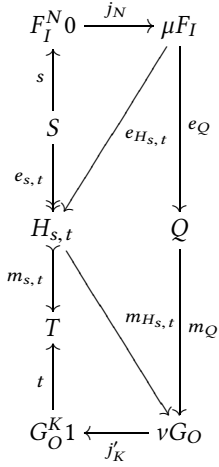
\[ 
\xymatrix{
F_I^N0 \ar[r]^{j_N} & \mu F_I \ar[dd]^{e_Q} \ar[ddl]^{e_{H_{s,t}}} \\
S \ar[u]^s \ar@{->>}[d]_{e_{s,t}} & \\
H_{s,t} \ar@{>->}[d]_{m_{s,t}} \ar[ddr]^{m_{H_{s,t}}} & Q  \ar[dd]^{m_Q} \\
T  & \\
G_O^K1 \ar[u]^{t} & \nu G_O \ar[l]^{j_K'}
}
 \]
The outward commutes by definition of $h_{s,t}$ and since $h_{s,t}=m_{s,t}\o e_{s,t}$. The upper left and lower left parts commute by \Cref{lem:est_vs_ehst}. It follows that the remaining part commutes when precomposed with $j_N\o s$ and postcomposed with $t\o j_K'$, which gives $h_{s,t}^H = h_{s,t}^Q$.
\item The second equality follows by postcomposing both sides of the equality  $h_{s,t}^H = h_{s,t}^{Q}$ with $f_T\colon T\to O$ and applying \Cref{lem:pointed_initial_final}.\qedhere
\end{enumerate}
\end{proof}
The key to the termination of the learning algorithm lies is in the following result.
\begin{lemma}\label{lem:not_closed_consistent}
Let $(s,t)$ be a closed and consistent pair as in \eqref{eq:st}, and suppose that $t$ is co-pointed. Then for every counterexample $c$ for $H_{s,t}$, the pair $(s\vee c,t)$ is not closed or not consistent.
\end{lemma}

\begin{proof}
Suppose for the contrary that the pair $(s\vee c, t)$ is closed and consistent. Denote by 
\[i\colon S\monoto S\vee C\quad\text{and}\quad i'\colon C\to S\vee C\]
 the two embeddings, satisfying $(s\vee c)\o i = s$ and $(s\vee c)\o i'= c$. Via diagonal fill-in we obtain a unique $j\colon H_{s,t}\monoto H_{s\vee c, t}$ such that the following diagram commutes:
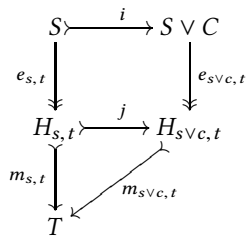
\[
\xymatrix{
S \ar@{>->}[r]^i \ar@{->>}[d]_{e_{s,t}} & S\vee C \ar@{->>}[d]^{e_{s\vee c, t}} \\
H_{s,t} \ar@{>->}[r]^j \ar@{>->}[d]_{m_{s,t}} & H_{s\vee c, t} \ar@{>->}[dl]^{m_{s\vee c, t}} \\
T &
}
\]
We shall show below that $j$ is an automata homomorphism. In particular, $H_{s,t}$ and $H_{s\vee c,t}$ accept the same language by \Cref{prop:hom_pres_language}. Letting $H=H_{s\vee c, t}$, we compute
\begin{flalign*}
& L_{H_{s,t}}\o j_N\o c &  \\
&= L_{H}\o j_N \o c & \text{since $L_{H_{s,t}}=L_H$} \\
&= f_H\o e_{H} \o j_N \o c & \text{def. $L_H$} \\
&= f_T\o t\o j_K' \o m_H \o e_H \o j_N \o c & \text{by \Cref{lem:pointed_initial_final}} \\
&= f_T\o t\o j_K' \o m_H \o e_H \o j_N \o (s\vee c) \o i' & \text{def. $i'$} \\
&= f_T \o h_{s\vee c,t}^H\o i' & \text{def. $h_{s\vee c,t}^H$} \\
&= f_T\o h_{s\vee c,t}^Q\o i' & \text{by Prop.~\ref{prop:hypothesis_correct_for_st}} \\
&=~~~ \cdots & \\
&= L_Q\o j_N\o c & \text{compute backwards}
\end{flalign*}  
This contradicts the fact that $c$ is a counterexample for $H_{s,t}$. 

\medskip \noindent
To conclude the proof, it only remains to verify our above claim that $j$ is an automata homomorphism.
\begin{enumerate}
\item \emph{$j$ preserves transitions.} Observe first that we have
\begin{equation}\label{eq:s_vs_sc} m_{s,t}\o l_{s,t} = m_{s\vee c, t}\o l_{s\vee c,t}\o Fi, \end{equation}
as shown by the commutative diagram below:
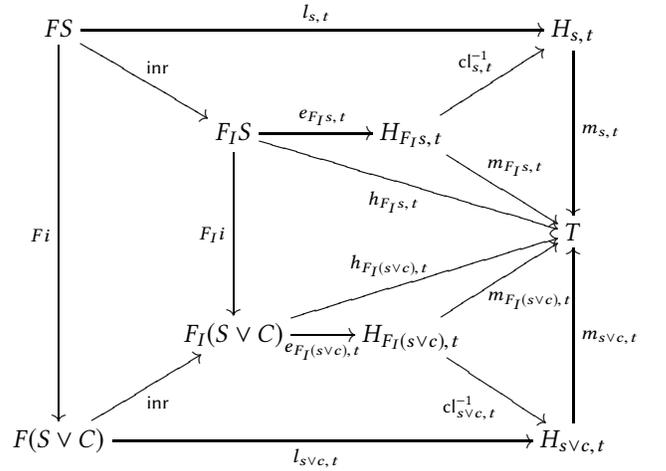
\[
\xymatrix{
FS \ar[dddd]_{Fi} \ar[rrr]^{l_{s,t}} \ar[dr]^\inr & & & H_{s,t} \ar[dd]^{m_{s,t}} \\
& F_I S \ar[r]^{e_{F_Is,t}} \ar[drr]_{h_{F_Is,t}}  \ar[dd]_{F_Ii} & H_{F_Is,t} \ar[dr]^{m_{F_Is,t}} \ar[ur]^{\cl_{s,t}^{-1}} & \\
& & & T \\
& F_I(S\vee C) \ar[urr]^{h_{F_I(s\vee c),t}} \ar[r]_{e_{F_I(s\vee c),t}} & H_{F_I(s\vee c),t} \ar[ur]_{~~m_{F_I(s\vee c),t}} \ar[dr]_{\cl_{s\vee c,t}^{-1}} & \\
F(S\vee C) \ar[ur]_\inr \ar[rrr]_{l_{s\vee c,t}}  & & & H_{s\vee c, t} \ar[uu]_{m_{s\vee c,t}}
} 
\]
Here the left-hand part commutes by naturality of $\inr$, the central triangle commutes by definition of $h_{\dash,t}$ (using that $(s\vee c)\o i = s$), and all remaining parts commute by definition.

Now, consider the following diagram:
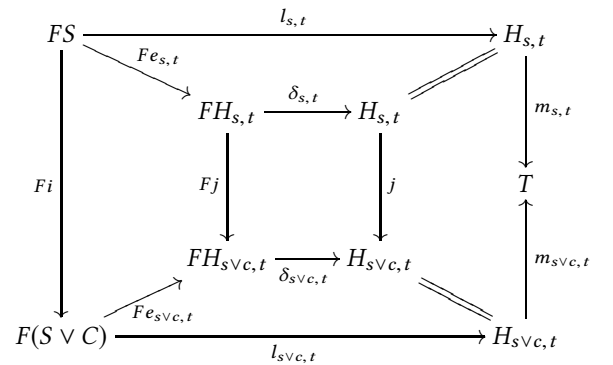
\[
\xymatrix@R-1em{
FS \ar[rrr]^{l_{s,t}} \ar[dddd]_{Fi} \ar[dr]^{Fe_{s,t}} & & & H_{s,t} \ar[dd]^{m_{s,t}} \\
& FH_{s,t} \ar[r]^{\delta_{s,t}} \ar[dd]_{Fj} & H_{s,t} \ar[dd]^j \ar@{=}[ur] & \\
& & & T \\
& FH_{s\vee c,t} \ar[r]_{\delta_{s\vee c,t}} & H_{s\vee c,t}  \ar@{=}[dr] & \\
F(S\vee C) \ar[ur]_{Fe_{s\vee c,t}} \ar[rrr]_{l_{s\vee c,t}}  & & & H_{s\vee c,t} \ar[uu]_{m_{s\vee c,t}} 
}
\]
The outward commutes by \eqref{eq:s_vs_sc}, and all parts except the central square commute by definition. It follows that also the central square commutes, because it commutes when precomposed with the epimorphism $Fe_{s,t}$ and postcomposed with the monomorphism $m_{s\vee c,t}$. Thus, $j$ preserves transitions. 
\item \emph{$j$ preserves the initial state.} Observe first that we have 
\begin{equation}\label{eq:s_vs_sc_init} 
m_{s,t}\o i_{s,t} = m_{s\vee c,t}\o i_{s\vee c,t},
\end{equation}
as shown by the commutative diagram below:
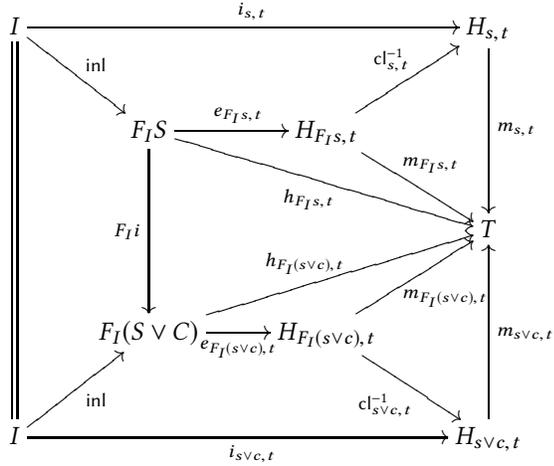
\[
\xymatrix{
I \ar@{=}[dddd] \ar[rrr]^{i_{s,t}} \ar[dr]^\inl & & & H_{s,t} \ar[dd]^{m_{s,t}} \\
& F_I S \ar[r]^{e_{F_Is,t}} \ar[drr]_{h_{F_Is,t}}  \ar[dd]_{F_Ii} & H_{F_Is,t} \ar[dr]^{m_{F_Is,t}} \ar[ur]^{\cl_{s,t}^{-1}} & \\
& & & T \\
& F_I(S\vee C) \ar[urr]^{h_{F_I(s\vee c),t}} \ar[r]_{e_{F_I(s\vee c),t}} & H_{F_I(s\vee c),t} \ar[ur]_{~~m_{F_I(s\vee c),t}} \ar[dr]_{\cl_{s\vee c,t}^{-1}} & \\
I \ar[ur]_\inl \ar[rrr]_{i_{s\vee c,t}}  & & & H_{s\vee c, t} \ar[uu]_{m_{s\vee c,t}}
} 
\]
Now consider the following diagram:
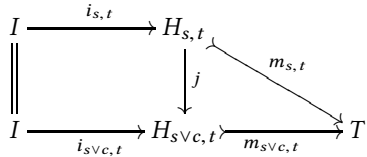
\[
\xymatrix@C+2em{
I \ar@{=}[d] \ar[r]^{i_{s,t}} & H_{s,t} \ar@{>->}[dr]^{m_{s,t}} \ar[d]^j & \\
I \ar[r]_{i_{s\vee c,t}} & H_{s\vee c,t} \ar@{>->}[r]_{m_{s\vee c,t}} & T
}
\]
The outward commutes by \eqref{eq:s_vs_sc_init}, and the right-hand triangle by the definition of $j$. Thus the left-hand part commutes, since it does when postcomposed with the monomorphism $m_{s\vee c,t}$. This proves that $j$ preserves the initial state.
\item \emph{$j$ preserves final states.} The proof is analogous to (2).\qedhere
\end{enumerate}
\end{proof}
With the above results at hand, we are ready to prove \Cref{thm:termination}:

 \begin{proof}[Proof of \Cref{thm:termination}]
The algorithm only terminates if a hypothesis $H_{s,t}$ constructed in Step (2) is correct (i.e.
it accepts the same language as the unknown automaton $Q$), in which case $H_{s,t}$ is returned. This automaton is minimal by \Cref{lem:est_vs_ehst}, so $H_{s,t}=\Min{L_Q}$.

Thus, we only need to verify that the algorithm eventually finds a correct hypothesis. For any $F_I$-subcoalgebra $r\colon (R,\rho)\monoto (F_I^M0,F_I^M\initial)$, let $e_r$ and $m_r$ denote the $(\E,\M)$-factorizations of $e_Q\o j_M\o r$. 
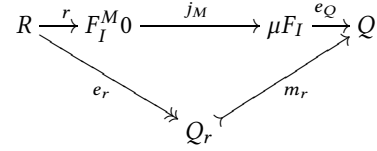
\[  
\xymatrix@C-1em{
	R \ar[r]^r \ar@{->>}[drr]_{e_r} & F_I^M0 \ar[rr]^{j_M}  && \mu F_I \ar[r]^{e_Q} & Q\\
	&& Q_r \ar@{>->}[urr]_{m_r} &&
}
\]
Similarly, for any $G_O$-quotient algebra $b\colon (G_O^M1,G_O^M!)\epito (B,\beta)$, let $\ol e_b$ and $\ol m_b$ be the $(\E,\M)$-factorization of $b\o j_M'\o m_Q$.
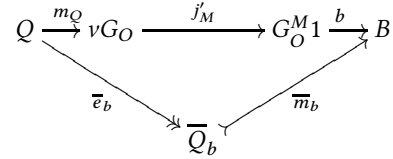
\[
\xymatrix@C-1em{
Q \ar[r]^{m_Q} \ar@{->>}[drr]_{\ol e_b} & \nu G_O \ar[rr]^{j_M'}  && G_O^M1 \ar[r]^{b} & B\\
 && \ol Q_b \ar@{>->}[urr]_{\ol m_b} &&
}
\]
Let $(s,t)$ and $(s',t')$ be two consecutive pairs appearing in an execution of the algorithm. We show below that the following statements hold:
\begin{enumerate}
\item If $(s',t')$ emerges from $(s,t)$ via ``Extend $s$'', then $m_s<m_{s'}$ and $\ol e_{t}=\ol e_{t'}$.
\item If $(s',t')$ emerges from $(s,t)$ via ``Extend $t$'', then $m_s=m_{s'}$ and  $\ol{e}_t< \ol{e}_{t'}$.
\item If $(s',t')$ emerges from $(s,t)$ by adding a counterexample, then $m_{s}\leq m_{s'}$ and $\ol e_t=\ol e_{t'}$
\end{enumerate} 
Letting $(s^0,t^0), (s^1,t^1), (s^2,t^2),\ldots$ denote the sequence of pairs constructed in an execution of the algorithm, it follows that we obtain two ascending chains
\[m_{s^0}\leq m_{s^1} \leq m_{s^2} \leq \cdots  \quad\text{and}\quad \ol e_{s^0}\leq \ol e_{s^1}\leq \ol e_{s^2} \leq \cdots.\]
of subobjects and quotients of $Q$, respectively. By our assumption that $Q$ is Noetherian, both chains must stabilize, i.e. all but finitely many of the relations $\leq$ are equalities. By (1) and (2), this implies that ``Extend $s$'' and ``Extend $t$'' are called only finitely often. Moreover, whenever a counterexample is added to $s$, this must be immediately followed by a call of ``Extend $s$'' oder ``Extend $t$'' by \Cref{lem:not_closed_consistent}. Thus also Step (2b) is executed only finitely often. This proves that the algorithm necessarily terminates.

\medskip\noindent It remains to establish the above statements (1)--(3). 

\begin{enumerate}
\item An application of ``Extend $s$'' to $(s,t)$ yields the new pair $(s',t')$ with
\[ s'=F_Is \o s_1 \quad\text{and}\quad t'=t.  \]
Thus, we trivially have $\ol e_t = \ol e_{t'}$. Moreover, $m_{s}\leq m_{s'}$ holds by the right-hand triangle in the diagram below, where the morphism $n_{s,s'}$ is obtained via diagonal fill-in:
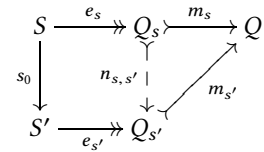
\[
\xymatrix{
S \ar[d]_{s_0} \ar@{->>}[r]^{e_s} & Q_s \ar@{>->}[r]^{m_s} \ar@{>-->}[d]_{n_{s,s'}}  & Q \\
S' \ar@{->>}[r]_{e_{s'}}  & Q_{s'} \ar@{>->}[ur]_{m_{s'}} &  
}
\]
To prove  $m_s<m_{s'}$, we need to show that $n_{s,s'}$ is not an isomorphism. To this end, consider the unique morphisms $d_s$ and $d_{s'}$ (defined via diagonal fill-in) such that the diagrams below commute:
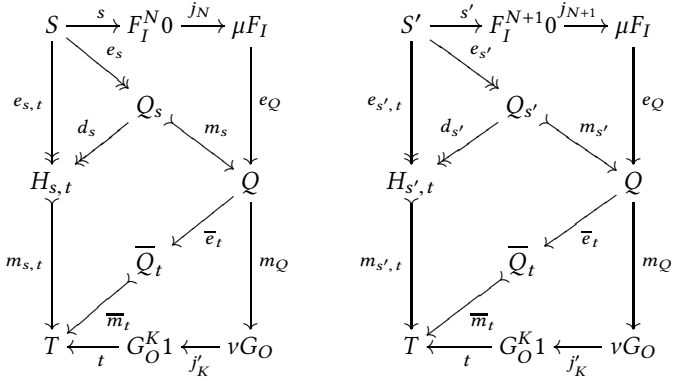
\[
\xymatrix@C-1em@R-1em{
S \ar[r]^s \ar@{->>}[dd]_{e_{s,t}} \ar@{->>}[dr]^{e_s} & F_I^N0 \ar[r]^{j_N} & \mu F_I \ar[dd]^{e_Q} \\
 & Q_s \ar@{>->}[dr]^{m_s} \ar@{->>}[dl]_{d_s} &  \\
H_{s,t} \ar@{>->}[dd]_{m_{s,t}} & & Q \ar[dl]^{\ol e_t} \ar[dd]^{m_Q} \\
 & \ol Q_t \ar@{>->}[dl]^{\ol m_t} & \\ 
T & G_O^K1  \ar[l]^{t} & \nu G_O \ar[l]^{j_K'}
}
\quad \quad
\xymatrix@C-1em@R-1em{
S' \ar[r]^{s'} \ar@{->>}[dd]_{e_{s',t}} \ar@{->>}[dr]^{e_{s'}} & F_I^{N+1}0 \ar[r]^{j_{N+1}} & \mu F_I \ar[dd]^{e_Q} \\
 & Q_{s'} \ar@{>->}[dr]^{m_{s'}} \ar@{->>}[dl]_{d_{s'}} &  \\
H_{s',t} \ar@{>->}[dd]_{m_{s',t}} & & Q \ar[dl]^{\ol e_t} \ar[dd]^{m_Q} \\
 & \ol Q_t \ar@{>->}[dl]^{\ol m_t} & \\ 
T & G_O^K1  \ar[l]^{t} & \nu G_O \ar[l]^{j_K'}
}
\]
Moreover, observe that we have the following commutative diagram:
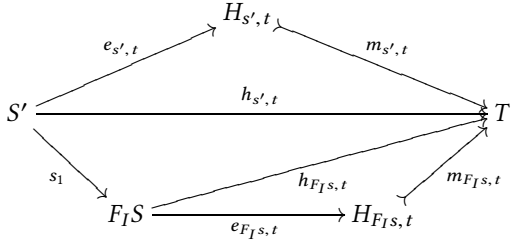
\[
\xymatrix{
& & H_{s',t} \ar@{>->}[drr]^{m_{s',t}} & & \\
S' \ar[rrrr]^{h_{s',t}} \ar[urr]^{e_{s',t}} \ar[dr]_{s_1}  & & & & T \\
& F_I S \ar[urrr]_{h_{F_Is,t}}  \ar[rr]_{e_{F_Is,t}}  &  &H_{F_Is,t} \ar@{>->}[ur]_{m_{F_Is,t}} &  
}
\]
By the choice of $s_1$ in ``Extend $s$'', we have $e_{F_Is,t}\o s_1\in \E$. The uniqueness of $(\E,\M)$-factorizations thus implies that, up to isomorphism,
\[ H_{s',t} = H_{F_Is,t},\quad e_{s',t}=e_{F_Is,t}\o s_1,\quad m_{s',t}=m_{F_Is,t}.  \]
We now claim that the following diagram commutes:
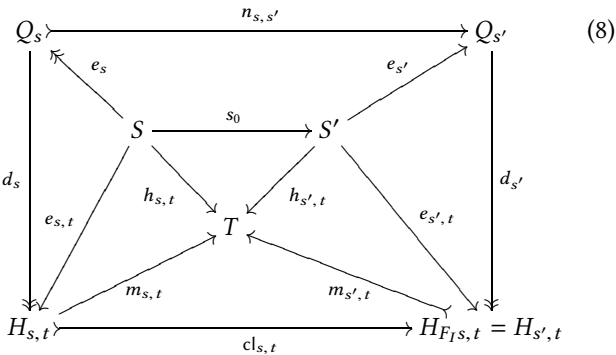
\begin{equation}\label{eq:nss}
\xymatrix{
Q_s \ar@{>->}[rrrr]^{n_{s,s'}} \ar@{->>}[ddd]_{d_s} & & & & Q_{s'} \ar@{->>}[ddd]^{d_{s'}} \\
& S \ar[ddl]_{e_{s,t}} \ar[dr]_{h_{s,t}} \ar@{->>}[ul]_{e_s} \ar[rr]^{s_0} &    & S' \ar[ddr]^{e_{s',t}} \ar[ur]^{e_{s'}} \ar[dl]^{h_{s',t}} & \\
& & T & & \\
H_{s,t} \ar[urr]_{m_{s,t}} \ar@{>->}[rrrr]_{\cl_{s,t}} & & & & H_{F_Is,t}=H_{s',t}  \ar@{>->}[ull]^{m_{s',t}} 
}
\end{equation}
All inner parts commute by definition. Thus also the outward commutes, since it does when precomposed with the epimorphism $e_s$ and postcomposed with the monomorphism $m_{s',t}$. 

We are ready to prove our claim that $n_{s,s'}$ is not an isomorphism. Suppose for the contrary that it is. Since $d_{s'}\in \E$, the diagram \eqref{eq:nss} yields $\cl_{s,t}\o d_s = d_{s'}\o n_{s,s'}\in \E$. Thus $\cl_{s,t}\in \E$. One the other hand, by definition of $\cl_{s,t}$ we have $m_{F_Is,t}\o \cl_{s,t} = m_{s,t}\in \M$ and thus $\cl_{s,t}\in \M$. But from $\cl_{s,t}\in \E\cap \M$ it follows that that $\cl_{s,t}$ is an isomorphism \cite[Prop. 14.6]{ahs}, contradicting the fact that the input pair $(s,t)$ of ``Extend $s$'' is not closed.
\item The proof is symmetric to (1).
\item Adding a counterexample $c$ means to to replace the pair $(s,t)$ by the pair $(s',t')$ with
\[ s'=s\vee c \quad\text{and}\quad t'=t.\]
Thus $\ol e_t = \ol e_{t'}$. Letting $i\colon (S,\sigma)\monoto (S\vee C, \sigma\vee \gamma)=(S',\sigma')$ denote the embedding with $s=(s\vee c)\o i$, diagonal fill-in yields a morphism $n_{s,s'}$ making the diagram below commute: 
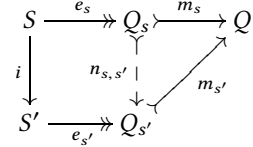
\[
\xymatrix{
S \ar[d]_{i} \ar@{->>}[r]^{e_s} & Q_s \ar@{>->}[r]^{m_s} \ar@{>-->}[d]_{n_{s,s'}}  & Q \\
S' \ar@{->>}[r]_{e_{s'}}  & Q_{s'} \ar@{>->}[ur]_{m_{s'}} &  
}
\]
This proves that $m_s \leq m_{s'}$.\qedhere
\end{enumerate}
\end{proof}

\section*{Details for \Cref{rem:complexity}}
Let $m$ and $n$ be the height (i.e. the length of the longest strictly ascending chain) of the poset of subobjects and quotients of $Q$, respectively. The proof of \Cref{thm:termination} shows that
\begin{enumerate}
\item ``Extend $s$'' is executed at most $m$ times;
\item ``Extend $t$'' is executed at most $n$ times;
\item Step (2b) is executed at most $m+n$ times.  
\end{enumerate}
Thus, Steps (1a), (1b) and (2b) are executed at most $2m+2n = O(m+n)$ times.

\section*{Details for \Cref{ex:noetherian}}
\begin{enumerate}
\item The statements for $\D=\Set, \Pos, \Vect{\K}$ are clear.
\item $\D=\JSL$: clearly every finite semilattice is Noetherian. Conversely, if $Q$ is a infinite semilattice, choose a sequence 
 \[q_0,q_1,q_2,\ldots\] of elements of $Q$ such that $q_{n+1}$ is not an element of the subsemilattice $\langle q_0,\ldots,q_n\rangle$ of $Q$ generated by $q_0,\ldots, q_n$. Since this subsemilattice is finite (of cardinality at most $2^{n+1}$), such a $q_{n+1}$ can always be chosen. Then
\[ \langle q_0 \rangle \monoto \langle q_0,q_1\rangle \monoto \langle q_0,q_1,q_2\rangle \monoto \ldots \]
is an infinite strictly ascending chain of subsemilattices of $Q$, showing that $Q$ is not Noetherian.
\item $\D=\Nom$: We show that orbit-finite sets have the claimed polynomial
height. Let~$X$ be an orbit-finite nominal set with~$n$ orbits. It is
clear that chains of subobjects, i.e.\ equivariant subsets, of~$X$
have length at most $n$. It remains to show the polynomial bound on
chains of quotients. The number of orbits decreases non-strictly along
such a chain, and can strictly decrease at most~$n$ times, so it
suffices to consider chains of quotients that retain the same number
of orbits. Such quotients are sums of quotients of single-orbit sets,
so it suffices to consider the case where~$X$ has only one
orbit. Then, all elements of~$X$ have supports of the same size~$k$;
since this number decreases non-strictly along a chain of quotients,
and can strictly decrease at most~$k$ times, it suffices to consider
chains of quotients that retain the same support size.

We now use the standard fact that~$X$ is a quotient of $\At^{*k}$, the
$k$-fold separated product of~$\At$; the same, of course, holds for
all quotients of~$X$. A quotient of $\At^{*k}$ whose elements retain
supports of size~$k$ is determined by a subgroup~$G$ of the symmetric
group~$S_k$. (Specifically, the quotient determined by~$G$ identifies
$(a_1,\dots,a_k)$ and $(a_{\pi(1)},\dots,a_{\pi(k)})$ for all
$(a_1,\dots,a_k)\in\At^{*k}$ and $\pi\in G$. Conversely, from a given
quotient~$e:X\epito Y$, we obtain~$G$ as consisting of all $\pi\in S_k$
such that~$e$ identifies $(a_1,\dots,a_k)$ and
$(a_{\pi(1)},\dots,a_{\pi(k)})$ for all
$(a_1,\dots,a_k)\in\At^{*k}$.) The given chain of quotients thus
corresponds to a chain of subgroups of~$S_k$, which for $k\ge 2$ has
length at most $2k-3$~\cite{Babai86}.
\end{enumerate}

\section*{Details for \Cref{rem:coalglogic}}
We demonstrate that the coalgebraic learning algorithm in~\cite{bkr19}
gets stuck when applied to the setting of $\Sigma$-automata in $\Nom$. In
the following, we assume some familiarity with the algorithm and the
notation introduced in \emph{op.~cit.}

A coalgebraic logic giving the semantics of nominal automata can be described in complete analogy to the $\Set$ case \cite[Example 1]{bkr19}. We instantiate the logical framework to
\[
\xymatrix{
\Nom^\op \ar@{=>}[dr]|\delta \ar[d]_{L^\op} & \ar[l]_<<<<<P \Nom \ar[d]^B \\
\Nom^\op & \Nom \ar[l]^<<<<<{P}
}
\]
where
\[ LX=1+\At \times X,\qquad BX = 2\times [\At,X],\qquad P=[\dash,2].\]
The right adjoint of $P$ is $Q=[\dash,2]\colon \Nom^{op}\to \Nom$. 
For each $X\in \Nom$, the map \[\delta_X\colon 1+\At\times [X,2] \to [2\times[\At,X],2]\] sends the unique element of $1$ to the left product projection, and $(a,f)\in \At\times [X,2]$ to $\delta_X(a,f)\in [2\times [\At,X],2]$ with
\[  
\delta_X(a,f)(b,g)=f(g(a))\quad \text{for $b\in 2$, $g\in [\At,X]$}.
\]
We have the initial algebra for $L$ given by $\Phi = \mu L = \At^*$, and the theory map 
\[ {\mathop{th}}^\gamma\colon X\to Q\Phi = [\At^*,2] \]
for a nominal automaton (i.e.\ $B$-coalgebra)~$X$ is just the unique coalgebra homomorphism from $X$ into the final coalgebra $\nu B = [\At^*,2]$ (cf. \Cref{ex:automata}). 

Now consider the nominal language $K\colon \At^*\to 2$ with $K(w)=1$ iff $w$ has even length. We assume that the unknown coalgebra is given by
\[ (\,X\xra{\gamma} BX\,)\quad=\quad (\,\At^* \xra{\langle K,\gamma'\rangle} 2\times [\At,\At^*]\,) \]
with $\gamma'(w)(a)=wa$ for $w\in \At^*$, $a\in \At$. (The state set~$X$ is effectively made known to the learner in advance since the learning algorithm computes subobjects of $X$. Thus, in the typical scenario~$X$ will be orbit-infinite like in the present example, although of course the language~$K$ can be accepted by an orbit-finite automaton.) The algorithm starts with the trivial observation table 
\[ S=\{\epsilon\} \monoto X \quad\text{and}\quad \Psi = \emptyset\monoto \Phi,\]
This table is closed and the induced conjecture is the trivial one-state automaton accepting all words in $\At^*$. Since $a\not\in K$ for $a\in \At$, the teacher provides the (minimal) counterexample $\{\epsilon\}+\At\monoto \Phi$. After adding it to $\Psi$, the new table is
\[ S=\{\epsilon\} \monoto X \quad\text{and}\quad \Psi = \{\epsilon\}+\At\monoto \Phi.\]
The next reachability step computes the set $\Gamma(S)$ of elements of $X$ reachable from $S=\{\epsilon\}$ in a single transition step:
\[ \Gamma(S)=\At. \]
Thus 
\[ S\vee \Gamma(S) = S\cup \Gamma(S) = \{\epsilon\}+\At\monoto X. \]
Viewing the elements of $Q\Psi=[\{\epsilon\}+\At,2]$ as finitely supported subsets of $\Psi=\{\epsilon\}+\At$, we can describe the map
\[ S\vee \Gamma(S) \monoto X \xra{\mathop{th}^\gamma} Q\Psi  \]
as sending $\epsilon$ to $\{\epsilon\}\seq \Psi$ and every $a\in \At$ to $\At\seq \Psi$, i.e. the image of this map is the discrete nominal set
\[ \overline{S} = \{ \{\epsilon\},\, \At \}\cong 2. \]
In order to close the table, Step 6 of the algorithm now requires to choose~a monomorphism $\ol{S}\monoto X$ subject to certain conditions. But clearly there exists no monomorphism from $\ol{S}=2$ to $X=\At^*$ in $\Nom$, i.e.~the algorithm cannot make the required choice.

\section*{Details for \Cref{ex:unpres}}
Our categorical notion of automata presentation involves quotients of $\MT$-algebras. For practical purposes, it is sometimes more convenient to work with the equivalent concept of a congruence:
\begin{remark}\label{rem:quotcong}
\begin{enumerate}
\item Recall that for a monad $\MT$ on $\Set$ given by a finitary signature $\Gamma$ and equations $E$ between $\Gamma$-terms, quotient algebras of a $\MT$-algebra (i.e.~$(\Gamma,E)$-algebra) $A$ correspond bijectively to {congruences} on $A$. Here a \emph{congruence} is an equivalence relation $\equiv$ on $A$ respecting all $\Gamma$-operations: for all $a,a'\in A$ with $a\equiv a'$, one has
\[ \gamma(a_1,\ldots, a_{i-1}, a, a_{i+1}\ldots, a_n) \equiv \gamma(a_1,\ldots, a_{i-1},a',a_{i+1},\ldots,a_n) \]
for $n>0$, $\gamma\in \Gamma_n$, $i\in \{1,\ldots,n\}$ and $a_j\in A$  ($j\neq i$). The bijection identifies a quotient $e\colon A\epito B$ with its kernel, i.e.~ the congruence given by
\[ a\equiv a' \quad\Lra\quad e(a)=e(a'). \]
Thus, if the object $TI$ is equipped with some $\Sigma$-automata structure $\Sigma\times T I \xra{\delta} T I$, the equivalence in \Cref{def:autpres}(3) states precisely that an equivalence relation $\equiv$ on $T I$ corresponding to a $\MT$-refinable quotient is a congruence on $\MT I$ iff for all $w,w'\in T I$ and $a\in \Sigma$, 
\[ w\equiv w' \quad\text{implies}\quad \delta(a,w)\equiv\delta(a,w'). \]
\item An analogous remark applies to monads $\MT$ on $\Set^S$ corresponding to a finitary $S$-sorted signature $\Gamma$ and equations between $\Gamma$-terms: quotient algebras of a $(\Gamma,E)$-algebra $A$ correspond to $S$-sorted congruence relations, i.e. families of equivalence relations $\equiv=(\equiv_s\seq A_s\times A_s)_{s\in S}$ respecting all operations. Thus, if $T I$ is equipped with the structure of a sorted $\Sigma$-automaton $\delta_{s,t}\colon \Sigma_{s,t}\times (T I)_s\to (T I)_t$ ($s,t\in S$), the equivalence in \Cref{def:autpres}(3) states precisely that an $S$-sorted equivalence relation $\equiv$ on $T I$ corresponding to a $\MT$-refinable quotient is a congruence on $\MT I$ iff for all $w,w'\in (T I)_s$ and $a\in \Sigma_{s,t}$, 
\[ w\equiv w' \quad\text{implies}\quad \delta_{s,t}(a,w)\equiv\delta_{s,t}(a,w'). \]
\end{enumerate}
\end{remark}
We will now describe automata presentations for semigroups, Wilke algebras, and general (ordered) $(\Gamma,E)$-algebras, including stabilization algebras. We will see that in all these cases, the equivalence in \Cref{def:autpres}(3) holds for arbitrary, not only $\MT$-refinable, quotients.

\paragraph{Semigroups.} The free semigroup $T_+I = I^+$ has a $\Sigma$-automata presentation $\delta\colon \Sigma\times I^+\to I^+$ given by the alphabet \[\Sigma = \{ \vec{a}\;:\; a\in I \} \cup \{ \vecr{a} \;:\; a\in I \}\] and the transitions
\[\delta(\vec{a},w)=wa\quad\text{and}\quad  \delta(\vecr{a},w)=aw \quad \text{for}\quad w\in I^+,\, a\in I.\]
We show that (1)--(3) of \Cref{def:autpres} (with $F=\Sigma\times \dash$ on $\Set$) are satisfied. (1) is clear by \Cref{rem:fexamples}. For (2), recall from \Cref{ex:automata_sorted} that $\mu F_I = I\times \Sigma^*$. The unique homomorphism $e_{I^+}\colon I\times \Sigma^*\to I^+$ interprets a word in $I\times \Sigma^*$ as a list of instructions for forming a word in $I^+$, e.g. \[e_{I^+}(a\vec{a}\vec{b}\vecr{b}\vec{a}) \;=\; baaba.\] 
Thus, $e_{I^+}$ is surjective: given $a_1\ldots a_n\in I^+$ with $a_i\in I$, we have 
\[ a_1\ldots a_n = e_{I^+}(a_1\vec{a_2}\cdots \vec{a_n})  \]
To show (3), we use \Cref{rem:quotcong}(1): we need to verify that an equivalence relation $\equiv$ on $I^+$ is a monoid congruence iff, for every $w,w'\in I^+$ and $a\in I$,  
\[ w\equiv w' \quad\text{implies}\quad wa\equiv w'a,\; aw\equiv aw'. \]
The ``only if'' direction is clear. For the ``if'' direction, let $w\equiv w'$ and $v\in I^+$; we need to show that $wv\equiv w'v$ and $vw\equiv vw'$. For the first equivalence, let $v=a_1\ldots a_n$. Then we get the chain of implications
\[ w\equiv w' \;\To\; wa_1\equiv w'a_1 \;\To\; \ldots \;\To \; wa_1\ldots a_n \equiv wa_1\ldots a_n, \]
i.e.  $wv\equiv w'v$. The proof of the second equivalence is symmetric.

\paragraph{Wilke algebras.}  The
free  Wilke algebra $T_\infty(I,\emptyset)=(I^+,I^{\mathsf{up}})$ can be
  presented as a two-sorted $\Sigma$-automaton with the sorted
  alphabet
  $\Sigma=(\Sigma_{+,+},\,\Sigma_{+,\omega},\,\Sigma_{\omega,\omega},\emptyset)$
   given by
  \begin{align*}
  	\Sigma_{+,+} &= \{ \vec{a}: a\in I \} \cup \{ \vecr{a} : a\in I \}\\
  \Sigma_{+,\omega}&=\{\omega\}\cup \{ \vec{v}^\omega : v\in I^+ \}\\
  \Sigma_{\omega,\omega}&=\{ \vecru{a}: a\in I \}
  \end{align*}
  and the transitions below, where  $v,w\in I^+$,
  $z\in I^{\mathsf{up}}$, $a\in I$:
  \begin{align*}
  	\delta_{+,+}(\vec{a},w)&=wa, & \delta_{+,+}(\vecr{a},w)&=aw,\\
  \delta_{+,\omega}(\omega,w)&=w^\omega,&
  \delta_{+,\omega}(\vec{v}^\omega,w)&=wv^\omega,\\
  \delta_{\omega,\omega}(\vecru{a},z)&= az. & &
  \end{align*}
We show that (1)--(3) of \Cref{def:autpres} (with $F$ the functor on $\Set^{\{+,\omega\}}$ from \Cref{ex:automata_sorted}) are satisfied. (1) is clear by \Cref{rem:fexamples}. For (2), recall from \Cref{ex:automata_sorted} that the initial algebra $\mu F_I$ consists of sorted words over $\Sigma$ with an additional first letter from $I$. The homomorphism $e_{(I^+,I^{\mathsf{up}})}\colon \mu F_I\to (I^+,I^{\mathsf{up}})$ views such a word as an instruction for forming a word in $(I^+,I^{\mathsf{up}})$, e.g.
\[ e_{(I^+,I^{\mathsf{up}})}(a\vec{b}\vec{a}\omega\vecru{a}\vecru{a})\;=\;aa(aba)^\omega. \] 
Thus $e_{(I^+,I^{\mathsf{up}})}$ is surjective: every finite word $w\in I^+$ is in the image of $e_{(I^+,I^{\mathsf{up}})}$ as in the case of semigroups, and for an ultimately periodic word $(a_1\ldots a_n)(b_1\ldots b_m)^\omega\in I^{\mathsf{up}}$ we have
\[  (a_1\ldots,a_n)(b_1\ldots b_m)^\omega =  e_{(I^+,I^{\mathsf{up}})}(b_1\vec{b_2}\cdots \vec{b_m}\omega \vecr{a_n}\cdots \vecr{a_1}). \]
To show (3), we use \Cref{rem:quotcong}(2): we need to verify that a two-sorted equivalence relation $\equiv$ on $(I^+, I^{\mathsf{up}})$ is a congruence w.r.t. the Wilke algebra structure iff, for each $w,w',v\in I^+$ with $w\equiv w'$ and $a\in I$, one has
\[ aw\equiv aw',\; wa\equiv w'a,\; w^\omega \equiv (w')^\omega,\; wv^\omega \equiv w'v^\omega, \]
and for each $z,z'\in I^{\mathsf{up}}$ with $z\equiv z'$ and $a\in I$ one has $az\equiv az'$. The ``only if'' direction is clear. For the ``if'' direction, we need to show that for all $v,w,w'\in I^+$ and $z,z'\in I^{\mathsf{up}}$,
\begin{itemize}
\item $w\equiv w'$ implies $vw\equiv vw'$,  $wv\equiv w'v$, $w^\omega\equiv (w')^\omega$ and $wz\equiv w'z$;
\item $z\equiv z'$ implies $wz\equiv wz'$.
\end{itemize}
Let us show that  $w\equiv w'$ implies $wz\equiv w'z$; the proofs of the other statements are similar. We have $z=a_1\ldots a_ny^\omega$ with $a_1,\ldots, a_n\in I$ and $y\in I^+$. From $w\equiv w'$ it follows that
\[ wa_1\equiv w'a_1,\; wa_1a_2\equiv w'a_1a_2,\;\cdots,\; wa_1\ldots a_n\equiv w'a_1\ldots a_n, \]
and thus  \[wz=wa_1\ldots a_ny^\omega \equiv w'a_1\ldots a_ny^\omega=w'z.\]

\paragraph{Stabilization algebras.} Suppose that $\MT$ is a monad on $\Set$ or $\Pos$ induced by a finitary signature $\Gamma$ and (in-)equations $E$; see \Cref{sec:preliminaries}. Then $\MT I$ can be presented as the $\Gamma$-automaton $\delta\colon F_{\Gamma}(TI) \to TI$ given by the $\Gamma$-algebra structure on the free $(\Gamma,E)$-algebra $TI$. We show that (1)--(3) of \Cref{def:autpres} are satisfied.

(1) is clear by \Cref{rem:fexamples}. For (2), observe that the initial algebra $\mu (F_\Gamma)_I$ is the algebra $T_\Gamma I$ of $\Gamma$-terms over $I$, and that the unique homomorphism $e_{TI}\colon T_\Gamma I\epito TI$ interprets $\Gamma$-terms in $TI$. Since the $\MT$-algebra $\MT I$ is generated by the set $I$ as a $\Gamma$-algebra, every element of $\MT I$ can be expressed as a $\Gamma$-term over $I$, i.e.~$e_{TI}$ is surjective. (3) is clear: the equivalence just amounts to the statement that if $e$ is a surjective homomorphism of (ordered) $\Gamma$-algebras and its domain satisfies all (in-)equations in $E$, then so does its codomain. 

By instantiating to the monad $\MT=\MT_S$ on $\Pos$, we see that the free stabilization algebra $\MT_S I$ has a $\Gamma$-automata presentation for the signature $\Gamma$ of \Cref{ex:reclang}(3).

\section*{Proof of \Cref{thm:ltolinl}}
	Suppose that $L$ is recognized via $e\colon \MT I \to (A,\alpha)$ and $p\colon A\to O$, where $(A,\alpha)$ is a finite $\MT$-algebra. We may assume that $e\in \E$. (Otherwise consider the $(\E,\M)$-factorization \[\xymatrix{ \MT I \ar@{->>}[r]^{e'} & (A',\alpha') \ar@{>->}[r]^m & (A,\alpha) }\] of $e$. Since $\D_f$ is closed under subobjects, $L$ is recognized by the finite $\MT$-algebra $(A',\alpha')$ via $e'$ and $p\o m$, i.e. we can replace $e$ by $e'$.)
	
	Since $(F,\delta)$ forms a weak automata presentation, the object $A$ can be equipped with an $F$-algebra structure $\delta_A\colon FA\to A$ such that $e\colon (TI,\delta)\epito (A,\delta_A)$ is an $F$-algebra homomorphism. Equipping $TI$ and $A$ with the initial states $\eta_{I}\colon I\to TI$ and $e\o \eta_I\colon I\to A$, respectively, we can view $TI$ and $A$ as $F_I$-algebras and $e$ as an $F_I$-algebra homomorphism. By initiality of $\mu F_I$, it follows that $e_A = e\o e_{TI}$. It follows that the diagram below commutes, which proves that the automaton $(A,\delta_A,e\o \eta_I,p)$ accepts the language $\lin{L}=L\o e_{TI}$. 
	\begin{equation}\label{eq:eti}
	\xymatrix{
		\mu F_I \ar@{->>}[r]^{e_{TI}} \ar[dr]_{e_A} & TI \ar[d]_e \ar[r]^L & O \\
		& A \ar[ur]_p & 
	}
	\end{equation}
Since $A$ is finite, we conclude that $\lin{L}$ is regular. \qed

\section*{Proof of \Cref{thm:synalgconst}}
The proof is illustrated by the diagram below:
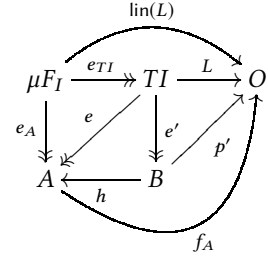
\[
\xymatrix{
\mu F_I \ar@{->>}[r]^{e_{TI}} \ar@/^2em/[rr]^{\lin{L}} \ar@{->>}[d]_{e_A} & TI \ar@{->>}[dl]_e \ar@{->>}[d]^{e'} \ar[r]^L & O \\
A \ar@/_4em/[rru]_{f_A} & B \ar[l]^h \ar[ur]_{p'} &
}
\]
	Let $A=\Min{\lin{L}}$ be the minimal automaton for the language $\lin{L}$. Equipping $TI$ with the initial states $\eta_I\colon I\to TI$ and the final states $L\colon TI\to O$, we can view $TI$ as an automaton accepting $\lin{L}=L\o e_{TI}$. Since $e_{TI}\in \E$ (that is, the automaton $TI$ reachable) and $A$ is minimal, there exists a unique automata homomorphism $e\colon TI\epito A$. We now prove the theorem by establishing the following claims:

\medskip\noindent \textbf{Claim 1.} For every finite quotient $\MT$-algebra $e'\colon \MT I \epito (B,\beta)$ that recognizes $L$, there exists a unique $h\colon B\to A$ with $e=h\o e'$. 

\medskip \noindent \emph{Proof.} As in the proof of \Cref{thm:ltolinl}, $B$ can be viewed as a reachable automaton recognizing $\lin{L}$. By minimality of $A$, there is an automata homomorphism $h\colon B\to A$. We have
		\[ h\o e' \o e_{TI} = e\o e_{TI}\]
		because both sides are $F_I$-algebra homomorphisms from $\mu F_I$ to $B$ and $\mu F_I$ is initial. Thus $h\o e'=e$ because $e_{TI}$ is epic.

\medskip\noindent \textbf{Claim 2.} The automaton $A$ can be equipped with $\MT$-algebra structure $(A,\alpha_A)$ such that $e\colon \MT I \epito (A,\alpha_A)$ is a $\MT$-homomorphism.

\medskip\noindent \emph{Proof.} Since $L$ is $\MT$-recognizable, we have $L=p'\o e'$ for some finite quotient $\MT$-algebra $e'\colon \MT\epito (B,\beta)$ and some $p'\colon A\to O$. By Claim 1, $e=h\o e'$ for some $h$, which shows that $e$ is $\MT$-refinable. Since $(F,\delta)$ is an automata presentation, we obtain the desired $\alpha_A$.

\medskip\noindent\textbf{Claim 3.} $e\colon \MT I\epito (A,\alpha_A)$ is a syntactic $\MT$-algebra for $L$. 

\medskip\noindent \emph{Proof.} The homomorphism $e$ recognizes $L$ via $f_A$: we have \[ L\o e_{TI} = \lin{L} = f_A\o e_A = f_A\o e\o e_{TI}\]
and thus $L=f_A\o e$ because $e_{TI}$ is epic. The universal  property of $e$ follows from Claim 1. \qed

\end{document}